\documentclass{article}

\pdfoutput=1

\usepackage[utf8]{inputenc}
\usepackage[T1]{fontenc}
\usepackage{amsmath}
\usepackage{tikz}
\usepackage{url}
\usetikzlibrary{calc, intersections, fpu}
\usetikzlibrary{arrows.meta} %cool arrows
\usetikzlibrary{decorations.text}
\usepackage{tikz-cd}
\usepackage{amsfonts}
\usepackage{amssymb}
\usepackage{amsthm}
\usepackage{graphicx}
\usepackage{stmaryrd}
\usepackage{comment}
\usepackage{algorithm}
\usepackage{imakeidx}
\usepackage[noend]{algpseudocode}
\usepackage{listings}
\usepackage{enumitem}
\usepackage{mathtools}
\usepackage{dsfont}
\usepackage{xcolor}
\usepackage{geometry}
\usepackage[colorlinks=true, linkcolor=blue, hyperindex=true, citecolor = green!60!purple]{hyperref}
\usepackage{pgfmath}
\usepackage{standalone}
\usepackage{nameref}
\usepackage{mathtools}

\newcommand\ens[1]{\left\{ #1 \right\}}
\newcommand\oper[4]{#1 \limits_{\substack{#2}}^{#3}{#4}}
\renewcommand{\H}{\mathrm H}
\newcommand{\N}{\mathds N}
\newcommand{\Z}{\mathds{Z}}
\newcommand{\R}{\mathds{R}}
\newcommand{\B}{\mathds{B}}

\newcommand{\Sp}{\mathds S}

\newcommand{\C}{\mathds C}

\newcommand{\lk}{\text{Lk}}

\newcommand{\control}[2]{ ($#2+#1$) .. #1}
\definecolor{coldef}{rgb}{0.1,0.1,0.9}
\newcommand{\emphdef}[1]{\textcolor{coldef}{#1}}

\newtheorem{theorem}{Theorem}[section]
\newtheorem{definition}[theorem]{Definition}
\newtheorem{proposition}[theorem]{Proposition}

\newtheorem{remark}[theorem]{Remark}
\newtheorem{lemma}[theorem]{Lemma}

\title{Hopf Arborescent Links, Minor Theory, and Decidability of the Genus Defect} 

\author{Pierre Dehornoy\thanks{Aix-Marseille Université, CNRS, I2M, Marseille, France, pierre.dehornoy@univ-amu.fr} \and \stepcounter{footnote}Corentin Lunel\thanks{LIGM, CNRS, Univ. Gustave Eiffel, ESIEE Paris, F-77454 Marne-la-Vall\'ee, France, corentin.lunel2@univ-eiffel.fr} \and Arnaud de Mesmay\thanks{LIGM, CNRS, Univ. Gustave Eiffel, ESIEE Paris, F-77454 Marne-la-Vall\'ee, France, arnaud.de-mesmay@univ-eiffel.fr}}
\date{}

\definecolor{dred}{rgb}{0.6,0,0}

\definecolor{dblue}{rgb}{0.3,0.3,0.8}

\definecolor{dgreen}{rgb}{0.0,0.5,0.0}

\begin{document}

\maketitle

\begin{abstract}
While the problem of computing the genus of a knot is now fairly well understood, no algorithm is known for its four-dimensional variants, both in the smooth and in the topological locally flat category. 
In this article, we investigate a class of knots and links called Hopf arborescent links, which are obtained as the boundaries of some iterated plumbings of Hopf bands. 
We show that for such links, computing the genus defects, which measure how much the four-dimensional genera differ from the classical genus, is decidable. 
Our proof is non-constructive, and is obtained by proving that Seifert surfaces of Hopf arborescent links under a relation of minors defined by containment of their Seifert surfaces form a well-quasi-order.
\end{abstract}

\section{Introduction}
A (tame) \emphdef{knot} is a polygonal embedding of the circle $\Sp^1$ into $\R^3$, or equivalently, $\Sp^3$, and a \emphdef{link} is a disjoint union of knots. Knot theory is both an old and very active mathematical field, yet from an algorithmic perspective, many problems arising naturally in knot theory are still shrouded in mystery. This can be illustrated with arguably the most fundamental algorithmic question in knot theory: in the \textsc{Knot Equivalence} problem, we are given two knots $K_1$ and $K_2$ and are tasked with deciding whether they are equivalent, that is, whether one can deform one into the other continuously without creating self-intersections. The best algorithm for this problem, due to Kuperberg, is elementary recursive~\cite{Kuperberg_elem_rec}, yet the problem is not even known to be \textbf{NP}-hard (see for example~\cite[Conclusion]{lackenby2017some}). We refer to Lackenby~\cite{Lackenby_Algorithms} for a survey on algorithmic problems in knot theory.

Given how seemingly hard testing the equivalence of knots is, a huge body of research has been devoted to designing and studying knot invariants in order to tell them apart. A classical invariant of a knot is its \emphdef{genus}: this is the smallest possible genus of an embedded oriented surface, called a \emphdef{Seifert surface}, having the knot as its boundary. Computing the genus of a knot turns out to be significantly more tractable: celebrated works of Hass, Lagarias and Pippenger~\cite{Hass_trivial_NP} and Agol, Hass and Thurston~\cite{aht}, building on the normal surface theory of Haken~\cite{haken}, have shown that deciding if a knot has genus at most $g$ is in \textbf{NP}, while Lackenby has proved that it is also in co-\textbf{NP}~\cite{lackenby2021efficient}. These algorithms run also well in practice within the software Regina~\cite{regina}.

There are, however, different notions of genus that are much less understood: considering~$\Sp^3$ as the boundary of the $4$-dimensional ball~$\B^4$, the $4$-genus of a knot $g_4(K)$ is roughly the smallest possible genus of a surface in $\B^4$ having the knot as its boundary. 
This comes in two flavours that are known to not be equivalent: the \emphdef{topological locally flat 4-genus} and the \emphdef{smooth 4-genus} depending on the regularity of the surface. 
We refer to the preliminaries for precise definitions.
A knot is (topologically or smoothly) \emphdef{slice} if it bounds a disc in $\B^4$, i.e., has $4$-genus zero. 
One of the motivations for the study of such $4$-dimensional invariants comes from algebraic geometry, as such surfaces arise naturally around singularities of algebraic curves in $\mathbb{C}^2$~\cite{kronheimer1994genus,rudolph1993}.
Another motivation is the slice-ribbon conjecture~\cite{fox1962} which states that a knot is smoothly slice if and only if it is ribbon, i.e., it bounds an immersed disc with only ribbon-type singularities in~$\Sp^3$. 
Unfortunately, no algorithmic framework at all is known to attack topological problems in 4-dimensional topology, and indeed many of these problems are known to be undecidable, e.g., the homeomorphism of $4$-manifolds~\cite{markov1958insolubility}. For some other problems, the decidability is a well-known open problem: this is the case for $4$-sphere recognition~\cite{weinberger2002homology} or embeddability of $2$-dimensional complexes in $\R^4$~\cite{matouvsek2010hardness}.
Similarly, no algorithm is known to decide the $4$-genus of a knot or even to decide whether it is slice. To illustrate how hard these problems are, it is only in a recent breakthrough of Picirillo~\cite{piccirillo} that it was proved that the Conway knot is not smoothly slice, although it only has $11$ crossings. From the perspective of lower bounds, recent work of de Mesmay, Rieck, Sedgwick and Tancer~\cite{unbearable} has proved that an analogue of the $4$-genus for links, the $4$-ball Euler characteristic, is \textbf{NP}-hard to compute, but it is also not known to be decidable.

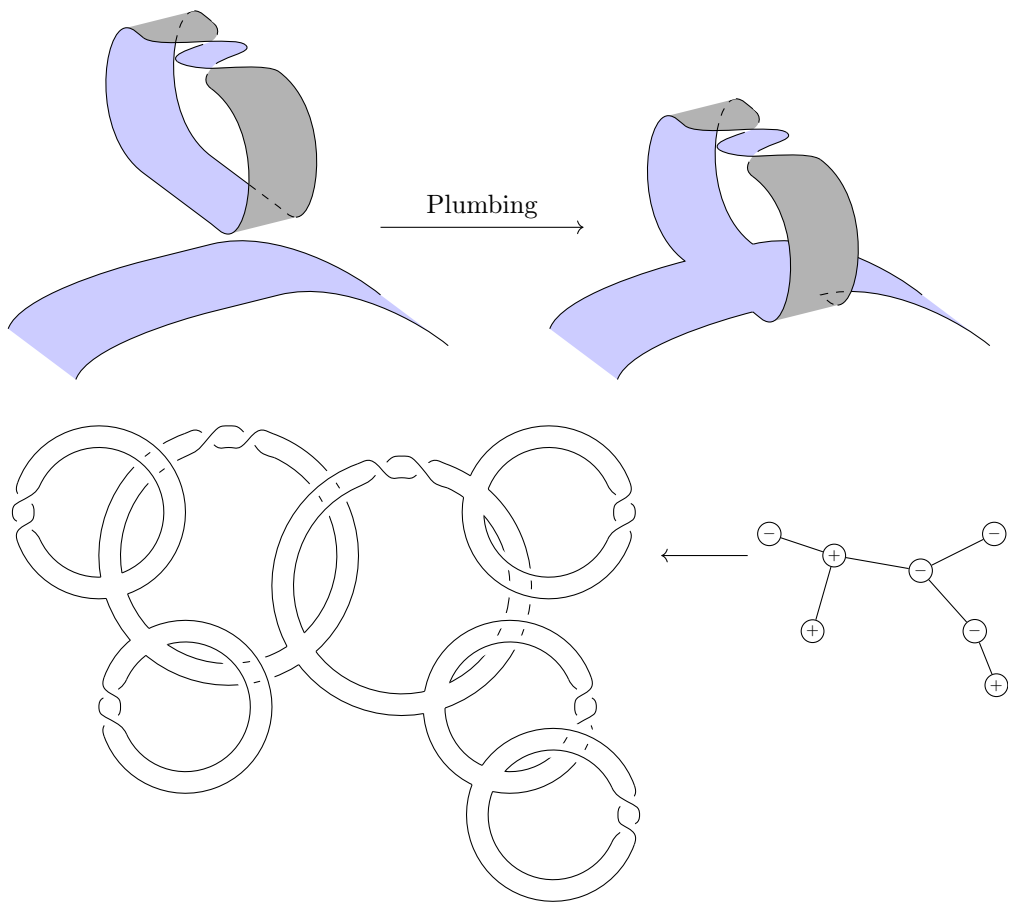
\begin{figure}[ht]
\begin{center}
\begin{tikzpicture}[scale = 0.9]
%Band filling
\fill [blue!20!white] (-2,-1) .. controls +(70:0.5) and \control{(0,0)}{(-166:0.5)} -- (1,0.25) .. controls +(14:1) and \control{(3.5,-0.5)}{(140:0.5)} -- (4.5,-1.25) .. controls +(140:0.5) and \control{(2,-0.5)}{(14:1)} -- (1,-0.75) .. controls +(-166:0.5) and \control{(-1,-1.75)}{(70:0.5)} -- cycle;
\draw (-2,-1) .. controls +(70:0.5) and \control{(0,0)}{(-166:0.5)} -- (1,0.25) .. controls +(14:1) and \control{(3.5,-0.5)}{(140:0.5)}; 
\draw [xshift = 1cm, yshift = -0.75cm] (-2,-1) .. controls +(70:0.5) and \control{(0,0)}{(-166:0.5)} -- (1,0.25) .. controls +(14:1) and \control{(3.5,-0.5)}{(140:0.5)};

\begin{scope}[yshift = 1.3cm]
%Coordinate
\coordinate (h) at (1,1.95);
\coordinate (l) at (1, 1.55); 
\coordinate (i) at (0.5, 1.65);
\coordinate (e) at (1.5, 1.9);

%Filling
\begin{scope}
\clip (h) .. controls +(180:0.2) and \control{(0,2)}{(-36.75:0.2)} .. controls +(143.25:0.1) and \control{(-0.25,2.15)}{(0:0.1)} .. controls +(180:0.35) and \control{(0,0)}{(143.25:1)} -- (1,-0.75) .. controls +(-36.75:0.1) and \control{(1.25,-0.9)}{(180:0.1)} .. controls +(0:0.35) and \control{(1,1.25)}{(-36.75:1)} -- cycle;
\fill [blue!20!white] (1,2.25) .. controls +(143.25:0.1) and \control{(0.75,2.4)}{(0:0.1)} .. controls +(180:0.35) and \control{(1,0.25)}{(143.25:1)} -- (2,-0.5) .. controls +(-36.75:0.1) and \control{(2.25,-0.65)}{(180:0.1)} -- ++(0, -1) -- ++(-3.5,0) -- ++(0,4.5) -- ++(2.5,0) -- cycle;
\end{scope}
\fill [blue!20!white] (h) .. controls +(-150:0.1) and \control{(i)}{(143.25:0.2)} .. controls +(-36.75:0.1) and \control{(l)}{(180:0.2)} .. controls +(60:0.2) and \control{(e)}{(-36.75:0.2)} .. controls +(143.25:0.1) and \control{(h)}{(0:0.2)};
\fill [white!40!gray] (-0.25,2.15) .. controls +(0:0.1) and \control{(0,2)}{(143.25:0.1)} .. controls +(-36.75:0.2) and \control{(h)}{(180:0.2)} .. controls +(30:0.1) and \control{(1,2.25)}{(-36.75:0.2)} .. controls +(143.25:0.1) and \control{(0.75,2.4)}{(0:0.1)}-- cycle;
\fill [white!40!gray] (1,1.25) .. controls +(143.25:0.1) and \control{(l)}{(-130:0.2)} .. controls +(180:0.2) and \control{(2,1.5)}{(143.25:0.2)} .. controls +(-36.75:1) and \control{(2.25,-0.65)}{(0:0.35)} -- (1.25,-0.9) .. controls +(0:0.35) and \control{(1,1.25)}{(-36.75:1)};

%Crossings
\begin{scope}
\clip (-0.25,2.15) .. controls +(0:0.1) and \control{(0,2)}{(143.25:0.1)} .. controls +(-36.75:0.2) and \control{(h)}{(180:0.2)} .. controls +(30:0.1) and \control{(1,2.25)}{(-36.75:0.2)} .. controls +(143.25:0.1) and \control{(0.75,2.4)}{(0:0.1)}-- cycle;
\draw [dashed] (0.75,2.4) .. controls +(180:0.35) and \control{(1,0.25)}{(143.25:1)};
\end{scope}
\begin{scope}
\clip (1,1.25) .. controls +(143.25:0.1) and \control{(l)}{(-130:0.2)} .. controls +(180:0.2) and \control{(2,1.5)}{(143.25:0.2)} .. controls +(-36.75:1) and \control{(2.25,-0.65)}{(0:0.35)} -- (1.25,-0.9) .. controls +(0:0.35) and \control{(1,1.25)}{(-36.75:1)};
\draw [dashed] (1,0.25) -- (2,-0.5) .. controls +(-36.75:0.1) and \control{(2.25,-0.65)}{(180:0.1)};
\end{scope}
\begin{scope}
\clip (h) .. controls +(180:0.2) and \control{(0,2)}{(-36.75:0.2)} .. controls +(143.25:0.1) and \control{(-0.25,2.15)}{(0:0.1)} .. controls +(180:0.35) and \control{(0,0)}{(143.25:1)} -- (1,-0.75) .. controls +(-36.75:0.1) and \control{(1.25,-0.9)}{(180:0.1)} .. controls +(0:0.35) and \control{(1,1.25)}{(-36.75:1)} -- cycle;
\draw (0.75,2.4) .. controls +(180:0.35) and \control{(1,0.25)}{(143.25:1)} -- (2,-0.5);
\end{scope}
\draw (1,2.25) .. controls +(143.25:0.1) and \control{(0.75,2.4)}{(0:0.1)};
\draw (2.25,-0.65) .. controls +(0:0.35) and \control{(2,1.5)}{(-36.75:1)};
\draw (0,2) .. controls +(143.25:0.1) and \control{(-0.25,2.15)}{(0:0.1)} .. controls +(180:0.35) and \control{(0,0)}{(143.25:1)} -- (1,-0.75) .. controls +(-36.75:0.1) and \control{(1.25,-0.9)}{(180:0.1)} .. controls +(0:0.35) and \control{(1,1.25)}{(-36.75:1)} ;
\begin{scope}[even odd rule] 
\clip (0, 2.25) rectangle (3, 1) (l) circle (0.1cm);
\draw (1,1.25) .. controls +(143.25:0.1) and \control{(l)}{(-130:0.2)} .. controls +(60:0.2) and \control{(e)}{(-36.75:0.2)};
\end{scope}
\draw (e) .. controls +(143.25:0.1) and \control{(h)}{(0:0.2)} .. controls +(180:0.2) and \control{(0,2)}{(-36.75:0.2)};
\begin{scope}[even odd rule] 
\clip (0,2) -- (1,1.25) -- (2,1.5) -- (1,2.25) -- (0,2) (h) circle (0.13cm);
\draw (1,2.25) .. controls +(-36.75:0.2) and \control{(h)}{(30:0.1)} .. controls +(-150:0.1) and \control{(i)}{(143.25:0.2)};
\end{scope}
\draw (i) .. controls +(-36.75:0.1) and \control{(l)}{(180:0.2)} .. controls +(0:0.2) and \control{(2,1.5)}{(143.25:0.2)};
\end{scope}

\draw [->] (3.5, 0.5) -- (6.5, 0.5) node [midway, above] {Plumbing};

\begin{scope}[xshift = 8cm]
%Band filling
\fill [blue!20!white] (-2,-1) .. controls +(70:0.5) and \control{(0,0)}{(-166:0.5)} -- (1,0.25) .. controls +(14:1) and \control{(3.5,-0.5)}{(140:0.5)} -- (4.5,-1.25) .. controls +(140:0.5) and \control{(2,-0.5)}{(14:1)} -- (1,-0.75) .. controls +(-166:0.5) and \control{(-1,-1.75)}{(70:0.5)} -- cycle;
\draw (-2,-1) .. controls +(70:0.5) and \control{(0,0)}{(-166:0.5)} -- (1,0.25) .. controls +(14:1) and \control{(3.5,-0.5)}{(140:0.5)}; 
\draw [xshift = 1cm, yshift = -0.75cm] (-2,-1) .. controls +(70:0.5) and \control{(0,0)}{(-166:0.5)} -- (1,0.25) .. controls +(14:1) and \control{(3.5,-0.5)}{(140:0.5)};

%Coordinate
\coordinate (h) at (1,1.95);
\coordinate (l) at (1, 1.55); 
\coordinate (i) at (0.5, 1.65);
\coordinate (e) at (1.5, 1.9);

%Filling
\begin{scope}
\clip (h) .. controls +(180:0.2) and \control{(0,2)}{(-36.75:0.2)} .. controls +(143.25:0.1) and \control{(-0.25,2.15)}{(0:0.1)} .. controls +(180:0.35) and \control{(0,0)}{(143.25:1)} -- (1,-0.75) .. controls +(-36.75:0.1) and \control{(1.25,-0.9)}{(180:0.1)} .. controls +(0:0.35) and \control{(1,1.25)}{(-36.75:1)} -- cycle;
\fill [blue!20!white] (1,2.25) .. controls +(143.25:0.1) and \control{(0.75,2.4)}{(0:0.1)} .. controls +(180:0.35) and \control{(1,0.25)}{(143.25:1)} -- (2,-0.5) .. controls +(-36.75:0.1) and \control{(2.25,-0.65)}{(180:0.1)} -- ++(0, -1) -- ++(-3.5,0) -- ++(0,4.5) -- ++(2.5,0) -- cycle;
\end{scope}
\fill [blue!20!white] (h) .. controls +(-150:0.1) and \control{(i)}{(143.25:0.2)} .. controls +(-36.75:0.1) and \control{(l)}{(180:0.2)} .. controls +(60:0.2) and \control{(e)}{(-36.75:0.2)} .. controls +(143.25:0.1) and \control{(h)}{(0:0.2)};
\fill [white!40!gray] (-0.25,2.15) .. controls +(0:0.1) and \control{(0,2)}{(143.25:0.1)} .. controls +(-36.75:0.2) and \control{(h)}{(180:0.2)} .. controls +(30:0.1) and \control{(1,2.25)}{(-36.75:0.2)} .. controls +(143.25:0.1) and \control{(0.75,2.4)}{(0:0.1)}-- cycle;
\fill [white!40!gray] (1,1.25) .. controls +(143.25:0.1) and \control{(l)}{(-130:0.2)} .. controls +(180:0.2) and \control{(2,1.5)}{(143.25:0.2)} .. controls +(-36.75:1) and \control{(2.25,-0.65)}{(0:0.35)} -- (1.25,-0.9) .. controls +(0:0.35) and \control{(1,1.25)}{(-36.75:1)};

%Crossings
\begin{scope}
\clip (-0.25,2.15) .. controls +(0:0.1) and \control{(0,2)}{(143.25:0.1)} .. controls +(-36.75:0.2) and \control{(h)}{(180:0.2)} .. controls +(30:0.1) and \control{(1,2.25)}{(-36.75:0.2)} .. controls +(143.25:0.1) and \control{(0.75,2.4)}{(0:0.1)}-- cycle;
\draw [dashed] (0.75,2.4) .. controls +(180:0.35) and \control{(1,0.25)}{(143.25:1)};
\end{scope}
\begin{scope}
\clip (1,1.25) .. controls +(143.25:0.1) and \control{(l)}{(-130:0.2)} .. controls +(180:0.2) and \control{(2,1.5)}{(143.25:0.2)} .. controls +(-36.75:1) and \control{(2.25,-0.65)}{(0:0.35)} -- (1.25,-0.9) .. controls +(0:0.35) and \control{(1,1.25)}{(-36.75:1)};
\draw [dashed] (4.5, -1.25) .. controls +(140:0.5) and \control{(2,-0.5)}{(14:1)} .. controls +(-36.75:0.1) and \control{(2.25,-0.65)}{(180:0.1)};
\end{scope}
\begin{scope}
\clip (h) .. controls +(180:0.2) and \control{(0,2)}{(-36.75:0.2)} .. controls +(143.25:0.1) and \control{(-0.25,2.15)}{(0:0.1)} .. controls +(180:0.35) and \control{(0,0)}{(143.25:1)} -- (1,-0.75) .. controls +(-36.75:0.1) and \control{(1.25,-0.9)}{(180:0.1)} .. controls +(0:0.35) and \control{(1,1.25)}{(-36.75:1)} -- cycle;
\draw (0.75,2.4) .. controls +(180:0.35) and \control{(1,0.25)}{(143.25:1)};
\end{scope}
\draw (1,2.25) .. controls +(143.25:0.1) and \control{(0.75,2.4)}{(0:0.1)};
\draw (2.25,-0.65) .. controls +(0:0.35) and \control{(2,1.5)}{(-36.75:1)};
\draw (0,2) .. controls +(143.25:0.1) and \control{(-0.25,2.15)}{(0:0.1)} .. controls +(180:0.35) and \control{(0,0)}{(143.25:1)} (1,-0.75) .. controls +(-36.75:0.1) and \control{(1.25,-0.9)}{(180:0.1)} .. controls +(0:0.35) and \control{(1,1.25)}{(-36.75:1)} ;
\begin{scope}[even odd rule] 
\clip (0, 2.25) rectangle (3, 1) (l) circle (0.1cm);
\draw (1,1.25) .. controls +(143.25:0.1) and \control{(l)}{(-130:0.2)} .. controls +(60:0.2) and \control{(e)}{(-36.75:0.2)};
\end{scope}
\draw (e) .. controls +(143.25:0.1) and \control{(h)}{(0:0.2)} .. controls +(180:0.2) and \control{(0,2)}{(-36.75:0.2)};
\begin{scope}[even odd rule] 
\clip (0,2) -- (1,1.25) -- (2,1.5) -- (1,2.25) -- (0,2) (h) circle (0.13cm);
\draw (1,2.25) .. controls +(-36.75:0.2) and \control{(h)}{(30:0.1)} .. controls +(-150:0.1) and \control{(i)}{(143.25:0.2)};
\end{scope}
\draw (i) .. controls +(-36.75:0.1) and \control{(l)}{(180:0.2)} .. controls +(0:0.2) and \control{(2,1.5)}{(143.25:0.2)};
\end{scope}
\end{tikzpicture}
\begin{tikzpicture}[scale = 1.15]
\clip (-2.75,-4.25) rectangle (9,2);
%Large circle 1
\def\r{0} \def\xs{0} \def\ys{0}
\begin{scope}[even odd rule]
\clip (-2.75,-4.25) rectangle (9,2) [xshift =2cm, yshift = -0.35cm] (0:1.25) arc (0:-180:1.25) -- (-180:1.5) arc (-180:0:1.5) --cycle;
\clip (-2.75,-4.25) rectangle (9,2) [xshift =-1.5cm, yshift = 0.5cm] (0:0.75) arc (0:-180:0.75) -- (-180:1) arc (-180:0:1) --cycle;
\clip (-2.75,-4.25) rectangle (9,2) [xshift =-0.5cm, yshift = -1.75cm] (90:0.75) arc (90:180:0.75) -- (180:1) arc (180:90:1) --cycle;
\draw [xshift = \xs cm, yshift = \ys cm, rotate = \r] (115:1.25) arc (115:425:1.25);
\draw [xshift = \xs cm, yshift = \ys cm, rotate = \r] (110:1.5) arc (110:430:1.5);
\end{scope}
\begin{scope}[xshift = \xs cm, yshift = \ys cm, rotate = \r]
\coordinate (l) at (80:1.375); 
\coordinate (h) at (100:1.375); 
\coordinate (i) at (90:1.25); 
\coordinate (e) at (90:1.5);
\begin{scope}[even odd rule]
\clip (-1,0.5) rectangle (1,2) (l) circle (0.1cm);
\draw (115:1.25) .. controls +(25:0.2) and \control{(h)}{(-130:0.2)} .. controls +(50:0.2) and \control{(e)}{(180:0.1)} .. controls +(0:0.1) and \control{(l)}{(130:0.2)} .. controls +(-50:0.2) and \control{(65:1.25)}{(165:0.2)}; 
\end{scope}
\begin{scope}[even odd rule]
\clip (-1,0.5) rectangle (1,2) (h) circle (0.1cm);
\draw (110:1.5) .. controls +(20:0.1) and \control{(h)}{(130:0.2)} .. controls +(-50:0.2) and \control{(i)}{(180:0.1)} .. controls +(0:0.1) and \control{(l)}{(-130:0.2)} .. controls +(50:0.2) and \control{(70:1.5)}{(160:0.1)}; 
\end{scope}
\end{scope}

%Large circle 2
\def\r{0} \def\xs{2} \def\ys{-0.35}
\begin{scope}[even odd rule]
\clip (-2.75,-4.25) rectangle (9,2) [xshift =0cm, yshift = 0cm] (0:1.25) arc (0:-90:1.25) -- (-90:1.5) arc (-90:0:1.5) --cycle;
\clip (-2.75,-4.25) rectangle (9,2) [xshift =3.7cm, yshift = 0.5cm] (0:0.75) arc (0:190:0.75) -- (190:1) arc (190:0:1) --cycle;
\clip (-2.75,-4.25) rectangle (9,2) [xshift =3.25cm, yshift = -1.75cm] (150:0.75) arc (150:210:0.75) -- (210:1) arc (210:150:1) --cycle;
\draw [xshift = \xs cm, yshift = \ys cm, line width=0.2cm, white] (180:1.25) arc (180:90:1.25);
\draw [xshift = \xs cm, yshift = \ys cm, line width=0.2cm, white] (180:1.5) arc (180:90:1.5);
\draw [xshift = \xs cm, yshift = \ys cm, rotate = \r] (115:1.25) arc (115:425:1.25);
\draw [xshift = \xs cm, yshift = \ys cm, rotate = \r] (110:1.5) arc (110:430:1.5);
\end{scope}
\begin{scope}[xshift = \xs cm, yshift = \ys cm, rotate = \r]
\coordinate (l) at (80:1.375); 
\coordinate (h) at (100:1.375); 
\coordinate (i) at (90:1.25); 
\coordinate (e) at (90:1.5);
\begin{scope}[even odd rule]
\clip (-1,0.5) rectangle (1,2) (h) circle (0.1cm);
\draw (115:1.25) .. controls +(25:0.2) and \control{(h)}{(-130:0.2)} .. controls +(50:0.2) and \control{(e)}{(180:0.1)} .. controls +(0:0.1) and \control{(l)}{(130:0.2)} .. controls +(-50:0.2) and \control{(65:1.25)}{(165:0.2)}; 
\end{scope}
\begin{scope}[even odd rule]
\clip (-1,0.5) rectangle (1,2) (l) circle (0.1cm);
\draw (110:1.5) .. controls +(20:0.1) and \control{(h)}{(130:0.2)} .. controls +(-50:0.2) and \control{(i)}{(180:0.1)} .. controls +(0:0.1) and \control{(l)}{(-130:0.2)} .. controls +(50:0.2) and \control{(70:1.5)}{(160:0.1)}; 
\end{scope}
\end{scope}

%Circle 3
\def\r{0} \def\xs{3.7} \def\ys{0.5}
\begin{scope}[even odd rule]
\clip (-2.75,-4.25) rectangle (9,2) [xshift =2cm, yshift = -0.35cm] (30:1.25) arc (30:180:1.25) -- (180:1.5) arc (180:30:1.5) --cycle;
\draw [xshift = \xs cm, yshift = \ys cm, line width=0.15cm, white] (-90:0.75) arc (-90:-140:0.75);
\draw [xshift = \xs cm, yshift = \ys cm, rotate = \r] (25:0.75) arc (25:335:0.75);
\draw [xshift = \xs cm, yshift = \ys cm, line width=0.2cm, white] (-90:1) arc (-90:-140:1);
\draw [xshift = \xs cm, yshift = \ys cm, rotate = \r] (20:1) arc (20:340:1);
\end{scope}
\begin{scope}[xshift = \xs cm, yshift = \ys cm, rotate = \r]
\coordinate (l) at (-10:0.875); 
\coordinate (h) at (10:0.875); 
\coordinate (i) at (0:0.75); 
\coordinate (e) at (0:1);
\begin{scope}[even odd rule]
\clip (-1.5,-1) rectangle (1.5,1) (h) circle (0.1cm);
\draw (25:0.75) .. controls +(-65:0.1) and \control{(h)}{(150:0.1)} .. controls +(-30:0.1) and \control{(e)}{(90:0.1)} .. controls +(-90:0.1) and \control{(l)}{(30:0.1)} .. controls +(-150:0.1) and \control{(-25:0.75)}{(65:0.1)}; 
\end{scope}
\begin{scope}[even odd rule]
\clip (-1.5,-1) rectangle (1.5,1)  (l) circle (0.1cm);
\draw (20:1) .. controls +(-70:0.1) and \control{(h)}{(30:0.1)} .. controls +(-150:0.1) and \control{(i)}{(90:0.1)} .. controls +(-90:0.1) and \control{(l)}{(150:0.1)} .. controls +(-30:0.1) and \control{(-20:1)}{(70:0.1)}; 
\end{scope}
\end{scope}

%Circle 4
\def\r{0} \def\xs{3.25} \def\ys{-1.75}
\begin{scope}[even odd rule]
\clip (-2.75,-4.25) rectangle (9,2) [xshift =2cm, yshift = -0.35cm] (-60:1.25) arc (-60:-180:1.25) -- (-180:1.5) arc (-180:-60:1.5) --cycle;
\clip (-2.75,-4.25) rectangle (9,2) [xshift =3.75cm, yshift = -3cm] (130:0.75) arc (130:210:0.75) -- (210:1) arc (210:130:1) --cycle;
\draw [xshift = \xs cm, yshift = \ys cm, line width=0.2cm, white] (60:0.75) arc (60:120:0.75);
\draw [xshift = \xs cm, yshift = \ys cm, rotate = \r] (25:0.75) arc (25:335:0.75);
\draw [xshift = \xs cm, yshift = \ys cm, line width=0.2cm, white] (60:1) arc (60:120:1);
\draw [xshift = \xs cm, yshift = \ys cm, rotate = \r] (20:1) arc (20:340:1);
\end{scope}
\begin{scope}[xshift = \xs cm, yshift = \ys cm, rotate = \r]
\coordinate (l) at (-10:0.875); 
\coordinate (h) at (10:0.875); 
\coordinate (i) at (0:0.75); 
\coordinate (e) at (0:1);
\begin{scope}[even odd rule]
\clip (-1.5,-1) rectangle (1.5,1) (h) circle (0.1cm);
\draw (25:0.75) .. controls +(-65:0.1) and \control{(h)}{(150:0.1)} .. controls +(-30:0.1) and \control{(e)}{(90:0.1)} .. controls +(-90:0.1) and \control{(l)}{(30:0.1)} .. controls +(-150:0.1) and \control{(-25:0.75)}{(65:0.1)}; 
\end{scope}
\begin{scope}[even odd rule]
\clip (-1.5,-1) rectangle (1.5,1)  (l) circle (0.1cm);
\draw (20:1) .. controls +(-70:0.1) and \control{(h)}{(30:0.1)} .. controls +(-150:0.1) and \control{(i)}{(90:0.1)} .. controls +(-90:0.1) and \control{(l)}{(150:0.1)} .. controls +(-30:0.1) and \control{(-20:1)}{(70:0.1)}; 
\end{scope}
\end{scope}

%Circle 5
\def\r{0} \def\xs{3.75} \def\ys{-3}
\begin{scope}[even odd rule]
\clip (-2.75,-4.25) rectangle (9,2) [xshift =3.25cm, yshift = -1.75cm] (150:0.75) arc (150:270:0.75) -- (270:1) arc (270:150:1) --cycle;
\draw [xshift = \xs cm, yshift = \ys cm, line width=0.2cm, white] (60:0.75) arc (60:120:0.75);
\draw [xshift = \xs cm, yshift = \ys cm, rotate = \r] (25:0.75) arc (25:335:0.75);
\draw [xshift = \xs cm, yshift = \ys cm, line width=0.2cm, white] (60:1) arc (60:120:1);
\draw [xshift = \xs cm, yshift = \ys cm, rotate = \r] (20:1) arc (20:340:1);
\end{scope}
\begin{scope}[xshift = \xs cm, yshift = \ys cm, rotate = \r]
\coordinate (l) at (-10:0.875); 
\coordinate (h) at (10:0.875); 
\coordinate (i) at (0:0.75); 
\coordinate (e) at (0:1);
\begin{scope}[even odd rule]
\clip (-1.5,-1) rectangle (1.5,1) (l) circle (0.1cm);
\draw (25:0.75) .. controls +(-65:0.1) and \control{(h)}{(150:0.1)} .. controls +(-30:0.1) and \control{(e)}{(90:0.1)} .. controls +(-90:0.1) and \control{(l)}{(30:0.1)} .. controls +(-150:0.1) and \control{(-25:0.75)}{(65:0.1)}; 
\end{scope}
\begin{scope}[even odd rule]
\clip (-1.5,-1) rectangle (1.5,1)  (h) circle (0.1cm);
\draw (20:1) .. controls +(-70:0.1) and \control{(h)}{(30:0.1)} .. controls +(-150:0.1) and \control{(i)}{(90:0.1)} .. controls +(-90:0.1) and \control{(l)}{(150:0.1)} .. controls +(-30:0.1) and \control{(-20:1)}{(70:0.1)}; 
\end{scope}
\end{scope}

%Circle 6
\def\r{180} \def\xs{-1.5} \def\ys{0.5}
\begin{scope}[even odd rule]
\clip (-2.75,-4.25) rectangle (9,2) [xshift =0cm, yshift = 0cm] (150:1.25) arc (150:210:1.25) -- (210:1.5) arc (210:150:1.5) --cycle;
\draw [xshift = \xs cm, yshift = \ys cm, line width=0.2cm, white] (0:0.75) arc (0:90:0.75);
\draw [xshift = \xs cm, yshift = \ys cm, rotate = \r] (25:0.75) arc (25:335:0.75);
\draw [xshift = \xs cm, yshift = \ys cm, line width=0.2cm, white] (0:1) arc (0:90:1);
\draw [xshift = \xs cm, yshift = \ys cm, rotate = \r] (20:1) arc (20:340:1);
\end{scope}
\begin{scope}[xshift = \xs cm, yshift = \ys cm, rotate = \r]
\coordinate (l) at (-10:0.875); 
\coordinate (h) at (10:0.875); 
\coordinate (i) at (0:0.75); 
\coordinate (e) at (0:1);
\begin{scope}[even odd rule]
\clip (-1.5,-1) rectangle (1.5,1) (h) circle (0.1cm);
\draw (25:0.75) .. controls +(-65:0.1) and \control{(h)}{(150:0.1)} .. controls +(-30:0.1) and \control{(e)}{(90:0.1)} .. controls +(-90:0.1) and \control{(l)}{(30:0.1)} .. controls +(-150:0.1) and \control{(-25:0.75)}{(65:0.1)}; 
\end{scope}
\begin{scope}[even odd rule]
\clip (-1.5,-1) rectangle (1.5,1)  (l) circle (0.1cm);
\draw (20:1) .. controls +(-70:0.1) and \control{(h)}{(30:0.1)} .. controls +(-150:0.1) and \control{(i)}{(90:0.1)} .. controls +(-90:0.1) and \control{(l)}{(150:0.1)} .. controls +(-30:0.1) and \control{(-20:1)}{(70:0.1)}; 
\end{scope}
\end{scope}

%Circle 7
\def\r{180} \def\xs{-0.5} \def\ys{-1.75}
\begin{scope}[even odd rule]
\clip (-2.75,-4.25) rectangle (9,2) [xshift =0cm, yshift = 0cm] (150:1.25) arc (150:250:1.25) -- (250:1.5) arc (250:150:1.5) --cycle;
\draw [xshift = \xs cm, yshift = \ys cm, line width=0.2cm, white] (0:0.75) arc (0:90:0.75);
\draw [xshift = \xs cm, yshift = \ys cm, rotate = \r] (25:0.75) arc (25:335:0.75);
\draw [xshift = \xs cm, yshift = \ys cm, line width=0.2cm, white] (0:1) arc (0:90:1);
\draw [xshift = \xs cm, yshift = \ys cm, rotate = \r] (20:1) arc (20:340:1);
\end{scope}
\begin{scope}[xshift = \xs cm, yshift = \ys cm, rotate = \r]
\coordinate (l) at (-10:0.875); 
\coordinate (h) at (10:0.875); 
\coordinate (i) at (0:0.75); 
\coordinate (e) at (0:1);
\begin{scope}[even odd rule]
\clip (-1.5,-1) rectangle (1.5,1) (l) circle (0.1cm);
\draw (25:0.75) .. controls +(-65:0.1) and \control{(h)}{(150:0.1)} .. controls +(-30:0.1) and \control{(e)}{(90:0.1)} .. controls +(-90:0.1) and \control{(l)}{(30:0.1)} .. controls +(-150:0.1) and \control{(-25:0.75)}{(65:0.1)}; 
\end{scope}
\begin{scope}[even odd rule]
\clip (-1.5,-1) rectangle (1.5,1)  (h) circle (0.1cm);
\draw (20:1) .. controls +(-70:0.1) and \control{(h)}{(30:0.1)} .. controls +(-150:0.1) and \control{(i)}{(90:0.1)} .. controls +(-90:0.1) and \control{(l)}{(150:0.1)} .. controls +(-30:0.1) and \control{(-20:1)}{(70:0.1)}; 
\end{scope}
\end{scope}

\draw [<-] (5,0) -- (6,0);

%Tree
\begin{scope}[xshift = 7cm, yshift = 0, scale = 0.5]
\def\d{0.025}
\def\a{0.6}
\node (c1) [draw, inner sep = \d cm, circle] at (0,0) {{\tiny $+$}};
\node (c2) [draw, inner sep = \d cm, circle] at (2,-0.35) {{\tiny $-$}};
\node (c3) [draw, inner sep = \d cm, circle] at (3.7,0.5) {{\tiny $-$}}; 
\node (c4) [draw, inner sep = \d cm, circle] at (3.25,-1.75) {{\tiny $-$}}; 
\node (c5) [draw, inner sep = \d cm, circle] at (3.75,-3) {{\tiny $+$}};
\node (c6) [draw, inner sep = \d cm, circle] at (-1.5,0.5) {{\tiny $-$}};
\node (c7) [draw, inner sep = \d cm, circle] at (-0.5,-1.75) {{\tiny $+$}};
\draw (c1) -- (c2) (c1) -- (c6) (c1) -- (c7) (c3) -- (c2) (c2) -- (c4) (c4) -- (c5);
\end{scope}
\end{tikzpicture}
\end{center}
  \caption{Top: A positive Hopf band and a plumbing. Bottom: An Hopf arborescent link and an associated planar tree.}
  \label{F:intro}
\end{figure}

\subparagraph*{Our results.} The goal of this paper is to investigate the structure of a particular class of links, which we call \emphdef{Hopf arborescent links}, in order to prove the decidability of some of their $4$-dimensional invariants. 
This family of links is informally defined as follows (we refer to Section~\ref{S:prelim} for more precise definitions). 
A \emphdef{Hopf band} is the surface pictured in Figure~\ref{F:intro} (top left), it can be either \emphdef{positive} or \emphdef{negative} depending on how it twists. 
If they are unlinked, i.e., there exists a sphere separating them, two Hopf bands can be \emphdef{plumbed} together by identifying a square in one to a square in the other, as pictured in Figure~\ref{F:intro} (top). 
The class of \emphdef{Hopf arborescent links} is the class of links arising as the boundary of some iterated tree-like sequence of plumbings of Hopf bands. 
Hopf arborescent links can naturally be described using labelled trees, see Figure~\ref{F:intro} (bottom), and are a subfamily of the more general arborescent links~\cite{bonahon1979new,Gabai_Genus_arborescent}. 
The (topological or smooth) \emphdef{genus defect} of a knot is defined as $\Delta g(K)=g(K)-g_4(K)$. 
It measures how much its $4$-genus differs from its classical genus, and this definition can be extended to oriented links by considering surfaces having that oriented link as their boundary. 
Our main result is the following:

\begin{theorem}\label{T:main}
For any fixed $k$, deciding whether an Hopf arborescent link $L$ has genus defect at most~$k$ is decidable. 
This holds both in the topological and smooth categories.
\end{theorem}

The proof of Theorem~\ref{T:main} is not constructive. It is obtained as a corollary of another result, which establishes a well-behaved minor theory for Hopf arborescent links. 
A subsurface~$\Sigma' \subseteq \Sigma$ of a Seifert surface $\Sigma$ is \emphdef{incompressible} if the complement $\Sigma \setminus \Sigma'$ has no open disc component. 
Given two Seifert surfaces $\Sigma_1$ and $\Sigma_2$ in $\R^3$, we say that $\Sigma_1$ is a \emphdef{surface-minor} of~$\Sigma_2$ (or \emphdef{minor} for short), denoted by $\Sigma_1 \preccurlyeq \Sigma_2$, if $\Sigma_1$ is isotopic to an incompressible subsurface of~$\Sigma_2$. 
This minor relation was introduced by Baader~\cite{baader2014positive} (see also~\cite{baader2012minor,baader2021minor}) with the goal of characterising those links that are closure of positive braids and whose signature is equal to twice their genus.
The underlying question, still open today, is whether the canonical Seifert surfaces associated to positive braid closures form a well-quasi-order. 
A particularity of Hopf arborescent links is that they are \emphdef{fibred} (see the definition in Section~\ref{S:prelim}), which implies that to each link is associated a canonical Seifert surface (see Theorem~\ref{th_fibred_surface_unique}), that we call a \emphdef{Hopf arborescent surface}. 
The notion of surface-minor naturally implies a minor relation for Hopf arborescent links.
Our second result proves that Hopf arborescent surfaces are well-quasi-ordered under surface-minors. 

\begin{theorem}\label{T:wqo}
The minor relation $\preccurlyeq$ is a well-quasi-order for the set of Hopf arborescent surfaces, that is, for any infinite sequence $(\Sigma_n)_{n \in \mathbb{N}}$ of Hopf arborescent surfaces, there exists $i<j$ in $\mathbb{N}$ such that $\Sigma_i \preccurlyeq \Sigma_j$.
\end{theorem}

The idea behind the proof of Theorem~\ref{T:wqo} is to study a specific subset of the possible surface-minors that interacts nicely with an encoding of Hopf arborescent surfaces via labelled plane trees. 
We can then leverage the celebrated Kruskal tree theorem~\cite{Kruskal_tree_minor} to prove that the minor relation is a well-quasi-order. The connection from Theorem~\ref{T:wqo} to Theorem~\ref{T:main} follows from the fact the genus defect is minor-monotone, i.e., if $\Sigma_1 \preccurlyeq \Sigma_2$ are Seifert surfaces of minimal genus for links $K_1$ and $K_2$, then $\Delta g(K_1) \leq \Delta g(K_2)$. This is not a new observation (see~\cite[Lemma~6]{baader2018topological}), we provide a proof in Proposition~\ref{prop_stable_defect} for completeness. Therefore, for Hopf arborescent links, having genus defect at most $k$ is characterized by a finite number of forbidden minors, and the algorithm of Theorem~\ref{T:main} proceeds by checking those. However, testing whether a surface is a minor of another one seems to be a very hard problem: even testing whether two tori are isotopic is already as hard a knot equivalence and an algorithm for genus two surfaces was only very recently found~\cite{baroni}. This problem is circumvented thanks to our restriction of the minor relation to one that is well-tailored to the arborescent structure of our links, which allows us to work entirely at the level of trees. In particular, Theorem~\ref{T:main} does not strictly follow from Theorem~\ref{T:wqo} but rather from its proof (see Proposition~\ref{prop_link_wqo}).

\smallskip

While the algorithms behind Theorem~\ref{T:main} are not explicit, we would like to offer three reasons to motivate our results. 
First, Theorem~\ref{T:main} proves that the corresponding problems are \emph{not} undecidable, which is significant in the landscape of 4-dimensional topology. 
Second, this kind of existential algorithmic result has been a strong guiding light in algorithm design in the past decades: for a vast family of graph problems, the fact that an algorithm merely exists follows from Robertson-Seymour theory, and this has provided a strong impetus to actually look for explicit algorithms and optimise their complexity. This has been particularly influential in parameterized algorithms, we refer for example to the discussion in Chapter 6.3 in the book on parameterized algorithms~\cite{cygan2015parameterized}, where it is conjectured that a result like our Theorem~\ref{T:main} precludes W[1]-hardness. Similarly, we are hopeful that our results can inspire future work aiming at developing explicit algorithms in 4-dimensional topology. Additionally, our framework directly proves that any property that is stable with respect to our link-minor relation (see Section~\ref{ss:minors})  is decidable on the class of Hopf arborescent links.
Finally, while it is certainly not the case that minor-based approaches can encompass the entirety of knot theory, it is fruitful to delineate exactly the classes which they can illuminate. In that respect, we find it interesting that our proof of Theorem~\ref{T:wqo} strongly relies on the structure of Hopf bands and does not seem to generalize to the wider family of arborescent knots, even when one bounds the number of twists in each band (see Remark~\ref{R:minors}).

\subparagraph*{Related work.} It was observed by Baader and Dehornoy~\cite{baader2012minor} that the natural Seifert surfaces for another class of knots, the positive braid knots with bounded braid index (we refer to the papers for the relevant definitions) also form a well-quasi-order. 
Furthermore, Liechti~\cite{liechti2020genus} proved that even without bounding the braid index, the set of positive braid knots of bounded genus defect is characterized by a finite number of forbidden minors. 
Since the minor relation in that setting simply amounts to removing letters in the braid presentation, this readily yields decidability as in our Theorem~\ref{T:main}. 
While the two results are incomparable, we emphasize that our result also applies to links, and also features negative crossings (coming from negative Hopf bands): this extends the impact of our result to the smooth category, while for strongly quasipositive knots (and thus positive braid knots), the smooth defect is zero since the smooth $4$-genus and the classical genus coincide~\cite{rudolph1993}.

All the knots and links we consider, as well as those considered by various authors in the context of surface-minor theory, are fibred (see again the definition in Section~\ref{S:prelim}). 
This property is important as it brings control on the classical genus of the links. 
Also it is easy to construct infinite families of incomparable surfaces when dropping this assumption: the set $(A_n)_{n\in\Z}$ of those unknotted annuli in~$\Sp^3$ with $n$ twists forms an infinite antichain. 
In this direction, an optimistic conjecture would be that the collection of all fibred surfaces in~$\Sp^3$ is a well-quasi-ordered set. 
If true, that would provide a strong generalisation of Theorem~\ref{T:wqo}. 
However no strategy of proof is known to the authors for such a statement.

Also, it follows from a result of Giroux and Goodman~\cite{giroux2006stable} that \emph{any} fibred link can be obtained from the unknot from a sequence of plumbings and deplumbings (a natural reverse operation to plumbing) of Hopf bands. 
While these (de)plumbings might not have the arborescent structure that characterizes ours, this shows that Hopf bands can be considered as basic building blocks for a wide class of three-dimensional objects.

\subparagraph*{Organization of this paper.} After providing background and going through basic concepts of knot theory, fibred surfaces and well-quasi-orders in Section~\ref{S:prelim}, we focus on Hopf arborescent links in Section~\ref{sec_Hopflink}. There, we define a precise construction of these links from plane trees, investigate this class of links, and prove Theorem~\ref{T:wqo}. In Section~\ref{sec_decidability}, we prove Theorem~\ref{T:main}. Finally, in Section~\ref{sec_defect}, we first provide examples of Hopf arborescent links with non zero defect, and then explain how to combine them to obtain examples with arbitrarily large defect.

\section{Preliminaries}\label{S:prelim}

\subparagraph*{Knot theory.} We only recall the definitions that are critical to this paper, and refer to the textbooks of Burde and Zieschang~\cite{burde2002knots} or Rolfsen~\cite{Rolfsen_Knots} for a more general introduction to knot theory, and to Teichner~\cite{teichner} for a beginner-friendly introduction to its 4-dimensional aspects. 
A \emphdef{knot}, respectively a \emphdef{link}, is a polygonal embedding of the circle $\Sp^1$, respectively of a disjoint union of circles, into $\Sp^3$. 
Every knot and every component of a link inherits an orientation from the orientation of~$\Sp^1$. For algorithmic purposes, we assume that an input link is given as a link diagram, i.e., a directed plane $4$-valent graph with decorations at vertices indicating which strands are going over and under. Since our article focuses on decidability problems, switching to a different input (e.g., polygonal curves in $\R^3$) makes no difference. 
A \emphdef{Seifert surface} for a link $L$ is a compact connected oriented surface~$\Sigma$ embedded in $\Sp^3$ such that the oriented boundary of $\Sigma$ is $L$. 
Throughout this article, we consider knots, links, and surfaces embedded in $\Sp^3$ up to isotopies (continuous deformations without self-crossings). 
The \emphdef{(classical) genus} of a link~$L$, denoted by $g(L)$, is the smallest possible genus of a Seifert surface for $L$.

%The 3-dimensional sphere can be seen as the boundary of the 4-dimensional ball~$\B^4$. Being embedded in $\Sp^3$, a knot or a link can also be obtained as the boundary of surfaces embedded in~$\B^4$. However, a simple coning construction shows that any knot in $\Sp^3$ bounds a topological disc in $\B^4$. This motivates the following definition: a surface $\Sigma$ embedded in $\B^4$ is \emphdef{locally flat} if for each point $x \in \Sigma$, there is a neighbourhood $U$ in $\Sigma$ and a neighbourhood $V$ in $\B^4$ such that the pair $(U,V)$ is homeomorphic to $(\mathring{\B^2},\mathring{\B^4})$. 
The 3-dimensional sphere can be seen as the boundary of the 4-dimensional ball~$\B^4$. Being embedded in $\Sp^3$, a knot or a link can also be obtained as the boundary of surfaces embedded in~$\B^4$. However, any knot $K$ in $\Sp^3$ can be used as a base that tapers to a point, the apex, inside~$B^4$ to define a cone that bounds $K$. Hence, any knot in $\Sp^3$ bounds a topological disc in $\B^4$. This motivates the following definition: a surface $\Sigma$ embedded in $\B^4$ is \emphdef{locally flat} if for each point $x \in \Sigma$, there is a neighbourhood $U$ in $\Sigma$ and a neighbourhood $V$ in $\B^4$ such that the pair $(U,V)$ is homeomorphic to the standard $(\mathring{\B^2},\mathring{\B^4})$. In the coning construction above, the latter condition is not satisfied at the apex, where the boundary of a disc is the knot instead of a standard $\mathbb{S}^1$.
The \emphdef{topological} (respectively \emphdef{smooth}) \emphdef{$4$-dimensional genus}, or simply \emphdef{$4$-genus} of a link $L$, denoted by\footnote{Our notation is intentionally ambiguous with respect to smooth or topological genus, because all our arguments will apply equally well in both categories.} $g_4(L)$ is the smallest possible genus of a compact connected oriented surface that is locally flat (respectively smoothly) embedded in $\B^4$, and that has $L$ as its boundary.

A knot is \emphdef{topologically} (resp. \emphdef{smoothly}) \emphdef{slice} if it bounds a locally flat (resp. a smooth) disc in $\B^4$. 
The \emphdef{(topological or smooth) defect} of a link $L$ is the quantity $\Delta g(L)=g(L)-g_4(L)$, where $g_4(L)$ denotes the topological or smooth $4$-genus of $L$.

A \emphdef{positive}, resp. \emphdef{negative, Hopf band} is an unknotted annulus with a positive, resp. negative, full twist, as pictured in Figure~\ref{pic_def_hopf_band}. 
A \emphdef{Hopf link} is the boundary of a Hopf band. 
Note that the two components of a positive Hopf link have linking number $+1$, while the two components of a negative Hopf link have linking number~$-1$. 
A Hopf band naturally retracts to a trivial knot, which we call its \emphdef{core}.

\begin{figure}[H]
\begin{center}
\begin{tikzpicture}
%Crossings and intersections
\coordinate (cl) at (0,0);
\coordinate (ch) at (0,1);
\coordinate (clw) at (-120:0.5);
\coordinate (cle) at (-60:0.5);
\coordinate (chw) at ($(120:0.5)+(0,1)$);
\coordinate (che) at ($(60:0.5)+(0,1)$);
\coordinate (bll) at (-0.75,-1.5);
\coordinate (blh) at (-0.75,-1);
\coordinate (bhh) at (-0.75,2.5);
\coordinate (bhl) at (-0.75,2);
\node (l) [circle, inner sep = 2pt, fill, white] at (cl) {};
\node (h) [circle, inner sep = 2pt, fill, white] at (ch) {};

%Surface
\fill [blue, opacity = 0.2] (cl) -- (l.north east) .. controls +(45:0.35) and \control{(h.south east)}{(-45:0.35)} -- (ch) -- (h.north east) .. controls +(45:0.2) and \control{(che)}{(-90:0.1)} .. controls +(90:0.4) and \control{(bhh)}{(0.8,0)} .. controls +(-1.5,0) and \control{(bll)}{(-1.5,0)} .. controls +(0.8,0) and \control{(cle)}{(-90:0.4)} .. controls +(90:0.1) and \control{(l.south east)}{(-45:0.35)} -- cycle;
\fill [white!40!gray] (cl) -- (l.north east) .. controls +(45:0.35) and \control{(h.south east)}{(-45:0.35)} -- (ch) -- (h.north west) .. controls +(135:0.2) and \control{(chw)}{(-90:0.1)} .. controls +(90:0.2) and \control{(bhl)}{(0.4,0)} .. controls +(-0.9,0) and \control{(blh)}{(-0.9,0)} .. controls +(0.4,0) and \control{(clw)}{(-90:0.2)} .. controls +(90:0.1) and \control{(l.south west)}{(-135:0.35)} -- cycle;
\fill [white] (cl) -- (l.north west) .. controls +(135:0.35) and \control{(h.south west)}{(-135:0.35)} -- (ch) -- (h.north west) .. controls +(135:0.2) and \control{(chw)}{(-90:0.1)} .. controls +(90:0.2) and \control{(bhl)}{(0.4,0)} .. controls +(-0.9,0) and \control{(blh)}{(-0.9,0)} .. controls +(0.4,0) and \control{(clw)}{(-90:0.2)} .. controls +(90:0.1) and \control{(l.south west)}{(-135:0.35)} -- cycle;

%Strands
\begin{scope}[thick]
\draw (l.north east) .. controls +(45:0.35) and \control{(h.south east)}{(-45:0.35)} node [pos = 0.8, sloped] {{\tiny $>$}};
\draw (l.north west) .. controls +(135:0.35) and \control{(h.south west)}{(-135:0.35)};
\draw (l.south east) .. controls +(-45:0.2) and \control{(cle)}{(90:0.1)};
\draw (l.south west) .. controls +(-135:0.2) and \control{(clw)}{(90:0.1)};
\draw (blh) .. controls +(0.4,0) and \control{(clw)}{(-90:0.2)};
\draw (bll) .. controls +(0.8,0) and \control{(cle)}{(-90:0.4)}; 
\draw (h.north east) .. controls +(45:0.2) and \control{(che)}{(-90:0.1)} node [pos = 0.3, sloped] {{\tiny $>$}};
\draw (h.north west) .. controls +(135:0.2) and \control{(chw)}{(-90:0.1)};
\draw (bhl) .. controls +(0.4,0) and \control{(chw)}{(90:0.2)};
\draw (bhh) .. controls +(0.8,0) and \control{(che)}{(90:0.4)};
\draw (blh) .. controls +(-0.9,0) and \control{(bhl)}{(-0.9,0)};  
\draw (bll) .. controls +(-1.5,0) and \control{(bhh)}{(-1.5,0)};
\draw (l.south east) -- (l.north west);
\draw (h.south east) -- (h.north west);

\begin{scope}[xshift = -0.025cm,yshift = 1.65cm, orange]
\draw (-0.15,0) arc (-180:135:0.15) -- ++(90:0.1);
\draw ($(0,0)+(135:0.15)$) -- ++(0:0.1);
\end{scope}

\end{scope}

\begin{scope}[xshift = -4.5cm]
%Crossings and intersections
\coordinate (cl) at (0,0);
\coordinate (ch) at (0,1);
\coordinate (clw) at (-120:0.5);
\coordinate (cle) at (-60:0.5);
\coordinate (chw) at ($(120:0.5)+(0,1)$);
\coordinate (che) at ($(60:0.5)+(0,1)$);
\coordinate (bll) at (-0.75,-1.5);
\coordinate (blh) at (-0.75,-1);
\coordinate (bhh) at (-0.75,2.5);
\coordinate (bhl) at (-0.75,2);
\node (l) [circle, inner sep = 2pt, fill, white] at (cl) {};
\node (h) [circle, inner sep = 2pt, fill, white] at (ch) {};

%Surface
\fill [blue, opacity = 0.2] (cl) -- (l.north east) .. controls +(45:0.35) and \control{(h.south east)}{(-45:0.35)} -- (ch) -- (h.north east) .. controls +(45:0.2) and \control{(che)}{(-90:0.1)} .. controls +(90:0.4) and \control{(bhh)}{(0.8,0)} .. controls +(-1.5,0) and \control{(bll)}{(-1.5,0)} .. controls +(0.8,0) and \control{(cle)}{(-90:0.4)} .. controls +(90:0.1) and \control{(l.south east)}{(-45:0.35)} -- cycle;
\fill [white!40!gray] (cl) -- (l.north east) .. controls +(45:0.35) and \control{(h.south east)}{(-45:0.35)} -- (ch) -- (h.north west) .. controls +(135:0.2) and \control{(chw)}{(-90:0.1)} .. controls +(90:0.2) and \control{(bhl)}{(0.4,0)} .. controls +(-0.9,0) and \control{(blh)}{(-0.9,0)} .. controls +(0.4,0) and \control{(clw)}{(-90:0.2)} .. controls +(90:0.1) and \control{(l.south west)}{(-135:0.35)} -- cycle;
\fill [white] (cl) -- (l.north west) .. controls +(135:0.35) and \control{(h.south west)}{(-135:0.35)} -- (ch) -- (h.north west) .. controls +(135:0.2) and \control{(chw)}{(-90:0.1)} .. controls +(90:0.2) and \control{(bhl)}{(0.4,0)} .. controls +(-0.9,0) and \control{(blh)}{(-0.9,0)} .. controls +(0.4,0) and \control{(clw)}{(-90:0.2)} .. controls +(90:0.1) and \control{(l.south west)}{(-135:0.35)} -- cycle;

%Strands
\begin{scope}[thick]
\draw (l.north east) .. controls +(45:0.35) and \control{(h.south east)}{(-45:0.35)} node [pos = 0.8, sloped] {{\tiny $>$}};
\draw (l.north west) .. controls +(135:0.35) and \control{(h.south west)}{(-135:0.35)};
\draw (l.south east) .. controls +(-45:0.2) and \control{(cle)}{(90:0.1)};
\draw (l.south west) .. controls +(-135:0.2) and \control{(clw)}{(90:0.1)};
\draw (blh) .. controls +(0.4,0) and \control{(clw)}{(-90:0.2)};
\draw (bll) .. controls +(0.8,0) and \control{(cle)}{(-90:0.4)}; 
\draw (h.north east) .. controls +(45:0.2) and \control{(che)}{(-90:0.1)}node [pos = 0.3, sloped] {{\tiny $>$}};
\draw (h.north west) .. controls +(135:0.2) and \control{(chw)}{(-90:0.1)};
\draw (bhl) .. controls +(0.4,0) and \control{(chw)}{(90:0.2)};
\draw (bhh) .. controls +(0.8,0) and \control{(che)}{(90:0.4)};
\draw (blh) .. controls +(-0.9,0) and \control{(bhl)}{(-0.9,0)};  
\draw (bll) .. controls +(-1.5,0) and \control{(bhh)}{(-1.5,0)};
\draw (l.south west) -- (l.north east);
\draw (h.south west) -- (h.north east);

\begin{scope}[xshift = -0.025cm,yshift = 1.65cm, orange]
\draw (-0.15,0) arc (-180:135:0.15) -- ++(90:0.1);
\draw ($(0,0)+(135:0.15)$) -- ++(0:0.1);
\end{scope}
\end{scope}
\end{scope}

\begin{scope}[xshift = 3cm]
%Crossings and intersections
\coordinate (cl) at (0,0);
\coordinate (ch) at (0,1);
\coordinate (clw) at (-120:0.5);
\coordinate (cle) at (-60:0.5);
\coordinate (chw) at ($(120:0.5)+(0,1)$);
\coordinate (che) at ($(60:0.5)+(0,1)$);
\coordinate (bll) at (-0.75,-1.5);
\coordinate (blh) at (-0.75,-1);
\coordinate (bhh) at (-0.75,2.5);
\coordinate (bhl) at (-0.75,2);
\node (l) [circle, inner sep = 2pt, fill, white] at (cl) {};
\node (h) [circle, inner sep = 2pt, fill, white] at (ch) {};

%Surface
\fill [blue, opacity = 0.2] (cl) -- (l.north east) .. controls +(45:0.35) and \control{(h.south east)}{(-45:0.35)} -- (ch) -- (h.north east) .. controls +(45:0.2) and \control{(che)}{(-90:0.1)} .. controls +(90:0.4) and \control{(bhh)}{(0.8,0)} .. controls +(-1.5,0) and \control{(bll)}{(-1.5,0)} .. controls +(0.8,0) and \control{(cle)}{(-90:0.4)} .. controls +(90:0.1) and \control{(l.south east)}{(-45:0.35)} -- cycle;
\fill [white!40!gray] (cl) -- (l.north east) .. controls +(45:0.35) and \control{(h.south east)}{(-45:0.35)} -- (ch) -- (h.north west) .. controls +(135:0.2) and \control{(chw)}{(-90:0.1)} .. controls +(90:0.2) and \control{(bhl)}{(0.4,0)} .. controls +(-0.9,0) and \control{(blh)}{(-0.9,0)} .. controls +(0.4,0) and \control{(clw)}{(-90:0.2)} .. controls +(90:0.1) and \control{(l.south west)}{(-135:0.35)} -- cycle;
\fill [white] (cl) -- (l.north west) .. controls +(135:0.35) and \control{(h.south west)}{(-135:0.35)} -- (ch) -- (h.north west) .. controls +(135:0.2) and \control{(chw)}{(-90:0.1)} .. controls +(90:0.2) and \control{(bhl)}{(0.4,0)} .. controls +(-0.9,0) and \control{(blh)}{(-0.9,0)} .. controls +(0.4,0) and \control{(clw)}{(-90:0.2)} .. controls +(90:0.1) and \control{(l.south west)}{(-135:0.35)} -- cycle;

%Core
\begin{scope}[purple]
\draw (0,1) .. controls +(90:1) and \control{($(bhh)!0.5!(bhl)$)}{(0:0.4)} .. controls + (180:1.2) and \control{($(blh)!0.5!(bll)$)}{(-180:1.2)} .. controls +(0:0.4) and \control{(0,0)}{(-90:1)} -- cycle;
\end{scope}

%Strands
\begin{scope}[thick]
\draw (l.north east) .. controls +(45:0.35) and \control{(h.south east)}{(-45:0.35)};
\draw (l.north west) .. controls +(135:0.35) and \control{(h.south west)}{(-135:0.35)};
\draw (l.south east) .. controls +(-45:0.2) and \control{(cle)}{(90:0.1)};
\draw (l.south west) .. controls +(-135:0.2) and \control{(clw)}{(90:0.1)};
\draw (blh) .. controls +(0.4,0) and \control{(clw)}{(-90:0.2)};
\draw (bll) .. controls +(0.8,0) and \control{(cle)}{(-90:0.4)}; 
\draw (h.north east) .. controls +(45:0.2) and \control{(che)}{(-90:0.1)};
\draw (h.north west) .. controls +(135:0.2) and \control{(chw)}{(-90:0.1)};
\draw (bhl) .. controls +(0.4,0) and \control{(chw)}{(90:0.2)};
\draw (bhh) .. controls +(0.8,0) and \control{(che)}{(90:0.4)};
\draw (blh) .. controls +(-0.9,0) and \control{(bhl)}{(-0.9,0)};  
\draw (bll) .. controls +(-1.5,0) and \control{(bhh)}{(-1.5,0)};
\draw (l.south east) -- (l.north west);
\draw (h.south east) -- (h.north west);
\end{scope}
\end{scope}
\end{tikzpicture}
\caption{\label{pic_def_hopf_band} A negative Hopf band on the left and a positive one on the right with its core in red.}
\end{center}
\end{figure}
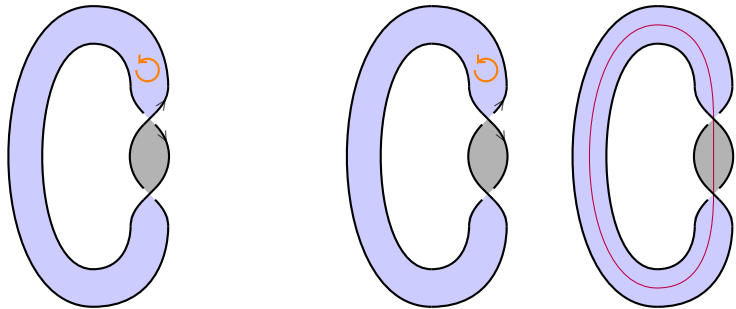

A link is \emphdef{fibred} if the complement $\Sp^3 \setminus L$ \emphdef{fibres} over $\Sp^1$, that is, if there exists a one-parameter continuous family of Seifert surfaces $(\Sigma_t)_{t\in \Sp^1}$ for $L$ which are disjoint except for their boundaries, and whose interiors together foliate  $\Sp^3 \setminus L$. 
These are called \emphdef{fibres} or \emphdef{fibre surfaces} (note that they are all isotopic by definition).

A positive, resp. negative, Hopf link is fibred, with fibre the positive, resp. negative, Hopf band. 
Indeed, seeing $\Sp^3$ as the unit sphere~$\{(z_1, z_2)\,|\,|z_1|^2+|z_2|^2=1\}$ in~$\C^2$, a positive Hopf link is given by the equation $z_1z_2=0$. 
For every argument~$\theta\in\Sp^1$, the equation $\arg(z_1z_2)=\theta$ describes a Seifert surface bounded by the Hopf link, and the collection of these surfaces describes the desired fibration.  

It is a folklore result that goes back at least to Stallings~\cite{Stallings_fibred} that a fibre surface of a fibred link is a Seifert surface of minimal (classical) genus, and moreover this surface is unique up to isotopy (see next paragraph).
So, for fibred links (and all links in this paper will be fibred), it makes sense to speak of the \emphdef{canonical} Seifert surface, by which we mean the unique (up to isotopy) fibre surface for that link.

\subparagraph*{Fibre surface, minimal genus, and uniqueness.}
Here, we comment on the following statements, which are considered as folklore in the knot theoretical community. 
The first one is fundamental for this paper since it brings control over the 3-genus of Hopf arborescent links. 
The second one justifies the notion of \emph{canonical} Seifert surface for fibred links. 

\begin{theorem}\label{th_fibred_surface_minimise}
Fibred surfaces minimise the 3-genus over surfaces with the same oriented boundary. 
\end{theorem}

\begin{theorem}\label{th_fibred_surface_unique}
Two fibred surfaces with the same oriented boundary are isotopic, relatively to the boundary. 
\end{theorem}

Theorem~\ref{th_fibred_surface_minimise} is often presented as a consequence of Stallings' Theorem~\cite{Stallings_fibred}. 
Here we present a semi-elementary proof (for which we claim no novelty), in the sense that it only uses basic notions and results from homology theory.

\begin{proof}[Proof of Theorem~\ref{th_fibred_surface_minimise}]
Consider $L=(L_1\sqcup\dots\sqcup L_k)$ an oriented link in~$\Sp^3$. 
We consider the 3-manifold~$N_L$ which is the complement of an open tubular neighbourhood of~$L$ in~$\Sp^3$. 
It has $k$ boundary components, which are all tori. 

Note the following isomorphisms, from elementary algebraic topology:
by excision~\cite[p119]{Hatcher_Algebraic_Topology}, one has~$\H_2(\Sp^3, L; \Z)\simeq\H_2(N_L, \partial N_L; \Z)$; 
in the long exact sequence $\dots\to\H_2(\Sp^3; \Z)\to\H_2(\Sp^3, L; \Z)\to \H_1(L; \Z)\to  \H_1(\Sp^3; \Z)\to\dots$ the first and last term are trivial so the boundary map gives an isomorphism $\partial: \H_2(\Sp^3, L; \Z)\to \H_1(L; \Z)$; 
and by Alexander duality~\cite[p254]{Hatcher_Algebraic_Topology}, $\H^1(\Sp^3\setminus L; \Z)$ is isomorphic to~$\H_1(L; \Z)\simeq \Z^k$. 

The second isomorphism states that a class in $\H_2(\Sp^3, L; \Z)$ is described by its boundary. 
That means that, when we restrict our attention to oriented Seifert surfaces for~$L$, i.e., surfaces~$S$ so that $\partial S=L$, all such surfaces lie in the class $\partial^{-1}(1, \dots, 1)$. 
In particular they are homologous.  
Given two such surfaces, one can consider their restrictions~$S_1, S_2$ to~$N_L$, where they are also homologous. 
Their boundaries are curves on the tori~$\partial N_L$, which therefore have the same homology class. 
With an isotopy it is thus possible to make $S_1, S_2$ parallel and disjoint in a neighbourhood of~$\partial N_L$.

The class of the Seifert surfaces bounded by~$L$ is dual to the class $\lk(\cdot, L_1)+ \dots+ \lk(\cdot, L_k)\in\H^1(\Sp^3\setminus L; \Z)$.
We denote the latter by~$\ell_L$.
We then consider the infinite cyclic covering~$\hat N_L\to N_L$ associated to $\ell_L$. 
This is the covering associated to the morphism~$\ell_L:\pi_1(\Sp^3\setminus L)\to \Z$, so that $\pi_1(\hat N_L)=\ker(\ell_L)$. 

One way to construct $\hat N_L$ is to consider an arbitrary Seifert surface~$S$ for~$L$, cut~$\hat N_L$ along~$S$, thus obtaining a 3-manifold with boundary~$N_{L,S}$. 
This boundary consists of three parts: two copies of~$S$ that we denote by $S^+$ and $S^-$ and call the ``horizontal part'' of the boundary, and a ``vertical part'' coming from~$\partial N_L$ which consists of $k$ annuli (or sutures) connecting $S^+$ to~$S^-$. 
Now consider $\Z$ copies $(N_{L,S}^n)_{n\in\Z}$ of~$N_{L,S}$. 
Call $S_n^+$ and $S_n^-$ the horizontal part of the boundaries, and for every $n\in\N$ glue $S_n^+$ to~$S_{n+1}^-$. 
The resulting 3-manifold~$\cup_{n\in\Z} N_{L, S}^n$ has a natural projection to~$N_L$ and the loops in~$N_L$ that lift to closed loops in~$\hat N_L$ are exactly those in the kernel of~$\ell_L$. 
Therefore we indeed constructed~$\hat N_L$.

Now suppose that $\Sp^3\setminus L$ fibres over the circle with fibres having $L$ as oriented boundary. 
Then the previous construction can be made by choosing for the surface~$S$ a fibre. 
In this case~$N_{L,S}$ is the product manifold~$S\times[0,1]$, and $\hat N_L$ is homeomorphic to~$S\times \R$. 

Consider now an arbitrary Seifert surface~$S'$ for~$L$ with oriented boundary~$L$. 
Then $S'$ is homologous to~$S$, and in particular it is also dual to~$\ell_L$. 
So~$S'$ lifts in~$\hat N_L$ as $\Z$ parallel copies $(S'_n)_{n\in\Z}$. 
Take any such component~$S'_0$. 
Since it lies in $\hat N_L\simeq S\times\R$, the projection on the first coordinate induces a map $f:S'_0\to S$. 
Looking at the neighbourhood of the common boundary~$\partial S=L=\partial S'_0$ we see that $f$ has degree $1$. 
Degree 1 maps induce surjections in homology, which implies the desired inequality.  
More concretely one can argue as follows: if there was a class $a\in\H^1(S; \Z)$ in the kernel of $f^*$, then there would be a class $b$ such that the cup-product $a\smile b$ is the fundamental class~$[S]\in\H^2(S; \Z)$, and therefore one would have $f^*([S])=f^*(a)\smile f^*(b)=0\neq [S'_0]$, a contradiction. 
Therefore there is linear injection from $\H^1(S; \Z)\simeq \Z^{2g(S)+k-1}$ into $\H_1(S'_0; \Z)\simeq \Z^{2g(S'_0)+k-1}$. 
This implies $g(S)\le g(S'_0)$, and proves the statement. 
\end{proof}

Theorem~\ref{th_fibred_surface_unique} is sometimes also attributed to Stallings, but we did not find a corresponding statement in Stallings' articles. 
However a proof is proposed by Whitten~\cite{Whitten}, 
% also explained in Chapter 5 of Burde-Zieschang~\cite{burde2002knots}, 
which relies on the fundamental work of Waldhausen~\cite{Waldhausen}. 
Indeed if $S_1$ is a fibre surface for a link~$L$, then the infinite cyclic covering $\hat N_L$ constructed in the proof of Theorem~\ref{th_fibred_surface_minimise} is homeomorphic to~$S_1\times\R$. 
If $S_2$ is a Seifert surface for~$L$ that is also a fibred surface, then Theorem~\ref{th_fibred_surface_minimise} implies that it is of minimal genus. 
Therefore any of its lifts in~$\hat N_L$ is incompressible (for otherwise it would not minimize the genus).
One can pick a lift $S_2'$ that is disjoint of~$S_1\times\{0\}$ and assume it lives in $S_1\times(0,N)$ for some large enough~$N$. 
Proposition 3.1 of Waldhausen~\cite{Waldhausen} then implies that $S_2'$ is isotopic to~$S_1\times\{0\}$. 
Projecting back in~$N_L$ yields an isotopy from~$S_2$ to~$S_1$.
The proof of~\cite[Proposition 3.1]{Waldhausen} seems too long to be detailed here. 
We only mention that it works by an induction on the complexity of the surface~$S_1$ and by cutting it along curves until it is a disc.

\subparagraph*{Well-quasi-orders.} We refer to Diestel~\cite[Chapter~12]{Diestel_Graph_Theory} for an introduction to well-quasi-orders and graph minor theory. An order $\preccurlyeq$ on a set $X$ is said to be a \emphdef{well-quasi-order} if for every infinite sequence $(x_n)_{n \in \N}$ there exist $i,j \in \N$ such that $i < j$ and $x_i \preccurlyeq x_j$. Equivalently, $\preccurlyeq$ is a well-quasi-order if it is well-founded and has no infinite \emphdef{antichain} that is, no infinite sequence $(x_n)_{n\in \N}$ such that no two elements of $(x_n)$ are comparable for $\preccurlyeq$.
A property $P$ is said to be stable for an order $\preccurlyeq$ if for any $x$ satisfying $P$ and $y \preccurlyeq x$, then $y$ satisfies $P$. It is well-known that if $\preccurlyeq$ is a well-quasi-order and $P$ is a property that is stable for $\preccurlyeq$, then there exists a finite family~$\mathcal{F}$ of elements of $X$ such that $x \in X$ satisfies $P$ if and only if there is no $f$ in $\mathcal{F}$ such that $f \preccurlyeq x$.  The family $\mathcal{F}$ is called a family of \emphdef{excluded minors} for the property $P$.
Therefore, if a parameter $p:X \rightarrow \N$ is monotone with respect to $\preccurlyeq$, i.e. $x \preccurlyeq y$ implies $p(x) \leq p(y)$, then for each $k \in \N$ the property $p(\cdot) < k$ is stable for $\preccurlyeq$ and hence is characterized by such finite family of excluded minors for $X$.

In this paper, a \emphdef{plane tree} is a rooted tree where each vertex $v$ has a label $\ell(v)$ from an alphabet $A$, and the tree is provided with the combinatorial data of an embedding in the plane: each vertex is additionally given a permutation recording the ordering of the edges to its children. 
The root induces an orientation on the tree: every edge $\ens{u,v}$ is directed from $u$ to $v$, written $u \to v$, when $u$ is closer to the root of the tree than $v$ (i.e., edges go toward the leaves), we refer to the trees of Figure~\ref{pic_3D_arborescent_link} for examples. 
In this paper, the alphabet $A$ is~$\{+, -\}$ endowed with the empty ordering $\leq$. 

A plane tree $T_1$ has a \emphdef{homeomorphic embedding into $T_2$}, written $T_1 \hookrightarrow T_2$, if $T_1$ can be obtained from $T_2$ by iteratively (i) removing a leaf and its adjacent edge, (ii) removing a root with a single child, its adjacent edge and rerooting at that child, (iii) reducing labels (with respect to $\leq$) and (iv) contracting paths into edges while preserving the labels of the endpoints, where all these operations must be consistent with the plane embedding. Throughout this paper, the relation $\leq$ will be trivial, so that (iii) will never apply. Notice that this definition extends the notion of \emph{topological minor} on graphs to a setting where vertices are labelled. In particular, it is more restrictive than the notion of \emph{minor}, which allows to contract any edge to a point: in our case we can only contract paths to at least one edge. This property turns out to be critical in order to make our proofs work. The famous theorem of Kruskal~\cite{Kruskal_tree_minor} (see Nash-Williams~\cite{nash-williams_kruskal} for a simpler proof) shows that this order is a well-quasi-order on the set of labelled plane trees.

\begin{theorem}[\cite{Kruskal_tree_minor,nash-williams_kruskal}]\label{T:kruskal}
The homeomorphic embedding on the set of labelled plane trees labelled by a well-quasi-order forms a well-quasi-order.
\end{theorem}

\section{Hopf arborescent links}\label{sec_Hopflink}

Arborescent links are a class of knots and links originally defined and studied by Conway\footnote{Conway called them algebraic links, but this denomination is now more used for the links that come from algebraic curves in~$\C^2$.}. 
This class has received much attention from knot theorists~\cite{bonahon1979new, Gabai_Genus_arborescent, sakuma}. 
In this paper we study a subclass that we call Hopf arborescent links.

\subsection{Hopf plumbing} 
The links that we investigate in this paper are boundaries of surfaces which are defined iteratively from Hopf bands using an operation called {plumbing}.

Let $H$ be a Hopf band and $\Sigma$ be an oriented surface with boundary, and let us assume that they are unlinked, that is, that there exists a sphere $S$ in $\mathbb{S}^3$  separating them. To plumb $H$ on~$\Sigma$, pick an arc~$\alpha$ on~$\Sigma$ whose endpoints lie on~$\partial\Sigma$ and which is not boundary parallel (i.e., $\alpha$ is not isotopic relatively to its endpoints, to an arc in~$\partial\Sigma$). Let $D$ be a small neighbourhood of~$\alpha$ in~$\Sigma$ that we see as a rectangle with two sides on~$\partial\Sigma$ and two sides in the interior of~$\Sigma$. Isotope $\Sigma$ within $\Sp^3 \setminus S$ so that it intersects $S$ exactly on $D$, see Figure~\ref{pic_def_plumbing_sphere}, left. 
Then, define similarly $D'$ a neighbourhood of the unique (up to isotopy) non boundary parallel arc in~$H$ with endpoints in $\partial H$. 
The orientations of $\Sigma$ and~$H$ induce an orientation of the normal direction to the surface (so that concatenating the orientation of the surface with the positive normal direction gives a positive basis in~$\Sp^3$). 
Finally, isotope $H$ within its component of $\Sp^3 \setminus S$, so that $D$ and $D'$ are identified on $S$ in a way that the sides of~$D$ that are on~$\partial S$ are matched with the sides of~$D'$ that are not on~$\partial H$ and the orientations of both rectangles match. 
The resulting surface is said to be obtained from $\Sigma$ by \emphdef{Hopf plumbing $H$ on top of~$\Sigma$ along~$\alpha$}, see Figure~\ref{pic_def_plumbing_sphere}.

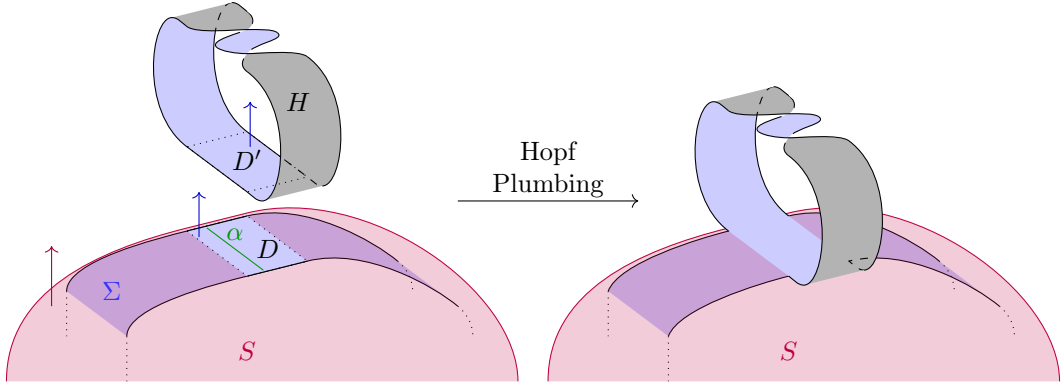
\begin{figure}[H]
\begin{center}
\begin{tikzpicture}[scale = 0.8]
%Band filling
\fill [blue!20!white] (-2,-1) .. controls +(90:0.5) and \control{(0,0)}{(-166:0.5)} -- (1,0.25) .. controls +(14:1) and \control{(3.5,-0.5)}{(140:0.5)} -- (4.5,-1.25) .. controls +(140:0.5) and \control{(2,-0.5)}{(14:1)} -- (1,-0.75) .. controls +(-166:0.5) and \control{(-1,-1.75)}{(90:0.5)} -- cycle;
\draw (-2,-1) .. controls +(90:0.5) and \control{(0,0)}{(-166:0.5)} -- (1,0.25) .. controls +(14:1) and \control{(3.5,-0.5)}{(140:0.5)}; 
\draw [xshift = 1cm, yshift = -0.75cm] (-2,-1) .. controls +(90:0.5) and \control{(0,0)}{(-166:0.5)} -- (1,0.25) .. controls +(14:1) and \control{(3.5,-0.5)}{(140:0.5)};
\draw [dotted] (-2,-1) -- ++(0,-0.75);
\draw [dotted, xshift = 1cm, yshift = -0.75cm] (-2,-1) -- ++(0,-0.75);
\draw [dotted] (3.5,-0.5) .. controls +(-40:0.2) and \control{(3.75,-1)}{(90:0.2)};
\draw [dotted, xshift = 1cm, yshift = -0.75cm] (3.5,-0.5) .. controls +(-40:0.2) and \control{(3.75,-1)}{(90:0.2)};

%Sphere
\fill [opacity = 0.2, purple] (-3,-2.5) .. controls +(90:1.5) and (-1.5,-0.325) .. (0,0.05) -- (1,0.3) .. controls (3,0.8) and \control{(5.5,-2.5)}{(90:1.5)};
\draw [purple] (-3,-2.5) .. controls +(90:1.5) and (-1.5,-0.325) .. (0,0.05) -- (1,0.3) .. controls (3,0.8) and \control{(5.5,-2.5)}{(90:1.5)};
\fill [blue!20!white] (0,0) -- (1,-0.75) -- (2, -0.5) -- (1,0.25) -- cycle;
\draw (0,0) -- (1,0.25) node (n1) [pos = 0.3, inner sep = 0] {} (1,-0.75) -- (2, -0.5) node (n2) [pos = 0.3, inner sep = 0] {};

%Normals
\node at (1, -2) [purple] {$S$};
\node at (1, 0) [below right] {$D$};
\draw [green!60!black] (n1) -- (n2);
\node at (0.8, -0.05) [green!60!black] {$\alpha$};
\draw [blue!80!black, ->] (0.2,-0.1) -- ++(0,0.75);
\draw [purple!80!black, ->] (-2.25,-1.25) -- ++(0,1);
\node [blue!80!white] at (-1.25,-1) {$\Sigma$};

\begin{scope}[yshift = 1.4cm]
%Coordinate
\coordinate (h) at (1,1.95);
\coordinate (l) at (1, 1.55); 
\coordinate (i) at (0.5, 1.65);
\coordinate (e) at (1.5, 1.9);

%Filling
\begin{scope}
\clip (h) .. controls +(180:0.2) and \control{(0,2)}{(-36.75:0.2)} .. controls +(143.25:0.1) and \control{(-0.25,2.15)}{(0:0.1)} .. controls +(180:0.35) and \control{(0,0)}{(143.25:1)} -- (1,-0.75) .. controls +(-36.75:0.1) and \control{(1.25,-0.9)}{(180:0.1)} .. controls +(0:0.35) and \control{(1,1.25)}{(-36.75:1)} -- cycle;
\fill [blue!20!white] (1,2.25) .. controls +(143.25:0.1) and \control{(0.75,2.4)}{(0:0.1)} .. controls +(180:0.35) and \control{(1,0.25)}{(143.25:1)} -- (2,-0.5) .. controls +(-36.75:0.1) and \control{(2.25,-0.65)}{(180:0.1)} -- ++(0, -1) -- ++(-3.5,0) -- ++(0,4.5) -- ++(2.5,0) -- cycle;
\end{scope}
\fill [blue!20!white] (h) .. controls +(-150:0.1) and \control{(i)}{(143.25:0.2)} .. controls +(-36.75:0.1) and \control{(l)}{(180:0.2)} .. controls +(60:0.2) and \control{(e)}{(-36.75:0.2)} .. controls +(143.25:0.1) and \control{(h)}{(0:0.2)};
\fill [white!40!gray] (-0.25,2.15) .. controls +(0:0.1) and \control{(0,2)}{(143.25:0.1)} .. controls +(-36.75:0.2) and \control{(h)}{(180:0.2)} .. controls +(30:0.1) and \control{(1,2.25)}{(-36.75:0.2)} .. controls +(143.25:0.1) and \control{(0.75,2.4)}{(0:0.1)}-- cycle;
\fill [white!40!gray] (1,1.25) .. controls +(143.25:0.1) and \control{(l)}{(-130:0.2)} .. controls +(180:0.2) and \control{(2,1.5)}{(143.25:0.2)} .. controls +(-36.75:1) and \control{(2.25,-0.65)}{(0:0.35)} -- (1.25,-0.9) .. controls +(0:0.35) and \control{(1,1.25)}{(-36.75:1)};
\node at (1.85,0.75) {$H$};

%Crossings
\begin{scope}
\clip (-0.25,2.15) .. controls +(0:0.1) and \control{(0,2)}{(143.25:0.1)} .. controls +(-36.75:0.2) and \control{(h)}{(180:0.2)} .. controls +(30:0.1) and \control{(1,2.25)}{(-36.75:0.2)} .. controls +(143.25:0.1) and \control{(0.75,2.4)}{(0:0.1)}-- cycle;
\draw [dashed] (0.75,2.4) .. controls +(180:0.35) and \control{(1,0.25)}{(143.25:1)};
\end{scope}
\begin{scope}
\clip (1,1.25) .. controls +(143.25:0.1) and \control{(l)}{(-130:0.2)} .. controls +(180:0.2) and \control{(2,1.5)}{(143.25:0.2)} .. controls +(-36.75:1) and \control{(2.25,-0.65)}{(0:0.35)} -- (1.25,-0.9) .. controls +(0:0.35) and \control{(1,1.25)}{(-36.75:1)};
\draw [dashed] (1,0.25) -- (2,-0.5) .. controls +(-36.75:0.1) and \control{(2.25,-0.65)}{(180:0.1)};
\end{scope}
\begin{scope}
\clip (h) .. controls +(180:0.2) and \control{(0,2)}{(-36.75:0.2)} .. controls +(143.25:0.1) and \control{(-0.25,2.15)}{(0:0.1)} .. controls +(180:0.35) and \control{(0,0)}{(143.25:1)} -- (1,-0.75) .. controls +(-36.75:0.1) and \control{(1.25,-0.9)}{(180:0.1)} .. controls +(0:0.35) and \control{(1,1.25)}{(-36.75:1)} -- cycle;
\draw (0.75,2.4) .. controls +(180:0.35) and \control{(1,0.25)}{(143.25:1)} -- (2,-0.5);
\end{scope}
\draw (1,2.25) .. controls +(143.25:0.1) and \control{(0.75,2.4)}{(0:0.1)};
\draw (2.25,-0.65) .. controls +(0:0.35) and \control{(2,1.5)}{(-36.75:1)};
\draw (0,2) .. controls +(143.25:0.1) and \control{(-0.25,2.15)}{(0:0.1)} .. controls +(180:0.35) and \control{(0,0)}{(143.25:1)} -- (1,-0.75) .. controls +(-36.75:0.1) and \control{(1.25,-0.9)}{(180:0.1)} .. controls +(0:0.35) and \control{(1,1.25)}{(-36.75:1)} ;
\begin{scope}[even odd rule] 
\clip (0, 2.25) rectangle (3, 1) (l) circle (0.1cm);
\draw (1,1.25) .. controls +(143.25:0.1) and \control{(l)}{(-130:0.2)} .. controls +(60:0.2) and \control{(e)}{(-36.75:0.2)};
\end{scope}
\draw (e) .. controls +(143.25:0.1) and \control{(h)}{(0:0.2)} .. controls +(180:0.2) and \control{(0,2)}{(-36.75:0.2)};
\begin{scope}[even odd rule] 
\clip (0,2) -- (1,1.25) -- (2,1.5) -- (1,2.25) -- (0,2) (h) circle (0.13cm);
\draw (1,2.25) .. controls +(-36.75:0.2) and \control{(h)}{(30:0.1)} .. controls +(-150:0.1) and \control{(i)}{(143.25:0.2)};
\end{scope}
\draw (i) .. controls +(-36.75:0.1) and \control{(l)}{(180:0.2)} .. controls +(0:0.2) and \control{(2,1.5)}{(143.25:0.2)};

\draw [dotted] (0,0) -- (1,-0.75) -- (2, -0.5) -- (1,0.25) -- cycle;
\draw [dotted, yshift = -1.4cm] (0,0) -- (1,-0.75) -- (2, -0.5) -- (1,0.25) -- cycle;
\node at (1, 0.1) [below] {$D'$};
\draw [blue!80!black, ->] (1.05,0) -- ++(0,0.75);
\end{scope}

\begin{scope}[xshift = 1cm]
\draw [->] (3.5, 0.5) -- (6.5, 0.5);
\node at (5, 0.25) [label={[align=center]Hopf \\ Plumbing}] {};
\end{scope}

\begin{scope}[xshift = 9cm]
%Band filling
\fill [blue!20!white] (-2,-1) .. controls +(70:0.5) and \control{(0,0)}{(-166:0.5)} -- (1,0.25) .. controls +(14:1) and \control{(3.5,-0.5)}{(140:0.5)} -- (4.5,-1.25) .. controls +(140:0.5) and \control{(2,-0.5)}{(14:1)} -- (1,-0.75) .. controls +(-166:0.5) and \control{(-1,-1.75)}{(70:0.5)} -- cycle;
\draw (-2,-1) .. controls +(70:0.5) and \control{(0,0)}{(-166:0.5)} -- (1,0.25) .. controls +(14:1) and \control{(3.5,-0.5)}{(140:0.5)}; 
\draw [xshift = 1cm, yshift = -0.75cm] (-2,-1) .. controls +(70:0.5) and \control{(0,0)}{(-166:0.5)} -- (1,0.25) .. controls +(14:1) and \control{(3.5,-0.5)}{(140:0.5)}; 
\draw [dotted] (-2,-1) -- ++(0,-0.75);
\draw [dotted, xshift = 1cm, yshift = -0.75cm] (-2,-1) -- ++(0,-0.75);
\draw [dotted] (3.5,-0.5) .. controls +(-40:0.2) and \control{(3.75,-1)}{(90:0.2)};
\draw [dotted, xshift = 1cm, yshift = -0.75cm] (3.5,-0.5) .. controls +(-40:0.2) and \control{(3.75,-1)}{(90:0.2)}; 

%Sphere
\fill [opacity = 0.2, purple] (-3,-2.5) .. controls +(90:1.5) and (-1.5,-0.325) .. (0,0.05) -- (1,0.3) .. controls (3,0.8) and \control{(5.5,-2.5)}{(90:1.5)};
\draw [purple] (-3,-2.5) .. controls +(90:1.5) and (-1.5,-0.325) .. (0,0.05) -- (1,0.3) .. controls (3,0.8) and \control{(5.5,-2.5)}{(90:1.5)};
\draw (0,0) -- (1,0.25) node (n1) [pos = 0.3, inner sep = 0] {} (1,-0.75) -- (2, -0.5) node (n2) [pos = 0.3, inner sep = 0] {};
\node at (1, -2) [purple] {$S$};
%\node [blue!80!white] at (-1.25,-1) {$\Sigma$};

%Coordinate
\coordinate (h) at (1,1.95);
\coordinate (l) at (1, 1.55); 
\coordinate (i) at (0.5, 1.65);
\coordinate (e) at (1.5, 1.9);

%Filling
\begin{scope}
\clip (h) .. controls +(180:0.2) and \control{(0,2)}{(-36.75:0.2)} .. controls +(143.25:0.1) and \control{(-0.25,2.15)}{(0:0.1)} .. controls +(180:0.35) and \control{(0,0)}{(143.25:1)} -- (1,-0.75) .. controls +(-36.75:0.1) and \control{(1.25,-0.9)}{(180:0.1)} .. controls +(0:0.35) and \control{(1,1.25)}{(-36.75:1)} -- cycle;
\fill [blue!20!white] (1,2.25) .. controls +(143.25:0.1) and \control{(0.75,2.4)}{(0:0.1)} .. controls +(180:0.35) and \control{(1,0.25)}{(143.25:1)} -- (2,-0.5) .. controls +(-36.75:0.1) and \control{(2.25,-0.65)}{(180:0.1)} -- ++(0, -1) -- ++(-3.5,0) -- ++(0,4.5) -- ++(2.5,0) -- cycle;
\end{scope}
\fill [blue!20!white] (h) .. controls +(-150:0.1) and \control{(i)}{(143.25:0.2)} .. controls +(-36.75:0.1) and \control{(l)}{(180:0.2)} .. controls +(60:0.2) and \control{(e)}{(-36.75:0.2)} .. controls +(143.25:0.1) and \control{(h)}{(0:0.2)};
\fill [white!40!gray] (-0.25,2.15) .. controls +(0:0.1) and \control{(0,2)}{(143.25:0.1)} .. controls +(-36.75:0.2) and \control{(h)}{(180:0.2)} .. controls +(30:0.1) and \control{(1,2.25)}{(-36.75:0.2)} .. controls +(143.25:0.1) and \control{(0.75,2.4)}{(0:0.1)}-- cycle;
\fill [white!40!gray] (1,1.25) .. controls +(143.25:0.1) and \control{(l)}{(-130:0.2)} .. controls +(180:0.2) and \control{(2,1.5)}{(143.25:0.2)} .. controls +(-36.75:1) and \control{(2.25,-0.65)}{(0:0.35)} -- (1.25,-0.9) .. controls +(0:0.35) and \control{(1,1.25)}{(-36.75:1)};

%Crossings
\begin{scope}
\clip (-0.25,2.15) .. controls +(0:0.1) and \control{(0,2)}{(143.25:0.1)} .. controls +(-36.75:0.2) and \control{(h)}{(180:0.2)} .. controls +(30:0.1) and \control{(1,2.25)}{(-36.75:0.2)} .. controls +(143.25:0.1) and \control{(0.75,2.4)}{(0:0.1)}-- cycle;
\draw [dashed] (0.75,2.4) .. controls +(180:0.35) and \control{(1,0.25)}{(143.25:1)};
\end{scope}
\begin{scope}
\clip (1,1.25) .. controls +(143.25:0.1) and \control{(l)}{(-130:0.2)} .. controls +(180:0.2) and \control{(2,1.5)}{(143.25:0.2)} .. controls +(-36.75:1) and \control{(2.25,-0.65)}{(0:0.35)} -- (1.25,-0.9) .. controls +(0:0.35) and \control{(1,1.25)}{(-36.75:1)};
\draw [dashed] (4.5, -1.25) .. controls +(140:0.5) and \control{(2,-0.5)}{(14:1)} .. controls +(-36.75:0.1) and \control{(2.25,-0.65)}{(180:0.1)};
\end{scope}
\begin{scope}
\clip (h) .. controls +(180:0.2) and \control{(0,2)}{(-36.75:0.2)} .. controls +(143.25:0.1) and \control{(-0.25,2.15)}{(0:0.1)} .. controls +(180:0.35) and \control{(0,0)}{(143.25:1)} -- (1,-0.75) .. controls +(-36.75:0.1) and \control{(1.25,-0.9)}{(180:0.1)} .. controls +(0:0.35) and \control{(1,1.25)}{(-36.75:1)} -- cycle;
\draw (0.75,2.4) .. controls +(180:0.35) and \control{(1,0.25)}{(143.25:1)};
\end{scope}
\draw (1,2.25) .. controls +(143.25:0.1) and \control{(0.75,2.4)}{(0:0.1)};
\draw (2.25,-0.65) .. controls +(0:0.35) and \control{(2,1.5)}{(-36.75:1)};
\draw (0,2) .. controls +(143.25:0.1) and \control{(-0.25,2.15)}{(0:0.1)} .. controls +(180:0.35) and \control{(0,0)}{(143.25:1)} (1,-0.75) .. controls +(-36.75:0.1) and \control{(1.25,-0.9)}{(180:0.1)} .. controls +(0:0.35) and \control{(1,1.25)}{(-36.75:1)} ;
\begin{scope}[even odd rule] 
\clip (0, 2.25) rectangle (3, 1) (l) circle (0.1cm);
\draw (1,1.25) .. controls +(143.25:0.1) and \control{(l)}{(-130:0.2)} .. controls +(60:0.2) and \control{(e)}{(-36.75:0.2)};
\end{scope}
\draw (e) .. controls +(143.25:0.1) and \control{(h)}{(0:0.2)} .. controls +(180:0.2) and \control{(0,2)}{(-36.75:0.2)};
\begin{scope}[even odd rule] 
\clip (0,2) -- (1,1.25) -- (2,1.5) -- (1,2.25) -- (0,2) (h) circle (0.13cm);
\draw (1,2.25) .. controls +(-36.75:0.2) and \control{(h)}{(30:0.1)} .. controls +(-150:0.1) and \control{(i)}{(143.25:0.2)};
\end{scope}
\draw (i) .. controls +(-36.75:0.1) and \control{(l)}{(180:0.2)} .. controls +(0:0.2) and \control{(2,1.5)}{(143.25:0.2)};
%\node at (0,1) [blue] {$v$};
\end{scope}
\end{tikzpicture}
\caption{A Hopf plumbing of a Hopf band $H$ on top of a Seifert surface $\Sigma$ along $\alpha$.}
\label{pic_def_plumbing_sphere}
\end{center}
\end{figure}

Hopf plumbing is a special case of a more general operation called a Murasugi sum, see~\cite{Murasagi_def_sum,Ozbagci_plumbing_history}. 
A key property of Murasugi sums, proved by Gabai~\cite{Gabai_Murasugi_sum}, is that it preserves fibredness. 
In the above setting, since Hopf bands are fibred, if $\Sigma$ is fibred, then the surface obtained from $\Sigma$ by Hopf plumbing $H$ on~$\Sigma$ along any arc is also fibred.

\subsection{From plane trees to Hopf arborescent surfaces and links}\label{S:HopfArborescent}

Recall that in this article, a plane tree is a rooted tree that is embedded in the plane and such that every vertex has a label~$+$ or~$-$. 
If $v$ is a vertex, we denote by~$\ell(v)$ its label. 
Let~$T$ be a plane tree. 
The associated surface $\Sigma(T)$ we construct is an oriented surface with boundary that retracts on the union of a finite set of oriented simple curves~$\mathcal{C}_T$ parametrized by the vertices of~$T$, such that every $\alpha \in \mathcal{C}_T$ is the core of a Hopf band embedded on $\Sigma$ and whose sign is the label of the corresponding vertex in~$T$. 
For a vertex~$v$ in~$T$, the curve~$\alpha(v)$ intersects another curve~$\alpha(v')$ if and only $vv'$ is an edge of~$T$, and the two curves intersect exactly once. 
Moreover, following~$\alpha(v)$ with its given orientation, the cyclic ordering of the intersection points with the curves~$\alpha(v')$ coincides with the cyclic orderings of the neighbours of~$v$ in the plane tree~$T$. 

\begin{figure}[ht]
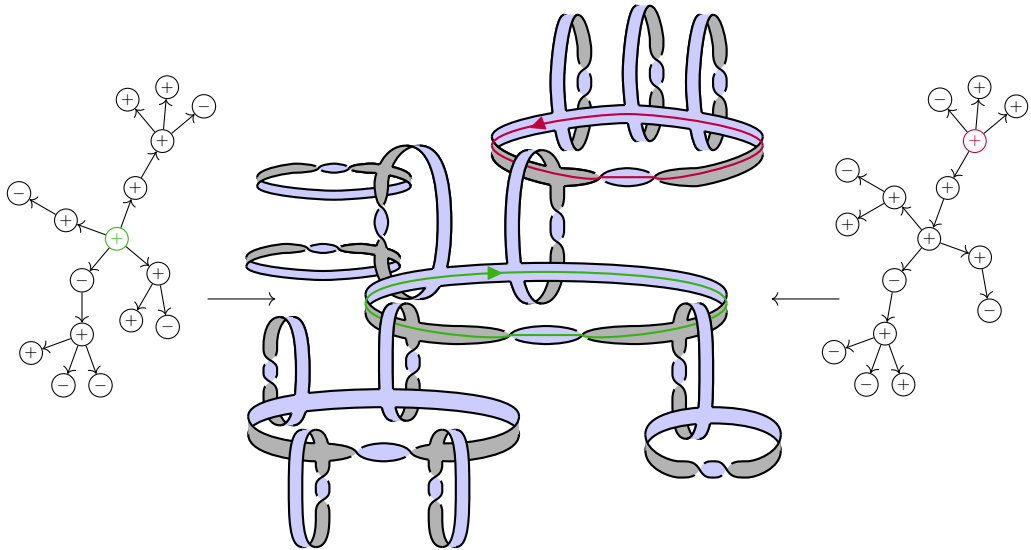

\begin{center}

\caption{A 3D-view of a Hopf arborescent link and its construction from two different plane trees. The chosen orientation of the root of each tree is indicated on the coloured core of the matching Hopf band. The orientation of the plane is counter-clockwise.}
\label{pic_3D_arborescent_link}
\end{center}
\end{figure}

We now describe the construction inductively, see Figure~\ref{pic_3D_arborescent_link} for an illustration.

\begin{enumerate}
\item Start from a Hopf band $H(v_r)$ where $v_r$ is the root of~$T$, and whose sign is the label~$\ell(v_r)$. 
\item For the induction step, assume that the tree $T'$ is obtained from $T$ by adding at a leaf~$v$ a finite number of leaves $v_1, \dots, v_k$ appearing in the plane in this order around $v$, and that the surface~$\Sigma(T)$ is already constructed with its set of core curves~$\mathcal{C}_T$.
  \begin{enumerate}
\item Since $v$ is a leaf in~$T$, the curve~$\alpha(v)$ intersects only one curve~$\alpha(v')$: the curve associated to $v'$, the parent of~$v$ in~$T$.
\item Starting from this intersection point we place $k$ points $p_1, \dots, p_k$ on~$\alpha(v)$ in this order. Then we draw on~$\Sigma(T)$ a family of $k$ arcs $\beta_1, \dots, \beta_k$ from~$\partial S$ to itself that correspond to those arcs that retracts on~$p_1, \dots, p_k$. 
Each such arc~$\beta_i$ intersects the collection~$\mathcal{C}_T$ exactly at the point~$p_i$. 
\item For $i=1, \dots, k$, perform the Hopf plumbing of a Hopf band~$H(v_i)$ of sign~$\ell(v_i)$ on top of~$\Sigma(T)$ along the arc~$\beta_i$. 
The resulting surface is~$\Sigma (T')$. 
\item\label{step:orient} Finally for every~$i$ orient the core of~$H(v_i)$ so that when going from $\alpha(v)$ to~$\alpha(v_i)$, we follow this rule: if $\ell(v)$ is positive, one turns to the left (with respect to the orientation of $\Sigma (T)$), and if $\ell(v)$ is negative, one turns to the right (once again with respect to the orientation of $\Sigma (T)$), see Figure~\ref{pic_orient_plumbing}.
  The set $\mathcal{C}_{T'}$ is the union of $\mathcal{C}_T$ with $\alpha(v_1), \dots, \alpha(v_k)$.
  \end{enumerate}
\end{enumerate}

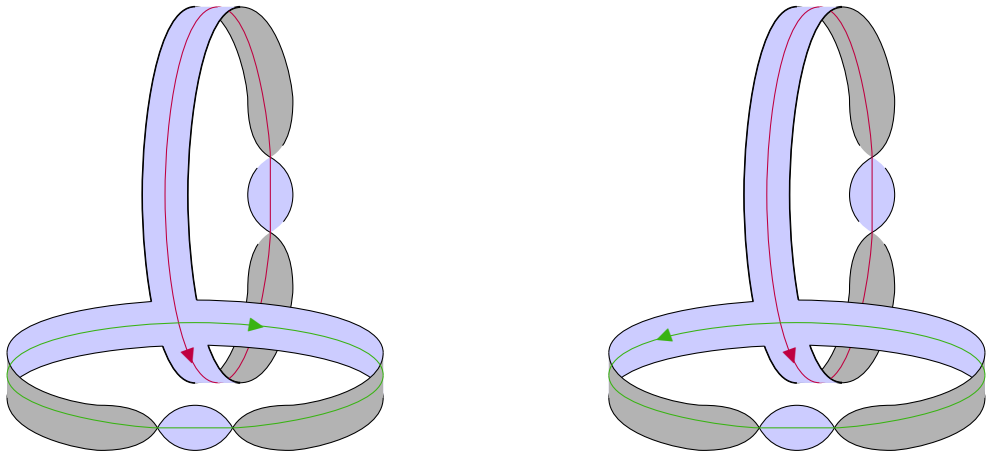
\begin{figure}[H]
\begin{center}
\begin{tikzpicture}[scale = 2]
\clip (-5.3,-0.5) rectangle (1.3,2.6);
%Band 1
%Param
\def\i{1}
\def\l{2.5}
\def\L{0.3}
\def\esp{0.35}
\def\e{0.25}
\def\c{0.15}
\begin{scope}
%Coordinate
\coordinate (bl\i) at (0,0);
\coordinate (tl\i) at ($(bl\i)+(0,\l)$);
\coordinate (mnl\i) at ($(bl\i)+0.5*(0,\l)+(-\esp,0)$);
\coordinate (br\i) at (\L,0);
\coordinate (tr\i) at ($(br\i)+(0,\l)$);
\coordinate (mnr\i) at ($(mnl\i)+(\L,0)$);
\coordinate (ml\i) at ($(bl\i)+0.5*(0,\l)+(\esp,0)$);
\coordinate (qbl\i) at ($(bl\i)+0.25*(0,\l)+(\esp,0)$);
\coordinate (qtl\i) at ($(bl\i)+0.75*(0,\l)+(\esp,0)$);
\coordinate (qbr\i) at ($(qbl\i)+(\L,0)$);
\coordinate (mr\i) at ($(ml\i)+(\L,0)$);
\coordinate (qtr\i) at ($(qtl\i)+(\L,0)$);
\coordinate (ql\i) at ($(bl\i)+0.4*(0,\l)+0.5*(\L,0)+(\esp,0)$);
\coordinate (qh\i) at ($(bl\i)+0.6*(0,\l)+0.5*(\L,0)+(\esp,0)$);
\end{scope}

%Crossing part
\fill [white!40!gray] (bl\i) .. controls +(0:\e) and \control{(qbl\i)}{(-90:0.5*\e)} .. controls +(90:0.5*\e) and \control{(ql\i)}{(-150:\c)} .. controls +(-30:\c) and \control{(qbr\i)}{(90:0.5*\e)}  .. controls +(-90:0.5*\e) and \control{(br\i)}{(0:\e)} -- cycle;
\fill [white!40!gray] (tl\i) .. controls +(0:\e) and \control{(qtl\i)}{(90:0.5*\e)} .. controls +(-90:0.5*\e) and \control{(qh\i)}{(150:\c)} .. controls +(30:\c) and \control{(qtr\i)}{(-90:0.5*\e)}  .. controls +(90:0.5*\e) and \control{(tr\i)}{(0:\e)} -- cycle;
\fill [blue!20!white] (ql\i) .. controls +(150:\c) and \control{(ml\i)}{(-90:0.25*\e)} .. controls +(90:0.25*\e) and \control{(qh\i)}{(-150:\c)} .. controls +(-30:\c) and \control{(mr\i)}{(90:0.25*\e)} .. controls +(-90:0.25*\e) and \control{(ql\i)}{(30:\c)};
\begin{scope}[even odd rule]
\clip (-3,-3) rectangle (3,3) (qh\i) circle (0.75*\c cm); 
\draw (bl\i) .. controls +(0:\e) and \control{(qbl\i)}{(-90:0.5*\e)} .. controls +(90:0.5*\e) and \control{(ql\i)}{(-150:\c)} .. controls +(30:\c) and \control{(mr\i)}{(-90:0.25*\e)} .. controls +(90:0.25*\e) and \control{(qh\i)}{(-30:\c)} .. controls +(150:\c) and \control{(qtl\i)}{(-90:0.5*\e)} .. controls +(90:0.5*\e) and \control{(tl\i)}{(0:\e)};
\end{scope}
\begin{scope}[even odd rule]
\clip (-3,-3) rectangle (3,3) (ql\i) circle (0.75*\c cm); 
\draw (br\i) .. controls +(0:\e) and \control{(qbr\i)}{(-90:0.5*\e)} .. controls +(90:0.5*\e) and \control{(ql\i)}{(-30:\c)} .. controls +(150:\c) and \control{(ml\i)}{(-90:0.25*\e)} .. controls +(90:0.25*\e) and \control{(qh\i)}{(-150:\c)} .. controls +(30:\c) and \control{(qtr\i)}{(-90:0.5*\e)} .. controls +(90:0.5*\e) and \control{(tr\i)}{(0:\e)};
\end{scope}

%non crossing
\fill [blue!20!white] (bl\i) .. controls +(180:\e) and \control{(mnl\i)}{(-90:2*\e)} .. controls +(90:2*\e) and \control{(tl\i)}{(180:\e)} -- (tr\i) .. controls +(180:\e) and \control{(mnr\i)}{(90:2*\e)} .. controls +(-90:2*\e) and \control{(br\i)}{(180:\e)};
\draw [thick] (bl\i) .. controls +(180:\e) and \control{(mnl\i)}{(-90:2*\e)} .. controls +(90:2*\e) and \control{(tl\i)}{(180:\e)}; 
\draw [thick] (br\i) .. controls +(180:\e) and \control{(mnr\i)}{(-90:2*\e)} .. controls +(90:2*\e) and \control{(tr\i)}{(180:\e)}; 

%Core 
\draw [purple] ($(bl\i)!0.5!(br\i)$) .. controls +(0:\e) and \control{(ql\i)}{(-90:\c)} -- (qh\i) .. controls +(90:\c) and \control{($(tl\i)!0.5!(tr\i)$)}{(0:\e)};

%Band 2
%Param
\def\i{2}
\begin{scope}[xshift = -1.25cm, yshift = 0.2cm, rotate = -90]
%Coordinate
\coordinate (bl\i) at (0,0);
\coordinate (tl\i) at ($(bl\i)+(0,\l)$);
\coordinate (mnl\i) at ($(bl\i)+0.5*(0,\l)+(-\esp,0)$);
\coordinate (br\i) at (\L,0);
\coordinate (tr\i) at ($(br\i)+(0,\l)$);
\coordinate (mnr\i) at ($(mnl\i)+(\L,0)$);
\coordinate (ml\i) at ($(bl\i)+0.5*(0,\l)+(\esp,0)$);
\coordinate (qbl\i) at ($(bl\i)+0.25*(0,\l)+(\esp,0)$);
\coordinate (qtl\i) at ($(bl\i)+0.75*(0,\l)+(\esp,0)$);
\coordinate (qbr\i) at ($(qbl\i)+(\L,0)$);
\coordinate (mr\i) at ($(ml\i)+(\L,0)$);
\coordinate (qtr\i) at ($(qtl\i)+(\L,0)$);
\coordinate (ql\i) at ($(bl\i)+0.4*(0,\l)+0.5*(\L,0)+(\esp,0)$);
\coordinate (qh\i) at ($(bl\i)+0.6*(0,\l)+0.5*(\L,0)+(\esp,0)$);

%non crossing
\fill [blue!20!white] (bl\i) .. controls +(180:\e) and \control{(mnl\i)}{(-90:2*\e)} .. controls +(90:2*\e) and \control{(tl\i)}{(180:\e)} -- (tr\i) .. controls +(180:\e) and \control{(mnr\i)}{(90:2*\e)} .. controls +(-90:2*\e) and \control{(br\i)}{(180:\e)};
\draw (bl\i) .. controls +(180:\e) and \control{(mnl\i)}{(-90:2*\e)} .. controls +(90:2*\e) and \control{(tl\i)}{(180:\e)}; 
\draw (br\i) .. controls +(180:\e) and \control{(mnr\i)}{(-90:2*\e)} .. controls +(90:2*\e) and \control{(tr\i)}{(180:\e)}; 

%Crossing part
\fill [white!40!gray] (bl\i) .. controls +(0:\e) and \control{(qbl\i)}{(-90:0.5*\e)} .. controls +(90:0.5*\e) and \control{(ql\i)}{(-150:\c)} .. controls +(-30:\c) and \control{(qbr\i)}{(90:0.5*\e)}  .. controls +(-90:0.5*\e) and \control{(br\i)}{(0:\e)} -- cycle;
\fill [white!40!gray] (tl\i) .. controls +(0:\e) and \control{(qtl\i)}{(90:0.5*\e)} .. controls +(-90:0.5*\e) and \control{(qh\i)}{(150:\c)} .. controls +(30:\c) and \control{(qtr\i)}{(-90:0.5*\e)}  .. controls +(90:0.5*\e) and \control{(tr\i)}{(0:\e)} -- cycle;
\fill [blue!20!white] (ql\i) .. controls +(150:\c) and \control{(ml\i)}{(-90:0.25*\e)} .. controls +(90:0.25*\e) and \control{(qh\i)}{(-150:\c)} .. controls +(-30:\c) and \control{(mr\i)}{(90:0.25*\e)} .. controls +(-90:0.25*\e) and \control{(ql\i)}{(30:\c)};
\begin{scope}[even odd rule]
%\clip (-3,-3) rectangle (3,5) (ql\i) circle (0.75*\c cm) (qh\i) circle (0.75*\c cm); 
\draw (bl\i) .. controls +(0:\e) and \control{(qbl\i)}{(-90:0.5*\e)} .. controls +(90:0.5*\e) and \control{(ql\i)}{(-150:\c)} .. controls +(30:\c) and \control{(mr\i)}{(-90:0.25*\e)} .. controls +(90:0.25*\e) and \control{(qh\i)}{(-30:\c)} .. controls +(150:\c) and \control{(qtl\i)}{(-90:0.5*\e)} .. controls +(90:0.5*\e) and \control{(tl\i)}{(0:\e)};
\end{scope}
\begin{scope}[even odd rule]
%\clip (-3,-3) rectangle (3,5) (qh\i) circle (0.75*\c cm) (ql\i) circle (0.75*\c cm); 
\draw (br\i) .. controls +(0:\e) and \control{(qbr\i)}{(-90:0.5*\e)} .. controls +(90:0.5*\e) and \control{(ql\i)}{(-30:\c)} .. controls +(150:\c) and \control{(ml\i)}{(-90:0.25*\e)} .. controls +(90:0.25*\e) and \control{(qh\i)}{(-150:\c)} .. controls +(30:\c) and \control{(qtr\i)}{(-90:0.5*\e)} .. controls +(90:0.5*\e) and \control{(tr\i)}{(0:\e)};
\end{scope}
\end{scope}

\def\i{1}
\fill [blue!20!white] (bl\i) .. controls +(180:\e) and \control{(mnl\i)}{(-90:2*\e)} .. controls +(90:2*\e) and \control{(tl\i)}{(180:\e)} -- (tr\i) .. controls +(180:\e) and \control{(mnr\i)}{(90:2*\e)} .. controls +(-90:2*\e) and \control{(br\i)}{(180:\e)};
\begin{scope}[even odd rule]
\def\i{2}
\clip (-1,3) rectangle (1,-1) [rotate=-90] (bl\i) .. controls +(180:\e) and \control{(mnl\i)}{(-90:2*\e)} .. controls +(90:2*\e) and \control{(tl\i)}{(180:\e)} -- (tr\i) .. controls +(180:\e) and \control{(mnr\i)}{(90:2*\e)} .. controls +(-90:2*\e) and \control{(br\i)}{(180:\e)};
\def\i{1}
\draw (bl\i) .. controls +(180:\e) and \control{(mnl\i)}{(-90:2*\e)} .. controls +(90:2*\e) and \control{(tl\i)}{(180:\e)}; 
\draw (br\i) .. controls +(180:\e) and \control{(mnr\i)}{(-90:2*\e)} .. controls +(90:2*\e) and \control{(tr\i)}{(180:\e)}; 
\end{scope}
\draw [purple] ($(bl\i)!0.5!(br\i)$) .. controls +(180:\e) and \control{($(mnl\i)!0.5!(mnr\i)$)}{(-90:2*\e)}  node (nc1) [pos = 0.3] {} .. controls +(90:2*\e) and \control{($(tl\i)!0.5!(tr\i)$)}{(180:\e)};
\node at (nc1) [rotate = 110, purple] {{\small $\blacktriangleleft$}};

%Core
\def\i{2}
\draw [green!70!purple, xshift = -1.6 cm, yshift = 0.05 cm, rotate = -90] ($(bl\i)!0.5!(br\i)$) .. controls +(180:\e) and \control{($(mnl\i)!0.5!(mnr\i)$)}{(-90:2*\e)}  node (nc2) [pos = 0.45] {} .. controls +(90:2*\e) and \control{($(tl\i)!0.5!(tr\i)$)}{(180:\e)};
\node at (nc2) [rotate = 20, green!70!purple] {{\small $\blacktriangleleft$}};
\draw [green!70!purple, xshift = -1.6 cm, yshift = 0.05 cm, rotate = -90] ($(bl\i)!0.5!(br\i)$) .. controls +(0:\e) and \control{(ql\i)}{(-90:\c)} -- (qh\i) .. controls +(90:\c) and \control{($(tl\i)!0.5!(tr\i)$)}{(0:\e)};

\begin{scope}[xshift = -4cm]
%Band 1
%Param
\def\i{1}
\def\l{2.5}
\def\L{0.3}
\def\esp{0.35}
\def\e{0.25}
\def\c{0.15}
\begin{scope}
%Coordinate
\coordinate (bl\i) at (0,0);
\coordinate (tl\i) at ($(bl\i)+(0,\l)$);
\coordinate (mnl\i) at ($(bl\i)+0.5*(0,\l)+(-\esp,0)$);
\coordinate (br\i) at (\L,0);
\coordinate (tr\i) at ($(br\i)+(0,\l)$);
\coordinate (mnr\i) at ($(mnl\i)+(\L,0)$);
\coordinate (ml\i) at ($(bl\i)+0.5*(0,\l)+(\esp,0)$);
\coordinate (qbl\i) at ($(bl\i)+0.25*(0,\l)+(\esp,0)$);
\coordinate (qtl\i) at ($(bl\i)+0.75*(0,\l)+(\esp,0)$);
\coordinate (qbr\i) at ($(qbl\i)+(\L,0)$);
\coordinate (mr\i) at ($(ml\i)+(\L,0)$);
\coordinate (qtr\i) at ($(qtl\i)+(\L,0)$);
\coordinate (ql\i) at ($(bl\i)+0.4*(0,\l)+0.5*(\L,0)+(\esp,0)$);
\coordinate (qh\i) at ($(bl\i)+0.6*(0,\l)+0.5*(\L,0)+(\esp,0)$);
\end{scope}

%Crossing part
\fill [white!40!gray] (bl\i) .. controls +(0:\e) and \control{(qbl\i)}{(-90:0.5*\e)} .. controls +(90:0.5*\e) and \control{(ql\i)}{(-150:\c)} .. controls +(-30:\c) and \control{(qbr\i)}{(90:0.5*\e)}  .. controls +(-90:0.5*\e) and \control{(br\i)}{(0:\e)} -- cycle;
\fill [white!40!gray] (tl\i) .. controls +(0:\e) and \control{(qtl\i)}{(90:0.5*\e)} .. controls +(-90:0.5*\e) and \control{(qh\i)}{(150:\c)} .. controls +(30:\c) and \control{(qtr\i)}{(-90:0.5*\e)}  .. controls +(90:0.5*\e) and \control{(tr\i)}{(0:\e)} -- cycle;
\fill [blue!20!white] (ql\i) .. controls +(150:\c) and \control{(ml\i)}{(-90:0.25*\e)} .. controls +(90:0.25*\e) and \control{(qh\i)}{(-150:\c)} .. controls +(-30:\c) and \control{(mr\i)}{(90:0.25*\e)} .. controls +(-90:0.25*\e) and \control{(ql\i)}{(30:\c)};
\begin{scope}[even odd rule]
\clip (-3,-3) rectangle (3,3) (ql\i) circle (0.75*\c cm); 
\draw (bl\i) .. controls +(0:\e) and \control{(qbl\i)}{(-90:0.5*\e)} .. controls +(90:0.5*\e) and \control{(ql\i)}{(-150:\c)} .. controls +(30:\c) and \control{(mr\i)}{(-90:0.25*\e)} .. controls +(90:0.25*\e) and \control{(qh\i)}{(-30:\c)} .. controls +(150:\c) and \control{(qtl\i)}{(-90:0.5*\e)} .. controls +(90:0.5*\e) and \control{(tl\i)}{(0:\e)};
\end{scope}
\begin{scope}[even odd rule]
\clip (-3,-3) rectangle (3,3) (qh\i) circle (0.75*\c cm); 
\draw (br\i) .. controls +(0:\e) and \control{(qbr\i)}{(-90:0.5*\e)} .. controls +(90:0.5*\e) and \control{(ql\i)}{(-30:\c)} .. controls +(150:\c) and \control{(ml\i)}{(-90:0.25*\e)} .. controls +(90:0.25*\e) and \control{(qh\i)}{(-150:\c)} .. controls +(30:\c) and \control{(qtr\i)}{(-90:0.5*\e)} .. controls +(90:0.5*\e) and \control{(tr\i)}{(0:\e)};
\end{scope}

%non crossing
\fill [blue!20!white] (bl\i) .. controls +(180:\e) and \control{(mnl\i)}{(-90:2*\e)} .. controls +(90:2*\e) and \control{(tl\i)}{(180:\e)} -- (tr\i) .. controls +(180:\e) and \control{(mnr\i)}{(90:2*\e)} .. controls +(-90:2*\e) and \control{(br\i)}{(180:\e)};
\draw [thick] (bl\i) .. controls +(180:\e) and \control{(mnl\i)}{(-90:2*\e)} .. controls +(90:2*\e) and \control{(tl\i)}{(180:\e)}; 
\draw [thick] (br\i) .. controls +(180:\e) and \control{(mnr\i)}{(-90:2*\e)} .. controls +(90:2*\e) and \control{(tr\i)}{(180:\e)}; 

%Core 
\draw [purple] ($(bl\i)!0.5!(br\i)$) .. controls +(0:\e) and \control{(ql\i)}{(-90:\c)} -- (qh\i) .. controls +(90:\c) and \control{($(tl\i)!0.5!(tr\i)$)}{(0:\e)};

%Band 2
%Param
\def\i{2}
\begin{scope}[xshift = -1.25cm, yshift = 0.2cm, rotate = -90]
%Coordinate
\coordinate (bl\i) at (0,0);
\coordinate (tl\i) at ($(bl\i)+(0,\l)$);
\coordinate (mnl\i) at ($(bl\i)+0.5*(0,\l)+(-\esp,0)$);
\coordinate (br\i) at (\L,0);
\coordinate (tr\i) at ($(br\i)+(0,\l)$);
\coordinate (mnr\i) at ($(mnl\i)+(\L,0)$);
\coordinate (ml\i) at ($(bl\i)+0.5*(0,\l)+(\esp,0)$);
\coordinate (qbl\i) at ($(bl\i)+0.25*(0,\l)+(\esp,0)$);
\coordinate (qtl\i) at ($(bl\i)+0.75*(0,\l)+(\esp,0)$);
\coordinate (qbr\i) at ($(qbl\i)+(\L,0)$);
\coordinate (mr\i) at ($(ml\i)+(\L,0)$);
\coordinate (qtr\i) at ($(qtl\i)+(\L,0)$);
\coordinate (ql\i) at ($(bl\i)+0.4*(0,\l)+0.5*(\L,0)+(\esp,0)$);
\coordinate (qh\i) at ($(bl\i)+0.6*(0,\l)+0.5*(\L,0)+(\esp,0)$);

%non crossing
\fill [blue!20!white] (bl\i) .. controls +(180:\e) and \control{(mnl\i)}{(-90:2*\e)} .. controls +(90:2*\e) and \control{(tl\i)}{(180:\e)} -- (tr\i) .. controls +(180:\e) and \control{(mnr\i)}{(90:2*\e)} .. controls +(-90:2*\e) and \control{(br\i)}{(180:\e)};
\draw (bl\i) .. controls +(180:\e) and \control{(mnl\i)}{(-90:2*\e)} .. controls +(90:2*\e) and \control{(tl\i)}{(180:\e)}; 
\draw (br\i) .. controls +(180:\e) and \control{(mnr\i)}{(-90:2*\e)} .. controls +(90:2*\e) and \control{(tr\i)}{(180:\e)}; 

%Crossing part
\fill [white!40!gray] (bl\i) .. controls +(0:\e) and \control{(qbl\i)}{(-90:0.5*\e)} .. controls +(90:0.5*\e) and \control{(ql\i)}{(-150:\c)} .. controls +(-30:\c) and \control{(qbr\i)}{(90:0.5*\e)}  .. controls +(-90:0.5*\e) and \control{(br\i)}{(0:\e)} -- cycle;
\fill [white!40!gray] (tl\i) .. controls +(0:\e) and \control{(qtl\i)}{(90:0.5*\e)} .. controls +(-90:0.5*\e) and \control{(qh\i)}{(150:\c)} .. controls +(30:\c) and \control{(qtr\i)}{(-90:0.5*\e)}  .. controls +(90:0.5*\e) and \control{(tr\i)}{(0:\e)} -- cycle;
\fill [blue!20!white] (ql\i) .. controls +(150:\c) and \control{(ml\i)}{(-90:0.25*\e)} .. controls +(90:0.25*\e) and \control{(qh\i)}{(-150:\c)} .. controls +(-30:\c) and \control{(mr\i)}{(90:0.25*\e)} .. controls +(-90:0.25*\e) and \control{(ql\i)}{(30:\c)};
\begin{scope}[even odd rule]
%\clip (-3,-3) rectangle (3,5) (ql\i) circle (0.75*\c cm) (qh\i) circle (0.75*\c cm); 
\draw (bl\i) .. controls +(0:\e) and \control{(qbl\i)}{(-90:0.5*\e)} .. controls +(90:0.5*\e) and \control{(ql\i)}{(-150:\c)} .. controls +(30:\c) and \control{(mr\i)}{(-90:0.25*\e)} .. controls +(90:0.25*\e) and \control{(qh\i)}{(-30:\c)} .. controls +(150:\c) and \control{(qtl\i)}{(-90:0.5*\e)} .. controls +(90:0.5*\e) and \control{(tl\i)}{(0:\e)};
\end{scope}
\begin{scope}[even odd rule]
%\clip (-3,-3) rectangle (3,5) (qh\i) circle (0.75*\c cm) (ql\i) circle (0.75*\c cm); 
\draw (br\i) .. controls +(0:\e) and \control{(qbr\i)}{(-90:0.5*\e)} .. controls +(90:0.5*\e) and \control{(ql\i)}{(-30:\c)} .. controls +(150:\c) and \control{(ml\i)}{(-90:0.25*\e)} .. controls +(90:0.25*\e) and \control{(qh\i)}{(-150:\c)} .. controls +(30:\c) and \control{(qtr\i)}{(-90:0.5*\e)} .. controls +(90:0.5*\e) and \control{(tr\i)}{(0:\e)};
\end{scope}
\end{scope}

\def\i{1}
\fill [blue!20!white] (bl\i) .. controls +(180:\e) and \control{(mnl\i)}{(-90:2*\e)} .. controls +(90:2*\e) and \control{(tl\i)}{(180:\e)} -- (tr\i) .. controls +(180:\e) and \control{(mnr\i)}{(90:2*\e)} .. controls +(-90:2*\e) and \control{(br\i)}{(180:\e)};
\begin{scope}[even odd rule]
\def\i{2}
\clip (-1,3) rectangle (1,-1) [rotate=-90] (bl\i) .. controls +(180:\e) and \control{(mnl\i)}{(-90:2*\e)} .. controls +(90:2*\e) and \control{(tl\i)}{(180:\e)} -- (tr\i) .. controls +(180:\e) and \control{(mnr\i)}{(90:2*\e)} .. controls +(-90:2*\e) and \control{(br\i)}{(180:\e)};
\def\i{1}
\draw (bl\i) .. controls +(180:\e) and \control{(mnl\i)}{(-90:2*\e)} .. controls +(90:2*\e) and \control{(tl\i)}{(180:\e)}; 
\draw (br\i) .. controls +(180:\e) and \control{(mnr\i)}{(-90:2*\e)} .. controls +(90:2*\e) and \control{(tr\i)}{(180:\e)}; 
\end{scope}
\draw [purple] ($(bl\i)!0.5!(br\i)$) .. controls +(180:\e) and \control{($(mnl\i)!0.5!(mnr\i)$)}{(-90:2*\e)}  node (nc1) [pos = 0.3] {} .. controls +(90:2*\e) and \control{($(tl\i)!0.5!(tr\i)$)}{(180:\e)};
\node at (nc1) [rotate = 110, purple] {{\small $\blacktriangleleft$}};

%Core
\def\i{2}
\draw [green!70!purple, xshift = -1.6 cm, yshift = 0.05 cm, rotate = -90] ($(bl\i)!0.5!(br\i)$) .. controls +(180:\e) and \control{($(mnl\i)!0.5!(mnr\i)$)}{(-90:2*\e)}   .. controls +(90:2*\e) and \control{($(tl\i)!0.5!(tr\i)$)}{(180:\e)}node (nc2) [pos = 0.25] {};
\node at (nc2) [rotate = -10, green!70!purple] {{\small $\blacktriangleright$}};
\draw [green!70!purple, xshift = -1.6 cm, yshift = 0.05 cm, rotate = -90] ($(bl\i)!0.5!(br\i)$) .. controls +(0:\e) and \control{(ql\i)}{(-90:\c)} -- (qh\i) .. controls +(90:\c) and \control{($(tl\i)!0.5!(tr\i)$)}{(0:\e)};
\end{scope}
\end{tikzpicture}
\caption{Orientation of the green core when its associated Hopf band (unsigned as it does not matter for the rule) is plumbed on top of the Hopf band with the red core.}
\label{pic_orient_plumbing}
\end{center}
\end{figure}

\begin{definition}
A \emphdef{Hopf arborescent surface} is a surface $\Sigma (T)$ obtained from a plane tree $T$ by this construction, see Figure~\ref{pic_3D_arborescent_link} for an example.
A \emphdef{Hopf arborescent link} is the boundary of a Hopf arborescent surface.
\end{definition}

Since Hopf bands are fibred and this property is preserved under plumbing, Hopf arborescent surfaces are fibres for their boundaries, and are thus of minimal (classical) genus. 
The arbitrary-looking rule that we use to orient the cores of the Hopf band in Step~\ref{step:orient} is new and will turn out to be key for our proofs of Theorem~\ref{T:wqo} and Proposition~\ref{P:mainminor}. 

\subsection{Connections to other classes of knots}\label{sec_connections}

Here, we provide some additional background on Hopf arborescent links, their plumbing structure and their relations to other classes (arborescent links and fibred links). 
In particular one may think that considering rooted trees and always plumbing the new Hopf bands on top of the surface is a strong restriction. 
We will show that this is not the case, i.e. that the family of surfaces and links obtained with less restriction on the sides on which one plumbs is the same as the family considered here.

\subparagraph*{Arborescent knots.} Our definition differs from the classical definition(s) of arborescent links (e.g., in~\cite{bonahon1979new,Gabai_Genus_arborescent}) and their encoding via plane trees in three ways. 
Firstly, we restrict our attention to Hopf bands and not general unknotted annuli, hence the name Hopf arborescent links. 
This restriction is essential to our approach and the minor theory we develop does not extend when one drops this assumption, as Remark~\ref{R:minors} shows. 

Secondly, an arborescent surface is, in general, defined as a surface that retracts on a collection of core curves whose intersection pattern is a tree. 
One can show that in this case, the surface can be reconstructed using an algorithm similar to the one we propose in Section~\ref{S:HopfArborescent}. 
The main difference is that when gluing a new band to the surface, it may be glued on top or below the surface. 
This difference could let one think that our construction is more peculiar than the standard one. 
However, it was remarked by Misev~\cite{Misev_phd} that gluing a Hopf band on top or below an arc yields isotopic surfaces, see Figure~\ref{F:Misev}. 
The only difference is that the orientation of the core curves of the glued band is then reversed. 
This change of orientation then implies that all vertices of the subtree subsequently glued to this band should be reversed to have the same surface. 
This argument shows that the set of surfaces we construct actually coincides with the set of surfaces obtained by the classical construction of arborescent links using Hopf bands only. 
The reason why we chose our presentation (i.e., always gluing on top when going away from the root of the tree) is that it behaves more nicely with respect to our minor theory, and in particular it simplifies the proof of Lemma~\ref{lem:minor3}.

Thirdly, the usual encoding of arborescent knots with plane trees~\cite[Chapter~12]{bonahon1979new} is more detailed in that it puts the number of twists between each pair of adjacent edges around each vertex. Here, since bands have either $2$ positive crossings or $2$ negative, Figure~\ref{pic_twist_equiv} shows that we do not need to specify where the twists are on the band, and thus indicating the sign of the band on each vertex is sufficient. 

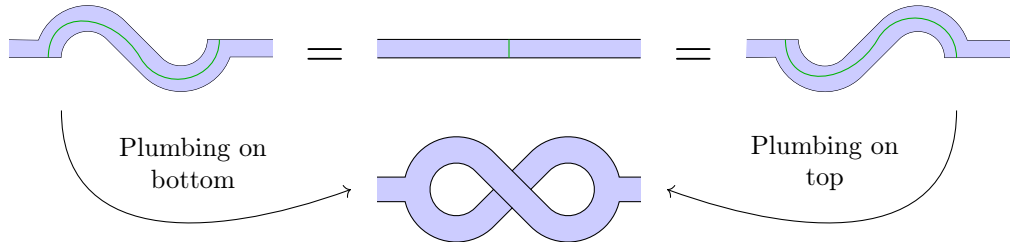
\begin{figure}[ht]
\begin{center}
\begin{tikzpicture}[scale = 0.7]
\begin{scope}[xshift = -7cm]
\clip (0,-1) rectangle (5,2);
\clip (0,0) -- (1,0) arc (180:45:0.5) -- ++(-45:1) arc (-135:-20:1) -- ++(1,0) -- ++ (0,1.5) -- ++(-5.2,0) -- cycle;
\clip (0,0.35) -- ($(1.5,0)+(160:1)$) arc (160:45:1) -- ++(-45:1) arc (-135:0:0.5) -- ++(1.5,0) -- ++ (0,-1.5) -- ++(-5.4,0) -- cycle;
\fill [blue!20!white] (0,-1) rectangle (5.2,1.5);
\draw (0,0) -- (1,0) arc (180:45:0.5) -- ++(-45:1) arc (-135:-20:1) -- ++(1,0);
\draw (0,0.35) -- ($(1.5,0)+(160:1)$) arc (160:45:1) -- ++(-45:1) arc (-135:0:0.5) -- ++(1.5,0); 
\draw [green!70!black] (0.75,0) .. controls +(90:1) and \control{(2.5,0)}{(120:1)} .. controls +(-60:0.8) and \control{(4,0.35)}{(-90:0.8)};
\end{scope}

\begin{scope}
\clip (0,-1) rectangle (5,2);
\fill [blue!20!white] (0,0) rectangle (5,0.35);
\draw (0,0) -- (5,0);
\draw (0,0.35) -- (5,0.35);
\draw [green!70!black] (2.5,0) -- (2.5,0.35);
\end{scope}

\begin{scope}[xshift=7cm, yshift=0.35cm]
\clip (0,-1) rectangle (5,2);
\clip (0,0) -- (1,0) arc (-180:-45:0.5) -- ++(45:1) arc (135:20:1) -- ++(1,0) -- ++ (0,-1.5) -- ++(-5.4,0) -- cycle;
\clip(0,-0.35) -- ($(1.5,0)+(-160:1)$) arc (-160:-45:1) -- ++(45:1) arc (135:0:0.5) -- ++(1.5,0)-- ++ (0,1.5) -- ++(-5.2,0) -- cycle;
\fill [blue!20!white] (0,-1) rectangle (5.2,1.5);
\draw (0,0) -- (1,0) arc (-180:-45:0.5) -- ++(45:1) arc (135:20:1) -- ++(1,0);
\draw (0,-0.35) -- ($(1.5,0)+(-160:1)$) arc (-160:-45:1) -- ++(45:1) arc (135:0:0.5) -- ++(1.5,0); 
\draw [green!70!black] (0.75,0) .. controls +(-90:1) and \control{(2.5,0)}{(-120:1)} .. controls +(60:0.8) and \control{(4,-0.35)}{(90:0.8)};
\end{scope}

\node at (-1,0) {{\huge $=$}};
\node at (6,0) {{\huge $=$}};
\node at (-3.5,-2) [align=center] {Plumbing on\\ bottom};
\node at (8.5,-2) [align=center] {Plumbing on\\ top};
\draw [->] (-6,-1) .. controls +(-90:2.5) and \control{(-0.5,-2.5)}{(-160:3)};
\draw [->] (11,-1) .. controls +(-90:2.5) and \control{(5.5,-2.5)}{(-20:3)};

\begin{scope}[yshift=-2.5cm]
\clip (0,-1) rectangle (5,2);
\coordinate (lle) at ($(1.5,0)+(-45:1)$);
\coordinate (lli) at ($(1.5,0)+(-45:0.5)$);
\coordinate (lhe) at ($(1.5,0)+(45:1)$);
\coordinate (lhi) at ($(1.5,0)+(45:0.5)$);
\coordinate (rhe) at ($(lli)+(45:1.5)$);
\coordinate (rhi) at ($(lle)+(45:1.5)$);
\coordinate (rle) at ($(lhi)+(-45:1.5)$);
\coordinate (rli) at ($(lhe)+(-45:1.5)$);
\path (lhe) arc (45:315:1) node [inner sep =0, pos = 0.45] (lu) {} node [inner sep =0, pos = 0.55] (ld) {};
\path (rhe) arc (135:-135:1) node [inner sep =0, pos = 0.45] (ru) {} node [inner sep =0, pos = 0.55] (rd) {};
\filldraw [fill=blue!20!white] ($(1.5,0)+(45:1)$) arc (45:315:1) -- ++(45:0.5) -- ++(-45:0.5) arc (-135:135:1) -- ++ (-135:0.5) -- cycle;
\fill [blue!20!white] ($(lu)+(-1,0)$) rectangle ($(ld)+(0.1,0)$);
\fill [blue!20!white] ($(ru)+(1,0)$) rectangle ($(rd)+(-0.1,0)$);
\filldraw [fill =white] ($(1.5,0)+(45:0.5)$) arc (45:315:0.5) -- ++(45:0.5) -- cycle;
\filldraw [fill =white] (rhi) arc (135:-135:0.5) -- ++(135:0.5) -- cycle;
\draw ($(ld)+(0.01,0)$) -- ++(-1,0) ($(lu)+(0.01,0)$) -- ++(-1,0) ($(rd)+(-0.01,0)$) -- ++(1,0) ($(ru)+(-0.01,0)$) -- ++(1,0);
\draw ($(lhi)+(-45:0.5)$) -- ++(-45:0.5) ($(lhe)+(-45:0.5)$) -- ++(-45:0.5);
\end{scope}
\end{tikzpicture}
\end{center}
  \caption{Surfaces obtained by plumbing a Hopf band inside or outside a Seifert surface are isotopic.}
  \label{F:Misev}
\end{figure}

\begin{figure}[H]
\begin{center}
\begin{tikzpicture}[scale = 0.8]
\clip (-0.1,-0.1) rectangle (17.1,1.1);
\fill [blue!20!white] (0,0) -- (2.5,0) .. controls +(0:0.25) and \control{(3,0.5)}{(-120:0.25)} .. controls +(120:0.25) and \control{(2.5,1)}{(0:0.25)} -- (0,1) -- cycle;
\fill [white!40!gray] (3,0.5) .. controls +(60:0.5) and \control{(3.5,1)}{(180:0.25)} .. controls +(0:0.25) and \control{(4,0.5)}{(120:0.25)} .. controls +(-120:0.25) and \control{(3.5,0)}{(0:0.25)} .. controls +(180:0.25) and \control{(3,0.5)}{(-60:0.25)};
\fill [blue!20!white] (4,0.5) .. controls +(-60:0.25) and \control{(4.5,0)}{(180:0.25)} -- (5,0) -- (5,1) -- (4.5,1) .. controls +(180:0.25) and \control{(4,0.5)}{(60:0.25)};
\begin{scope}[even odd rule]
\clip (-0.5,-0.5) rectangle (5,1.5) (3,0.5) circle (0.1cm);
\draw (0,0) -- (2.5,0) .. controls +(0:0.25) and \control{(3,0.5)}{(-120:0.25)} .. controls +(60:0.25) and \control{(3.5,1)}{(180:0.25)} .. controls +(0:0.25) and \control{(4,0.5)}{(120:0.25)} .. controls +(-60:0.25) and \control{(4.5,0)}{(180:0.25)} -- (5,0);
\end{scope}
\begin{scope}[even odd rule]
\clip (-0.5,-0.5) rectangle (5,1.5) (4,0.5) circle (0.1cm);
\draw (0,1) -- (2.5,1) .. controls +(0:0.25) and \control{(3,0.5)}{(120:0.25)} .. controls +(-60:0.25) and \control{(3.5,0)}{(180:0.25)} .. controls +(0:0.25) and \control{(4,0.5)}{(-120:0.25)} .. controls +(60:0.25) and \control{(4.5,1)}{(180:0.25)} -- (5,1);
\end{scope}
\draw [dotted] (1,0) -- (1,1) (2,0) -- (2,1);
\node at (1.5,0.5) {{\huge $\top$}};

\begin{scope}[xshift = 6cm]
\fill [blue!20!white] (0,0) -- (0.5,0) .. controls +(0:0.25) and \control{(1,0.5)}{(-120:0.25)} .. controls +(120:0.25) and \control{(0.5,1)}{(0:0.25)} -- (0,1) -- cycle;
\fill [white!40!gray] (1,0.5) .. controls +(60:0.5) and \control{(1.5,1)}{(180:0.25)} -- (3.5,1) .. controls +(0:0.25) and \control{(4,0.5)}{(120:0.25)} .. controls +(-120:0.25) and \control{(3.5,0)}{(0:0.25)} -- (1.5,0) .. controls +(180:0.25) and \control{(1,0.5)}{(-60:0.25)};
\fill [blue!20!white] (4,0.5) .. controls +(-60:0.25) and \control{(4.5,0)}{(180:0.25)} -- (5,0) -- (5,1) -- (4.5,1) .. controls +(180:0.25) and \control{(4,0.5)}{(60:0.25)};
\begin{scope}[even odd rule]
\clip (-0.5,-0.5) rectangle (5,1.5) (1,0.5) circle (0.1cm);
\draw (0,0) -- (0.5,0) .. controls +(0:0.25) and \control{(1,0.5)}{(-120:0.25)} .. controls +(60:0.25) and \control{(1.5,1)}{(180:0.25)} -- (3.5,1) .. controls +(0:0.25) and \control{(4,0.5)}{(120:0.25)} .. controls +(-60:0.25) and \control{(4.5,0)}{(180:0.25)} -- (5,0);
\end{scope}
\begin{scope}[even odd rule]
\clip (-0.5,-0.5) rectangle (5,1.5) (4,0.5) circle (0.1cm);
\draw (0,1) -- (0.5,1) .. controls +(0:0.25) and \control{(1,0.5)}{(120:0.25)} .. controls +(-60:0.25) and \control{(1.5,0)}{(180:0.25)} -- (3.5,0) .. controls +(0:0.25) and \control{(4,0.5)}{(-120:0.25)} .. controls +(60:0.25) and \control{(4.5,1)}{(180:0.25)} -- (5,1);
\end{scope}
\draw [dotted] (2,0) -- (2,1) (3,0) -- (3,1);
\node at (2.5,0.5) {{\huge $\bot$}};
\end{scope}

\begin{scope}[xshift = 12cm]
\fill [blue!20!white] (0,0) -- (0.5,0) .. controls +(0:0.25) and \control{(1,0.5)}{(-120:0.25)} .. controls +(120:0.25) and \control{(0.5,1)}{(0:0.25)} -- (0,1) -- cycle;
\fill [white!40!gray] (1,0.5) .. controls +(60:0.5) and \control{(1.5,1)}{(180:0.25)} .. controls +(0:0.25) and \control{(2,0.5)}{(120:0.25)} .. controls +(-120:0.25) and \control{(1.5,0)}{(0:0.25)} .. controls +(180:0.25) and \control{(1,0.5)}{(-60:0.25)};
\fill [blue!20!white] (2,0.5) .. controls +(-60:0.25) and \control{(2.5,0)}{(180:0.25)} -- (5,0) -- (5,1) -- (2.5,1) .. controls +(180:0.25) and \control{(2,0.5)}{(60:0.25)};
\begin{scope}[even odd rule]
\clip (-0.5,-0.5) rectangle (5,1.5) (1,0.5) circle (0.1cm);
\draw (0,0) -- (0.5,0) .. controls +(0:0.25) and \control{(1,0.5)}{(-120:0.25)} .. controls +(60:0.25) and \control{(1.5,1)}{(180:0.25)} .. controls +(0:0.25) and \control{(2,0.5)}{(120:0.25)} .. controls +(-60:0.25) and \control{(2.5,0)}{(180:0.25)} -- (5,0);
\end{scope}
\begin{scope}[even odd rule]
\clip (-0.5,-0.5) rectangle (5,1.5) (2,0.5) circle (0.1cm);
\draw (0,1) -- (0.5,1) .. controls +(0:0.25) and \control{(1,0.5)}{(120:0.25)} .. controls +(-60:0.25) and \control{(1.5,0)}{(180:0.25)} .. controls +(0:0.25) and \control{(2,0.5)}{(-120:0.25)} .. controls +(60:0.25) and \control{(2.5,1)}{(180:0.25)} -- (5,1);
\end{scope}
\draw [dotted] (3,0) -- (3,1) (4,0) -- (4,1);
\node at (3.5,0.5) {{\huge $\top$}};
\end{scope}
\node at (5.5,0.5) {{\huge $=$}};
\node at (11.5,0.5) {{\huge $=$}};
\end{tikzpicture}
\caption{The surface does not depend on where the twists are.}
\label{pic_twist_equiv}
\end{center}
\end{figure}
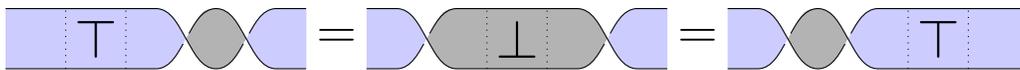

\smallskip

As noticed in Section~\ref{S:MainProof} with the set~$\mathcal{T}(L)$, it may happen that several trees yield the same (isotopy class of) oriented link and the same surface. 
In other words, our construction induces a well-defined map from plane trees to isotopy classes of surfaces in~$\Sp^3$, but this map may not be injective. 
This lack of injectivity is not an issue for the problems we address in this paper, for the preimage of a given link or surface is finite. 
However, it is carefully studied by Bonahon and Siebenmann~\cite{bonahon1979new} who give a recipe to detect (in the more general case of arborescent links) when two $\Z$-labelled plane trees yield isotopic links.

\subparagraph*{More general plumbing structures.}
By construction, the tree~$T$ involved in defining a Hopf arborescent surface $\Sigma(T)$ is the intersection graph of the set of core curves~$\mathcal{C}_T$. 
It may well happen that a surface retracts on different collections of Hopf curves~\cite{misev_arborescent,bonahon1979new} whose intersection pattern is very different from~$T$, and in particular is not a tree. 
In that case, the surface can still be obtained by a sequence of Hopf plumbings, but some plumbing are made on arcs that intersect the core of more than one Hopf band. 
On Figure~\ref{pic_tree_cycle_plumbing} for example, the green curve on the right is indeed the core of a Hopf band, see top of Figure~\ref{pic_cut_diago} to picture how a tubular neighbourhood of the core is a band with two positive crossings. 

\begin{figure}[H]
\begin{center}
\begin{tikzpicture}[scale = 0.85]
\begin{scope}
%Band 1
\def\i{1}
%coordinate
\coordinate (e\i) at (15:1.5);
\coordinate (i\i) at (15:1);
\coordinate (l\i) at (2.5:1.25);
\coordinate (h\i) at (27.5:1.25);

%Surface
\fill [blue!20!white] (h\i) .. controls +(70:0.1) and \control{(40:1.5)}{(-50:0.2)} arc (40:350:1.5) .. controls +(80:0.1) and \control{(l\i)}{(-40:0.2)} .. controls +(-140:0.2) and \control{(-15:1)}{(80:0.2)} arc (-15:-310:1) .. controls +(-50:0.2) and \control{(h\i)}{(170:0.2)};
\fill [white!40!gray] (l\i) .. controls +(40:0.2) and \control{(e\i)}{(-75:0.1)} .. controls +(105:0.1) and \control{(h\i)}{(-10:0.2)} .. controls +(-110:0.1) and \control{(i\i)}{(105:0.1)} .. controls +(-75:0.1) and \control{(l\i)}{(140:0.1)};

%Strands
\begin{scope}[even odd rule]
\clip (-1.6,1.6) rectangle (1.6,-1.6) (l\i) circle (0.1);
\draw [thick] (50:1) arc (50:345:1) .. controls +(80:0.2) and \control{(l\i)}{(-140:0.2)} .. controls +(40:0.2) and \control{(e\i)}{(-75:0.1)} .. controls +(105:0.1) and \control{(h\i)}{(-10:0.2)} .. controls +(170:0.2) and \control{(50:1)}{(-50:0.2)};
\end{scope}
\begin{scope}[even odd rule]
\clip (-1.6,1.6) rectangle (1.6,-1.6) (h\i) circle (0.1);
\draw [thick] (40:1.5) arc (40:350:1.5) .. controls +(80:0.1) and \control{(l\i)}{(-40:0.2)} .. controls +(140:0.1) and \control{(i\i)}{(-75:0.1)} .. controls +(105:0.1) and \control{(h\i)}{(-110:0.1)} .. controls +(70:0.1) and \control{(40:1.5)}{(-50:0.2)};
\end{scope}

%Core
\draw [purple] (10:1.25) arc (10:260:1.25);

%Band 2
\begin{scope}[yshift = -1.8cm, xshift = 0.5 cm]
\def\i{2}
%coordinate
\coordinate (e\i) at (15:1.5);
\coordinate (i\i) at (15:1);
\coordinate (l\i) at (2.5:1.25);
\coordinate (h\i) at (27.5:1.25);

%Surface
\fill [blue!20!white] (h\i) .. controls +(70:0.1) and \control{(40:1.5)}{(-50:0.2)} arc (40:350:1.5) .. controls +(80:0.1) and \control{(l\i)}{(-40:0.2)} .. controls +(-140:0.2) and \control{(-15:1)}{(80:0.2)} arc (-15:-310:1) .. controls +(-50:0.2) and \control{(h\i)}{(170:0.2)};
\fill [white!40!gray] (l\i) .. controls +(40:0.2) and \control{(e\i)}{(-75:0.1)} .. controls +(105:0.1) and \control{(h\i)}{(-10:0.2)} .. controls +(-110:0.1) and \control{(i\i)}{(105:0.1)} .. controls +(-75:0.1) and \control{(l\i)}{(140:0.1)};

%Strands
\begin{scope}[even odd rule]
\clip (-1.6,1.6) rectangle (1.6,-1.6) (l\i) circle (0.1);
\clip (-1.6,1.6) rectangle (1.6,-1.6) [yshift = 1.8cm, xshift = -0.5 cm] (-90:1) arc (-90:-10:1) -- (-10:1.5) arc (-10:-90:1.5) --cycle;
\draw [thick] (50:1) arc (50:345:1) .. controls +(80:0.2) and \control{(l\i)}{(-140:0.2)} .. controls +(40:0.2) and \control{(e\i)}{(-75:0.1)} .. controls +(105:0.1) and \control{(h\i)}{(-10:0.2)} .. controls +(170:0.2) and \control{(50:1)}{(-50:0.2)};
\end{scope}
\begin{scope}[even odd rule]
\clip (-1.6,1.6) rectangle (1.6,-1.6) (h\i) circle (0.1);
\clip (-1.6,1.6) rectangle (1.6,-1.6) [yshift = 1.8cm, xshift = -0.5 cm] (-90:1) arc (-90:-10:1) -- (-10:1.5) arc (-10:-90:1.5) --cycle;
\draw [thick] (40:1.5) arc (40:350:1.5) .. controls +(80:0.1) and \control{(l\i)}{(-40:0.2)} .. controls +(140:0.1) and \control{(i\i)}{(-75:0.1)} .. controls +(105:0.1) and \control{(h\i)}{(-110:0.1)} .. controls +(70:0.1) and \control{(40:1.5)}{(-50:0.2)};
\end{scope}
\end{scope}

%Core and overlap
\draw [purple] (10:1.25) arc (10:-100:1.25);
\draw [green!80!purple] ($(10:1.25)+(0.5,-1.8)$) arc (10:260:1.25);

%Band 3
\begin{scope}[yshift = -3.6cm, xshift = 1 cm]
\def\i{3}
%coordinate
\coordinate (e\i) at (15:1.5);
\coordinate (i\i) at (15:1);
\coordinate (l\i) at (2.5:1.25);
\coordinate (h\i) at (27.5:1.25);

%Surface
\fill [blue!20!white] (h\i) .. controls +(70:0.1) and \control{(40:1.5)}{(-50:0.2)} arc (40:350:1.5) .. controls +(80:0.1) and \control{(l\i)}{(-40:0.2)} .. controls +(-140:0.2) and \control{(-15:1)}{(80:0.2)} arc (-15:-310:1) .. controls +(-50:0.2) and \control{(h\i)}{(170:0.2)};
\fill [white!40!gray] (l\i) .. controls +(40:0.2) and \control{(e\i)}{(-75:0.1)} .. controls +(105:0.1) and \control{(h\i)}{(-10:0.2)} .. controls +(-110:0.1) and \control{(i\i)}{(105:0.1)} .. controls +(-75:0.1) and \control{(l\i)}{(140:0.1)};

%Strands
\begin{scope}[even odd rule]
\clip (-1.6,1.6) rectangle (1.6,-1.6) (l\i) circle (0.1);
\clip (-1.6,1.6) rectangle (1.6,-1.6) [yshift = 1.8cm, xshift = -0.5 cm] (-90:1) arc (-90:-10:1) -- (-10:1.5) arc (-10:-90:1.5) --cycle;
\draw [thick] (50:1) arc (50:345:1) .. controls +(80:0.2) and \control{(l\i)}{(-140:0.2)} .. controls +(40:0.2) and \control{(e\i)}{(-75:0.1)} .. controls +(105:0.1) and \control{(h\i)}{(-10:0.2)} .. controls +(170:0.2) and \control{(50:1)}{(-50:0.2)};
\end{scope}
\begin{scope}[even odd rule]
\clip (-1.6,1.6) rectangle (1.6,-1.6) (h\i) circle (0.1);
\clip (-1.6,1.6) rectangle (1.6,-1.6) [yshift = 1.8cm, xshift = -0.5 cm] (-90:1) arc (-90:-10:1) -- (-10:1.5) arc (-10:-90:1.5) --cycle;
\draw [thick] (40:1.5) arc (40:350:1.5) .. controls +(80:0.1) and \control{(l\i)}{(-40:0.2)} .. controls +(140:0.1) and \control{(i\i)}{(-75:0.1)} .. controls +(105:0.1) and \control{(h\i)}{(-110:0.1)} .. controls +(70:0.1) and \control{(40:1.5)}{(-50:0.2)};
\end{scope}
\end{scope}

%Core and overlap
\draw [green!80!purple] ($(10:1.25)+(0.5,-1.8)$) arc (10:-100:1.25);
\begin{scope} [yshift = -3.6cm, xshift = 1 cm]
\clip (-1.6,1.6) rectangle (1.6,-1.6);
\draw [orange] (0,0) circle (1.25);
\end{scope}
\end{scope}

%Graph
\draw [->] (2.5,-1) -- (3.4,-1);
\draw (4,0) -- (4,-2);
\fill [purple] (4,0) circle (0.1cm);
\fill [green!80!purple] (4,-1) circle (0.1cm);
\fill [orange] (4,-2) circle (0.1cm);

\begin{scope}[xshift = 8cm]
%Band 1
\def\i{1}
%coordinate
\coordinate (e\i) at (15:1.5);
\coordinate (i\i) at (15:1);
\coordinate (l\i) at (2.5:1.25);
\coordinate (h\i) at (27.5:1.25);

%Surface
\fill [blue!20!white] (h\i) .. controls +(70:0.1) and \control{(40:1.5)}{(-50:0.2)} arc (40:350:1.5) .. controls +(80:0.1) and \control{(l\i)}{(-40:0.2)} .. controls +(-140:0.2) and \control{(-15:1)}{(80:0.2)} arc (-15:-310:1) .. controls +(-50:0.2) and \control{(h\i)}{(170:0.2)};
\fill [white!40!gray] (l\i) .. controls +(40:0.2) and \control{(e\i)}{(-75:0.1)} .. controls +(105:0.1) and \control{(h\i)}{(-10:0.2)} .. controls +(-110:0.1) and \control{(i\i)}{(105:0.1)} .. controls +(-75:0.1) and \control{(l\i)}{(140:0.1)};

%Strands
\begin{scope}[even odd rule]
\clip (-1.6,1.6) rectangle (1.6,-1.6) (l\i) circle (0.1);
\draw [thick] (50:1) arc (50:345:1) .. controls +(80:0.2) and \control{(l\i)}{(-140:0.2)} .. controls +(40:0.2) and \control{(e\i)}{(-75:0.1)} .. controls +(105:0.1) and \control{(h\i)}{(-10:0.2)} .. controls +(170:0.2) and \control{(50:1)}{(-50:0.2)};
\end{scope}
\begin{scope}[even odd rule]
\clip (-1.6,1.6) rectangle (1.6,-1.6) (h\i) circle (0.1);
\draw [thick] (40:1.5) arc (40:350:1.5) .. controls +(80:0.1) and \control{(l\i)}{(-40:0.2)} .. controls +(140:0.1) and \control{(i\i)}{(-75:0.1)} .. controls +(105:0.1) and \control{(h\i)}{(-110:0.1)} .. controls +(70:0.1) and \control{(40:1.5)}{(-50:0.2)};
\end{scope}
%Core
\draw [purple] (10:1.25) arc (10:260:1.25);
\draw [green!80!purple] (-90:1.15) arc (-90:-270:1.15);

\begin{scope}[yshift = -1.8cm, xshift = 0.5 cm]
\def\i{2}
%coordinate
\coordinate (e\i) at (15:1.5);
\coordinate (i\i) at (15:1);
\coordinate (l\i) at (2.5:1.25);
\coordinate (h\i) at (27.5:1.25);

%Surface
\fill [blue!20!white] (h\i) .. controls +(70:0.1) and \control{(40:1.5)}{(-50:0.2)} arc (40:350:1.5) .. controls +(80:0.1) and \control{(l\i)}{(-40:0.2)} .. controls +(-140:0.2) and \control{(-15:1)}{(80:0.2)} arc (-15:-310:1) .. controls +(-50:0.2) and \control{(h\i)}{(170:0.2)};
\fill [white!40!gray] (l\i) .. controls +(40:0.2) and \control{(e\i)}{(-75:0.1)} .. controls +(105:0.1) and \control{(h\i)}{(-10:0.2)} .. controls +(-110:0.1) and \control{(i\i)}{(105:0.1)} .. controls +(-75:0.1) and \control{(l\i)}{(140:0.1)};

%Strands
\begin{scope}[even odd rule]
\clip (-1.6,1.6) rectangle (1.6,-1.6) (l\i) circle (0.1);
\clip (-1.6,1.6) rectangle (1.6,-1.6) [yshift = 1.8cm, xshift = -0.5 cm] (-90:1) arc (-90:-10:1) -- (-10:1.5) arc (-10:-90:1.5) --cycle;
\draw [thick] (50:1) arc (50:345:1) .. controls +(80:0.2) and \control{(l\i)}{(-140:0.2)} .. controls +(40:0.2) and \control{(e\i)}{(-75:0.1)} .. controls +(105:0.1) and \control{(h\i)}{(-10:0.2)} .. controls +(170:0.2) and \control{(50:1)}{(-50:0.2)};
\end{scope}
\begin{scope}[even odd rule]
\clip (-1.6,1.6) rectangle (1.6,-1.6) (h\i) circle (0.1);
\clip (-1.6,1.6) rectangle (1.6,-1.6) [yshift = 1.8cm, xshift = -0.5 cm] (-90:1) arc (-90:-10:1) -- (-10:1.5) arc (-10:-90:1.5) --cycle;
\draw [thick] (40:1.5) arc (40:350:1.5) .. controls +(80:0.1) and \control{(l\i)}{(-40:0.2)} .. controls +(140:0.1) and \control{(i\i)}{(-75:0.1)} .. controls +(105:0.1) and \control{(h\i)}{(-110:0.1)} .. controls +(70:0.1) and \control{(40:1.5)}{(-50:0.2)};
\end{scope}
\end{scope}

%Core and overlap
\draw [purple] (10:1.25) arc (10:-100:1.25);
\draw [green!80!purple] (-90:1.15) arc (-90:-60:1.15) .. controls +(30:0.3) and \control{($(90:1.25)+(0.5,-1.8)$)}{(0:0.3)} arc (90:260:1.25);

%Band 3
\begin{scope}[yshift = -3.6cm, xshift = 1 cm]
\def\i{3}
%coordinate
\coordinate (e\i) at (15:1.5);
\coordinate (i\i) at (15:1);
\coordinate (l\i) at (2.5:1.25);
\coordinate (h\i) at (27.5:1.25);

%Surface
\fill [blue!20!white] (h\i) .. controls +(70:0.1) and \control{(40:1.5)}{(-50:0.2)} arc (40:350:1.5) .. controls +(80:0.1) and \control{(l\i)}{(-40:0.2)} .. controls +(-140:0.2) and \control{(-15:1)}{(80:0.2)} arc (-15:-310:1) .. controls +(-50:0.2) and \control{(h\i)}{(170:0.2)};
\fill [white!40!gray] (l\i) .. controls +(40:0.2) and \control{(e\i)}{(-75:0.1)} .. controls +(105:0.1) and \control{(h\i)}{(-10:0.2)} .. controls +(-110:0.1) and \control{(i\i)}{(105:0.1)} .. controls +(-75:0.1) and \control{(l\i)}{(140:0.1)};

%Strands
\begin{scope}[even odd rule]
\clip (-1.6,1.6) rectangle (1.6,-1.6) (l\i) circle (0.1);
\clip (-1.6,1.6) rectangle (1.6,-1.6) [yshift = 1.8cm, xshift = -0.5 cm] (-90:1) arc (-90:-10:1) -- (-10:1.5) arc (-10:-90:1.5) --cycle;
\draw [thick] (50:1) arc (50:345:1) .. controls +(80:0.2) and \control{(l\i)}{(-140:0.2)} .. controls +(40:0.2) and \control{(e\i)}{(-75:0.1)} .. controls +(105:0.1) and \control{(h\i)}{(-10:0.2)} .. controls +(170:0.2) and \control{(50:1)}{(-50:0.2)};
\end{scope}
\begin{scope}[even odd rule]
\clip (-1.6,1.6) rectangle (1.6,-1.6) (h\i) circle (0.1);
\clip (-1.6,1.6) rectangle (1.6,-1.6) [yshift = 1.8cm, xshift = -0.5 cm] (-90:1) arc (-90:-10:1) -- (-10:1.5) arc (-10:-90:1.5) --cycle;
\draw [thick] (40:1.5) arc (40:350:1.5) .. controls +(80:0.1) and \control{(l\i)}{(-40:0.2)} .. controls +(140:0.1) and \control{(i\i)}{(-75:0.1)} .. controls +(105:0.1) and \control{(h\i)}{(-110:0.1)} .. controls +(70:0.1) and \control{(40:1.5)}{(-50:0.2)};
\end{scope}
\end{scope}

%Core and overlap
\draw [green!80!purple] ($(-100:1.25)+(0.5,-1.8)$) arc (-100:60:1.25) .. controls +(150:0.3) and \control{(-20:1.15)}{(-110:0.3)};
\begin{scope}
\clip (0,0) -- (-30:1) arc (-30:-15:1) .. controls +(80:0.2) and \control{(l1)}{(-140:0.2)} -- (h1) .. controls +(170:0.2) and \control{(50:1)}{(-50:0.2)} arc (50:90:1) -- ++(90:0.5) -- ++(0:1.5) -- ++(0,-3) -- ++(-1.5,0) -- cycle; 
\draw [green!80!purple] (-20:1.15) arc (-20:90:1.15);
\end{scope}
\begin{scope}
\clip (0,0) -- (l1) .. controls +(40:0.2) and \control{(e1)}{(-75:0.1)} .. controls +(105:0.1) and \control{(h1)}{(-10:0.2)} -- cycle; 
\draw [green!80!purple] (0:1.35) arc (0:25:1.35);
\end{scope}
\draw [orange] (1,-3.6) circle (1.25);

%Graph
\draw [->] (2.5,-1) -- (3.4,-1);
\draw (4,-0.5) -- +(-120:1) -- +(-60:1) -- cycle;
\fill [purple] (4,-0.5) circle (0.1cm);
\fill [green!80!purple] ($(4,-0.5)+(-120:1)$) circle (0.1cm);
\fill [orange] ($(4,-0.5)+(-60:1)$) circle (0.1cm);
\end{scope}
\end{tikzpicture}
\caption{A Hopf arborescent plumbing with a cyclic plumbing graph.}
\label{pic_tree_cycle_plumbing}
\end{center}
\end{figure}
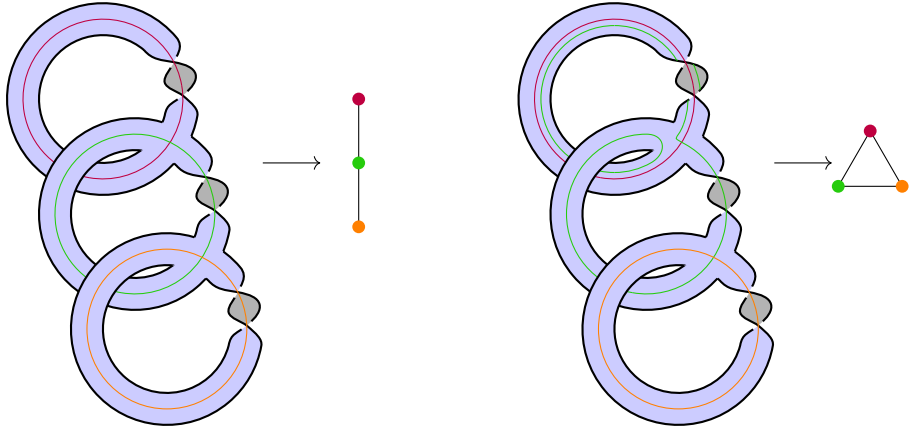

\subsection{Minors on surfaces, links, and plane trees}\label{ss:minors}

Since we focus our investigation on Hopf arborescent links, we define a stronger notion of minor that is well-tailored to these links. We say that a Hopf arborescent link $L_1$ is a \emphdef{link-minor} of $L_2$ if there exist $T_1$ and $T_2$, two labelled plane trees such that $\Sigma(T_1)$ and $\Sigma(T_2)$ are canonical Seifert surfaces of $L_1$ and $L_2$ respectively and $T_1 \hookrightarrow T_2$. The main result of this section is the following one, establishing that if $L_1$ is a link-minor of $L_2$ then the Seifert surface of $L_1$ is a surface-minor of the Seifert surface of $L_2$.

\begin{proposition}\label{P:mainminor}
Let $T_1$ and $T_2$ be two plane trees such that $T_1$ admits a homeomorphic embedding into $T_2$. Then the Hopf arborescent surface $\Sigma (T_1)$ is an incompressible subsurface of~$\Sigma (T_2)$.
\end{proposition}

The proof relies on Lemmas~\ref{lem:minor1}, \ref{lem:minor2}, and \ref{lem:minor3}, which correspond respectively to the operations (i), (ii), and (iv) defining homeomorphic embeddings of trees. We first prove:

\begin{lemma}\label{lem_cut}
Let $\Sigma$ be surface and $\gamma$ be an arc that is not boundary-parallel in $\Sigma$ with both extremities in $\partial \Sigma$. Then cutting $\Sigma$ along $\gamma$ yields a surface $\Sigma' = \overline{\Sigma \smallsetminus \gamma}$ such that $\Sigma' \preccurlyeq \Sigma$.
\end{lemma}

\begin{proof}
By definition of $\Sigma'$, there is a natural map $h: \Sigma' \rightarrow \Sigma$ that is injective except on $h^{-1} (\gamma) = \gamma_1 \cup \gamma_2$. Let $T$ be a tubular neighbourhood of $\gamma$ in $\Sigma$. Its boundary can be decomposed into $t_1,t_2$, two arcs isotopic to $\gamma$ in $\Sigma$ and two open arcs of $\partial \Sigma$. Isotope $h$ within~$\Sigma$ so that $h(\gamma_1) = t_1$, $h(\gamma_2) = t_2$, and $h(\Sigma') \cap T = t_1 \cup t_2$. It follows that $h(\Sigma')$ is a subsurface of $\Sigma$ such that $\Sigma \smallsetminus \Sigma'$ is not an open disc since $\gamma$ is not boundary-parallel.
\end{proof}

Lemma~\ref{lem_cut} essentially states that our surfaces behave well with respect to the surface-minor relation when cut along any essential arc. 
An important point is that cutting along an arc that is the diagonal of a plumbing rectangle merges two bands into one new band with two extra crossings that are either negative or positive depending on the diagonal, see Figure~\ref{pic_cut_diago}. 
So, cutting the plumbing of two positive Hopf bands along the diagonal that produces two negative crossings yields a positive Hopf band. 
Symmetrically, one can merge two negative Hopf bands into one negative by cutting along the other diagonal. 
Furthermore, when having a plumbing of two bands with opposite signs, one can merge them into a band with either sign depending on which cut one chooses. 

\begin{figure}[h]
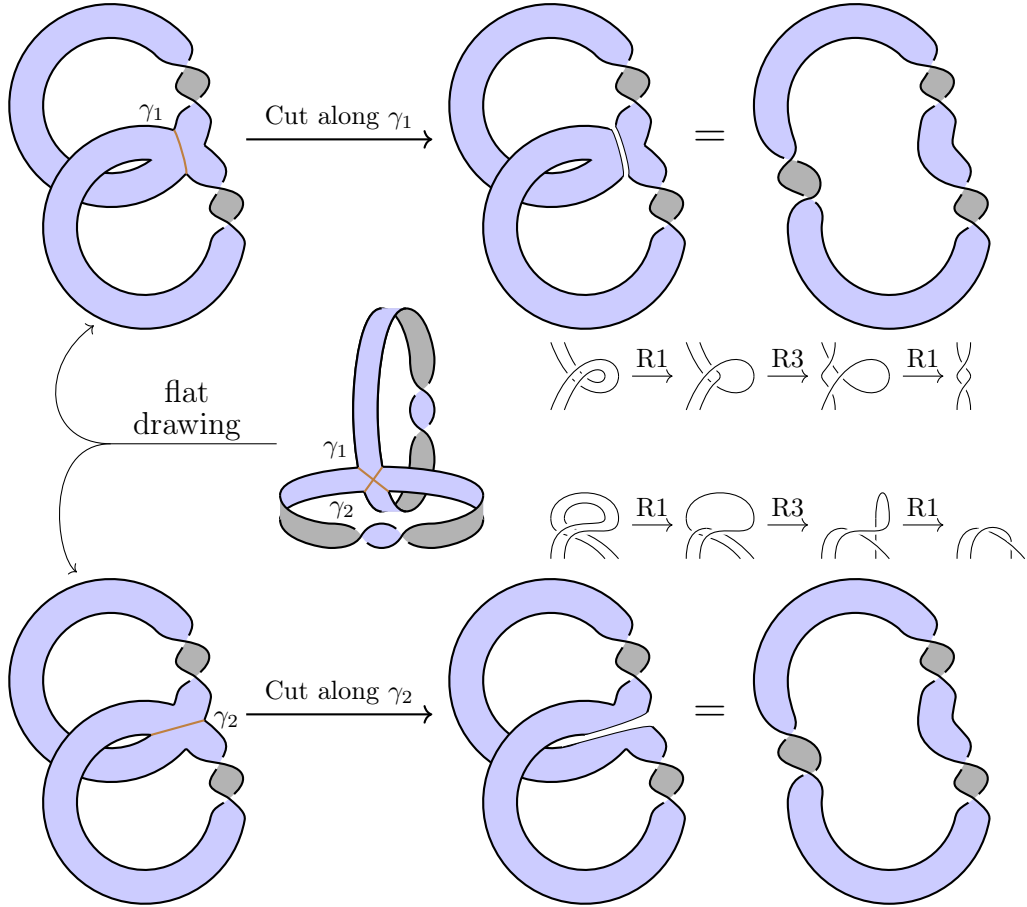

\begin{center}

\caption{Cutting a Hopf plumbing of two positive Hopf bands along the diagonal $\gamma_2$ (resp. $\gamma_1$) induces $2$ positive (resp. negative) crossings. The Reidemeister moves are showed to help understand which crossings are obtained. The top right surface is a Hopf band while the bottom right one is a band with $3$ positive full twists.}
\label{pic_cut_diago}
\end{center}
\end{figure}

\begin{lemma}\label{lem:minor1}
Assume that $T_1$ is obtained from $T_2$ by deleting a leaf. 
Then $\Sigma (T_1)$ is an incompressible subsurface of $\Sigma (T_2)$. 
\end{lemma}

\begin{proof}
Let $D$ be the plumbing rectangle of the Hopf band $H$ associated to the additional leaf $v$ of $T_2$ compared to $T_1$. By definition, $D$ has two sides $\gamma_1, \gamma_2$ in $\partial \Sigma (T_1)$. Thus $\gamma_1$ is also an arc of $\Sigma (T_2)$ with its extremities in $\partial \Sigma (T_2)$. By Lemma~\ref{lem_cut}, $\Sigma (T_2)$ cut along $\gamma_1$ is an incompressible subsurface $\Sigma'$ of $\Sigma (T_2)$. Furthermore, the remaining of $H$ is a disc that can be isotoped into a neighbourhood of $\gamma_2$ so that $\Sigma (T_1) = \Sigma' \preccurlyeq \Sigma (T_2)$, see Figure~\ref{pic_lem_delete_leaf}.\end{proof}

\begin{figure}[ht]
\begin{center}
\begin{tikzpicture}
\begin{scope}
\clip (-0.75,-2) rectangle (2.75,2.5);
%Band filling
\fill [blue!20!white] (-2,-1) .. controls +(70:0.5) and \control{(0,0)}{(-166:0.5)} -- (1,0.25) .. controls +(14:1) and \control{(3.5,-0.5)}{(140:0.5)} -- (4.5,-1.25) .. controls +(140:0.5) and \control{(2,-0.5)}{(14:1)} -- (1,-0.75) .. controls +(-166:0.5) and \control{(-1,-1.75)}{(70:0.5)} -- cycle;
\draw (-2,-1) .. controls +(70:0.5) and \control{(0,0)}{(-166:0.5)} -- (1,0.25) .. controls +(14:1) and \control{(3.5,-0.5)}{(140:0.5)}; 
\draw [xshift = 1cm, yshift = -0.75cm] (-2,-1) .. controls +(70:0.5) and \control{(0,0)}{(-166:0.5)} -- (1,0.25) .. controls +(14:1) and \control{(3.5,-0.5)}{(140:0.5)}; 
\node [blue!80!white] at (0,-0.5) {$\Sigma (T_2)$};

%Coordinate
\coordinate (h) at (1,1.95);
\coordinate (l) at (1, 1.55); 
\coordinate (i) at (0.5, 1.65);
\coordinate (e) at (1.5, 1.9);

%Filling
\begin{scope}
\clip (h) .. controls +(180:0.2) and \control{(0,2)}{(-36.75:0.2)} .. controls +(143.25:0.1) and \control{(-0.25,2.15)}{(0:0.1)} .. controls +(180:0.35) and \control{(0,0)}{(143.25:1)} -- (1,-0.75) .. controls +(-36.75:0.1) and \control{(1.25,-0.9)}{(180:0.1)} .. controls +(0:0.35) and \control{(1,1.25)}{(-36.75:1)} -- cycle;
\fill [blue!20!white] (1,2.25) .. controls +(143.25:0.1) and \control{(0.75,2.4)}{(0:0.1)} .. controls +(180:0.35) and \control{(1,0.25)}{(143.25:1)} -- (2,-0.5) .. controls +(-36.75:0.1) and \control{(2.25,-0.65)}{(180:0.1)} -- ++(0, -1) -- ++(-3.5,0) -- ++(0,4.5) -- ++(2.5,0) -- cycle;
\draw [purple] (1,-0.75) -- (2,-0.5);
\end{scope}
\fill [blue!20!white] (h) .. controls +(-150:0.1) and \control{(i)}{(143.25:0.2)} .. controls +(-36.75:0.1) and \control{(l)}{(180:0.2)} .. controls +(60:0.2) and \control{(e)}{(-36.75:0.2)} .. controls +(143.25:0.1) and \control{(h)}{(0:0.2)};
\fill [white!40!gray] (-0.25,2.15) .. controls +(0:0.1) and \control{(0,2)}{(143.25:0.1)} .. controls +(-36.75:0.2) and \control{(h)}{(180:0.2)} .. controls +(30:0.1) and \control{(1,2.25)}{(-36.75:0.2)} .. controls +(143.25:0.1) and \control{(0.75,2.4)}{(0:0.1)}-- cycle;
\fill [white!40!gray] (1,1.25) .. controls +(143.25:0.1) and \control{(l)}{(-130:0.2)} .. controls +(180:0.2) and \control{(2,1.5)}{(143.25:0.2)} .. controls +(-36.75:1) and \control{(2.25,-0.65)}{(0:0.35)} -- (1.25,-0.9) .. controls +(0:0.35) and \control{(1,1.25)}{(-36.75:1)};
\begin{scope}
\clip(1,1.25) .. controls +(143.25:0.1) and \control{(l)}{(-130:0.2)} .. controls +(180:0.2) and \control{(2,1.5)}{(143.25:0.2)} .. controls +(-36.75:1) and \control{(2.25,-0.65)}{(0:0.35)} -- (1.25,-0.9) .. controls +(0:0.35) and \control{(1,1.25)}{(-36.75:1)};
\draw [purple, dotted] (1,-0.75) -- (2,-0.5);
\end{scope}

%Crossings
\begin{scope}
\clip (-0.25,2.15) .. controls +(0:0.1) and \control{(0,2)}{(143.25:0.1)} .. controls +(-36.75:0.2) and \control{(h)}{(180:0.2)} .. controls +(30:0.1) and \control{(1,2.25)}{(-36.75:0.2)} .. controls +(143.25:0.1) and \control{(0.75,2.4)}{(0:0.1)}-- cycle;
\draw [dashed] (0.75,2.4) .. controls +(180:0.35) and \control{(1,0.25)}{(143.25:1)};
\end{scope}
\begin{scope}
\clip (1,1.25) .. controls +(143.25:0.1) and \control{(l)}{(-130:0.2)} .. controls +(180:0.2) and \control{(2,1.5)}{(143.25:0.2)} .. controls +(-36.75:1) and \control{(2.25,-0.65)}{(0:0.35)} -- (1.25,-0.9) .. controls +(0:0.35) and \control{(1,1.25)}{(-36.75:1)};
\draw [dashed] (4.5, -1.25) .. controls +(140:0.5) and \control{(2,-0.5)}{(14:1)} .. controls +(-36.75:0.1) and \control{(2.25,-0.65)}{(180:0.1)};
\end{scope}
\begin{scope}
\clip (h) .. controls +(180:0.2) and \control{(0,2)}{(-36.75:0.2)} .. controls +(143.25:0.1) and \control{(-0.25,2.15)}{(0:0.1)} .. controls +(180:0.35) and \control{(0,0)}{(143.25:1)} -- (1,-0.75) .. controls +(-36.75:0.1) and \control{(1.25,-0.9)}{(180:0.1)} .. controls +(0:0.35) and \control{(1,1.25)}{(-36.75:1)} -- cycle;
\draw (0.75,2.4) .. controls +(180:0.35) and \control{(1,0.25)}{(143.25:1)};
\end{scope}
\draw (1,2.25) .. controls +(143.25:0.1) and \control{(0.75,2.4)}{(0:0.1)};
\draw (2.25,-0.65) .. controls +(0:0.35) and \control{(2,1.5)}{(-36.75:1)};
\draw (0,2) .. controls +(143.25:0.1) and \control{(-0.25,2.15)}{(0:0.1)} .. controls +(180:0.35) and \control{(0,0)}{(143.25:1)} (1,-0.75) .. controls +(-36.75:0.1) and \control{(1.25,-0.9)}{(180:0.1)} .. controls +(0:0.35) and \control{(1,1.25)}{(-36.75:1)} ;
\begin{scope}[even odd rule] 
\clip (0, 2.25) rectangle (3, 1) (l) circle (0.1cm);
\draw (1,1.25) .. controls +(143.25:0.1) and \control{(l)}{(-130:0.2)} .. controls +(60:0.2) and \control{(e)}{(-36.75:0.2)};
\end{scope}
\draw (e) .. controls +(143.25:0.1) and \control{(h)}{(0:0.2)} .. controls +(180:0.2) and \control{(0,2)}{(-36.75:0.2)};
\begin{scope}[even odd rule] 
\clip (0,2) -- (1,1.25) -- (2,1.5) -- (1,2.25) -- (0,2) (h) circle (0.13cm);
\draw (1,2.25) .. controls +(-36.75:0.2) and \control{(h)}{(30:0.1)} .. controls +(-150:0.1) and \control{(i)}{(143.25:0.2)};
\end{scope}
\draw (i) .. controls +(-36.75:0.1) and \control{(l)}{(180:0.2)} .. controls +(0:0.2) and \control{(2,1.5)}{(143.25:0.2)};
\node at (1.85,0.75) {$H$};
\draw [purple] (0,0) -- (1,0.25);
\node at (0,0) [above right, purple] {$\gamma_2$};
\node at (1,-0.75) [above, purple] {$\gamma_1$};
\end{scope}

\node at (3.75,0) [align = center] {cut \\ along $\gamma_1$};
\draw [->] (3,-0.5) -- (4.5,-0.5);

\begin{scope}[xshift = 5.5cm]
\clip (-0.75,-2) rectangle (2.75,2.5);
%Band filling
\fill [blue!20!white] (-2,-1) .. controls +(70:0.5) and \control{(0,0)}{(-166:0.5)} -- (1,0.25) .. controls +(14:1) and \control{(3.5,-0.5)}{(140:0.5)} -- (4.5,-1.25) .. controls +(140:0.5) and \control{(2,-0.5)}{(14:1)} -- (1,-0.75) .. controls +(-166:0.5) and \control{(-1,-1.75)}{(70:0.5)} -- cycle;
\draw (-2,-1) .. controls +(70:0.5) and \control{(0,0)}{(-166:0.5)} -- (1,0.25) .. controls +(14:1) and \control{(3.5,-0.5)}{(140:0.5)}; 
\draw [xshift = 1cm, yshift = -0.75cm] (-2,-1) .. controls +(70:0.5) and \control{(0,0)}{(-166:0.5)} -- (1,0.25) .. controls +(14:1) and \control{(3.5,-0.5)}{(140:0.5)}; 
\node [blue!80!white] at (0,-0.5) {$\Sigma'$};

%Coordinate
\coordinate (h) at (1,1.95);
\coordinate (l) at (1, 1.55); 
\coordinate (i) at (0.5, 1.65);
\coordinate (e) at (1.5, 1.9);

%Filling
\begin{scope}
\clip (h) .. controls +(180:0.2) and \control{(0,2)}{(-36.75:0.2)} .. controls +(143.25:0.1) and \control{(-0.25,2.15)}{(0:0.1)} .. controls +(180:0.35) and \control{(0,0)}{(143.25:1)} -- (1,-0.75) .. controls +(-36.75:0.1) and \control{(1.25,-0.9)}{(180:0.1)} .. controls +(0:0.35) and \control{(1,1.25)}{(-36.75:1)} -- cycle;
\fill [blue!20!white] (1,2.25) .. controls +(143.25:0.1) and \control{(0.75,2.4)}{(0:0.1)} .. controls +(180:0.35) and \control{(1,0.25)}{(143.25:1)} -- (2,-0.5)-- (1,-0.75) -- ++(0, -1) -- ++(-3.5,0) -- ++(0,4.5) -- ++(2.5,0) -- cycle;
\draw (1,-0.75) -- (2,-0.5);
\end{scope}
\fill [blue!20!white] (h) .. controls +(-150:0.1) and \control{(i)}{(143.25:0.2)} .. controls +(-36.75:0.1) and \control{(l)}{(180:0.2)} .. controls +(60:0.2) and \control{(e)}{(-36.75:0.2)} .. controls +(143.25:0.1) and \control{(h)}{(0:0.2)};
\fill [white!40!gray] (-0.25,2.15) .. controls +(0:0.1) and \control{(0,2)}{(143.25:0.1)} .. controls +(-36.75:0.2) and \control{(h)}{(180:0.2)} .. controls +(30:0.1) and \control{(1,2.25)}{(-36.75:0.2)} .. controls +(143.25:0.1) and \control{(0.75,2.4)}{(0:0.1)}-- cycle;
\fill [white!40!gray] (1,1.25) .. controls +(143.25:0.1) and \control{(l)}{(-130:0.2)} .. controls +(180:0.2) and \control{(2,1.5)}{(143.25:0.2)} .. controls +(-36.75:1) and \control{(2.25,-0.65)}{(0:0.35)} -- (1.25,-0.9) .. controls +(0:0.35) and \control{(1,1.25)}{(-36.75:1)};
\begin{scope}
\clip(1,1.25) .. controls +(143.25:0.1) and \control{(l)}{(-130:0.2)} .. controls +(180:0.2) and \control{(2,1.5)}{(143.25:0.2)} .. controls +(-36.75:1) and \control{(2.25,-0.65)}{(0:0.35)} -- (1.25,-0.9) .. controls +(0:0.35) and \control{(1,1.25)}{(-36.75:1)};
%\draw [dotted] (1,-0.75) -- (2,-0.5);
\end{scope}

%Crossings
\begin{scope}
\clip (-0.25,2.15) .. controls +(0:0.1) and \control{(0,2)}{(143.25:0.1)} .. controls +(-36.75:0.2) and \control{(h)}{(180:0.2)} .. controls +(30:0.1) and \control{(1,2.25)}{(-36.75:0.2)} .. controls +(143.25:0.1) and \control{(0.75,2.4)}{(0:0.1)}-- cycle;
\draw [dashed] (0.75,2.4) .. controls +(180:0.35) and \control{(1,0.25)}{(143.25:1)};
\end{scope}
\begin{scope}
\clip (1,1.25) .. controls +(143.25:0.1) and \control{(l)}{(-130:0.2)} .. controls +(180:0.2) and \control{(2,1.5)}{(143.25:0.2)} .. controls +(-36.75:1) and \control{(2.25,-0.65)}{(0:0.35)} -- (1.25,-0.9) .. controls +(0:0.35) and \control{(1,1.25)}{(-36.75:1)};
\draw [dashed] (4.5, -1.25) .. controls +(140:0.5) and \control{(2,-0.5)}{(14:1)} -- (1,-0.75);
\end{scope}
\begin{scope}
\clip (h) .. controls +(180:0.2) and \control{(0,2)}{(-36.75:0.2)} .. controls +(143.25:0.1) and \control{(-0.25,2.15)}{(0:0.1)} .. controls +(180:0.35) and \control{(0,0)}{(143.25:1)} -- (1,-0.75) .. controls +(-36.75:0.1) and \control{(1.25,-0.9)}{(180:0.1)} .. controls +(0:0.35) and \control{(1,1.25)}{(-36.75:1)} -- cycle;
\draw (0.75,2.4) .. controls +(180:0.35) and \control{(1,0.25)}{(143.25:1)};
\end{scope}
\draw (1,2.25) .. controls +(143.25:0.1) and \control{(0.75,2.4)}{(0:0.1)};
\draw (2.25,-0.65) .. controls +(0:0.35) and \control{(2,1.5)}{(-36.75:1)};
\draw (0,2) .. controls +(143.25:0.1) and \control{(-0.25,2.15)}{(0:0.1)} .. controls +(180:0.35) and \control{(0,0)}{(143.25:1)} (1,-0.75) (1.25,-0.9) .. controls +(0:0.35) and \control{(1,1.25)}{(-36.75:1)} ;
\begin{scope}[even odd rule] 
\clip (0, 2.25) rectangle (3, 1) (l) circle (0.1cm);
\draw (1,1.25) .. controls +(143.25:0.1) and \control{(l)}{(-130:0.2)} .. controls +(60:0.2) and \control{(e)}{(-36.75:0.2)};
\end{scope}
\draw (e) .. controls +(143.25:0.1) and \control{(h)}{(0:0.2)} .. controls +(180:0.2) and \control{(0,2)}{(-36.75:0.2)};
\begin{scope}[even odd rule] 
\clip (0,2) -- (1,1.25) -- (2,1.5) -- (1,2.25) -- (0,2) (h) circle (0.13cm);
\draw (1,2.25) .. controls +(-36.75:0.2) and \control{(h)}{(30:0.1)} .. controls +(-150:0.1) and \control{(i)}{(143.25:0.2)};
\end{scope}
\draw (i) .. controls +(-36.75:0.1) and \control{(l)}{(180:0.2)} .. controls +(0:0.2) and \control{(2,1.5)}{(143.25:0.2)};
\node at (1.85,0.75) {$H$};
\draw [purple] (0,0) -- (1,0.25);
\draw (2.25,-0.65) -- (1.25,-0.9);
\node at (0,0) [above right, purple] {$\gamma_2$};
\end{scope}

\node at (8.75,-0.5) {{\huge $=$}};

\begin{scope}[xshift = 10cm]
\clip (-0.75,-2) rectangle (2.75,2.5);
%Band filling
\fill [blue!20!white] (-2,-1) .. controls +(70:0.5) and \control{(0,0)}{(-166:0.5)} -- (1,0.25) .. controls +(14:1) and \control{(3.5,-0.5)}{(140:0.5)} -- (4.5,-1.25) .. controls +(140:0.5) and \control{(2,-0.5)}{(14:1)} -- (1,-0.75) .. controls +(-166:0.5) and \control{(-1,-1.75)}{(70:0.5)} -- cycle;
\draw (-2,-1) .. controls +(70:0.5) and \control{(0,0)}{(-166:0.5)} -- (1,0.25) .. controls +(14:1) and \control{(3.5,-0.5)}{(140:0.5)}; 
\draw [xshift = 1cm, yshift = -0.75cm] (-2,-1) .. controls +(70:0.5) and \control{(0,0)}{(-166:0.5)} -- (1,0.25) .. controls +(14:1) and \control{(3.5,-0.5)}{(140:0.5)}; 
\node [blue!80!white] at (0,-0.5) {$\Sigma(T_1)$};
\end{scope}

\end{tikzpicture}
\caption{Deleting a leaf yields a surface-minor.}
\label{pic_lem_delete_leaf}
\end{center}
\end{figure}
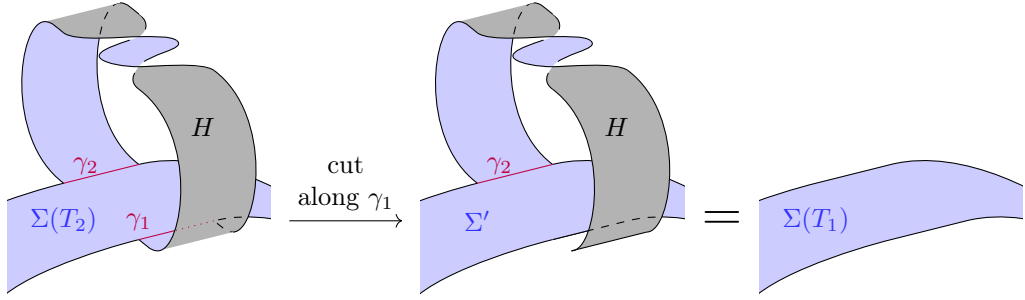

A very similar proof yields the following lemma.

\begin{lemma}\label{lem:minor2}
Assume that $T_1$ is a plane tree whose root has only one child, $T_2$ is the subtree rooted at that child. Then $\Sigma (T_1)$ is an incompressible subsurface of $\Sigma (T_2)$.
\end{lemma}

\begin{proof}
The proof is identical to the previous one: in that situation, $\Sigma (T_1)$ is obtained from~$\Sigma (T_2)$ by plumbing a Hopf band, and cutting along one of the boundaries of the plumbing disc provides the needed incompressible subsurface, as in Figure~\ref{pic_lem_delete_leaf}.
\end{proof}

\begin{lemma}\label{lem:minor3}
Assume that $T_2$ is a plane tree in which $u\to v \to w$ are three consecutive vertices where $v$ has degree~$2$, and that
$T_1$ is obtained from $T_2$ by contracting $u\to v\to w$ into a single edge~$u\to w$, while preserving the labels of the endpoints. 
Then $\Sigma (T_1)$ is an incompressible subsurface of $\Sigma (T_2)$. 
\end{lemma}

\begin{proof}
  By the construction of Hopf arborescent surfaces, the edge between $u$ and $v$ in~$T_2$ corresponds to a plumbing rectangle~$D$. It is important to recall here our orientation convention: if $u$ is labelled positively, the cores $\alpha(u)$ and $\alpha(v)$ are oriented so that one turns to the left when going from $u$ to $v$ at the rectangle $D$, while if $v$ is labelled negatively, one turns to the right, see Figure~\ref{pic_orient_plumbing}. Now, we consider two diagonal arcs $\gamma_1$ and $\gamma_2$ on the plumbing rectangle $D$ as pictured in Figure~\ref{pic_cut_diago}. When cutting along such a diagonal arc, we obtain a new surface in which the cores $\alpha(u)$ and $\alpha(v)$ merge into a single core. However, their orientations might mismatch, depending on whether we cut along $\gamma_1$ or $\gamma_2$. We take the convention that $\gamma_1$ is the arc the preserves the orientations, while $\gamma_2$ induces an orientation mismatch, see Figures~\ref{pic_cut_diago} and~\ref{pic_contract_path}.

  Now, let us first consider the case where the labels of $u$ and $v$ are the same. In this case, we consider the subsurface $\Sigma'$ of $\Sigma (T_2)$ obtained by cutting along $\gamma_1$. This has the effect of merging the core curves $\alpha(u)$ and $\alpha(v)$  in a way that respects their orientations. However, it might seem that since each curve $\alpha(u)$ and $\alpha(v)$ corresponds to a Hopf band, merging them like that yields a band that twists too much. But a key observation is that cutting along $\gamma_1$ adds a twist between these two bands, as pictured in Figure~\ref{pic_cut_diago}, and this twist is negative when the bands are positive, while it is positive when the bands are negative (indeed, this is the reason for our orientation convention). Therefore, the resulting surface $\Sigma'$ is exactly the same as the one corresponding to the tree $T_1$, and therefore $\Sigma (T_1)$ is an incompressible subsurface of $\Sigma (T_2)$ by Lemma~\ref{lem_cut}. See the top and bottom pictures of Figure~\ref{pic_contract_path} for an illustration.

  Now, let us consider the case where the label of $u$ is $+$ while the label of $v$ is $-$. In that case, we consider the surface $\Sigma'$ of $\Sigma (T_2)$ obtained by cutting along $\gamma_2$. This has the effect of merging the core curves $\alpha(u)$ and $\alpha(v)$ but with an orientation mismatch. We take the convention that the resulting core curve $\alpha'$ is oriented by $\alpha(u)$, and therefore disagrees with the orientation of $\alpha(v)$ while it follows it. Since $u$ and $v$ are labelled with opposite signs the two twists on their Hopf bands cancel out, but cutting along $\gamma_2$ adds a new positive twist, therefore we can consider $\alpha'$ as being the core curve of a positive Hopf band. Now, let us consider the plumbing rectangle $D'$ corresponding to the edge between $v$ and $w$. Due to the orientation mismatch, arriving at this rectangle from $\alpha'$, we are oriented in the direction opposed to the one we would arrive with if we were arriving from $\alpha(v)$. But due to the orientation convention, when going from $\alpha(v)$ to $\alpha(w)$ in $\Sigma (T_2)$ we turn to the right since $v$ is negative, while when going from $\alpha'$ to $\alpha(w)$ in $\Sigma'$ we turn to the left since $\alpha'$ is a positive band. Therefore, this effect cancels out the orientation mismatch, and $\Sigma'$ coincides exactly with the surface $\Sigma (T_1)$ corresponding to the tree $T_1$. Therefore $\Sigma (T_1)$ is an incompressible subsurface of $\Sigma (T_2)$ by Lemma~\ref{lem_cut}. See the third picture of Figure~\ref{pic_contract_path} for an illustration.

  The same cancellation effect happens when the label of $u$ is $-$ and the label of $v$ is $+$: when cutting along $\gamma_2$ we have an orientation mismatch which is cancelled out by the the fact that the new band is negative, and thus the orientation convention makes it turn in the opposite direction in the plumbing rectangle between $v$ and $w$. Therefore, in that case $\Sigma (T_1)$ is also an incompressible subsurface of $\Sigma (T_2)$ thanks to Lemma~\ref{lem_cut}. This is illustrated in the second picture of Figure~\ref{pic_contract_path}.
\end{proof}

\begin{figure}
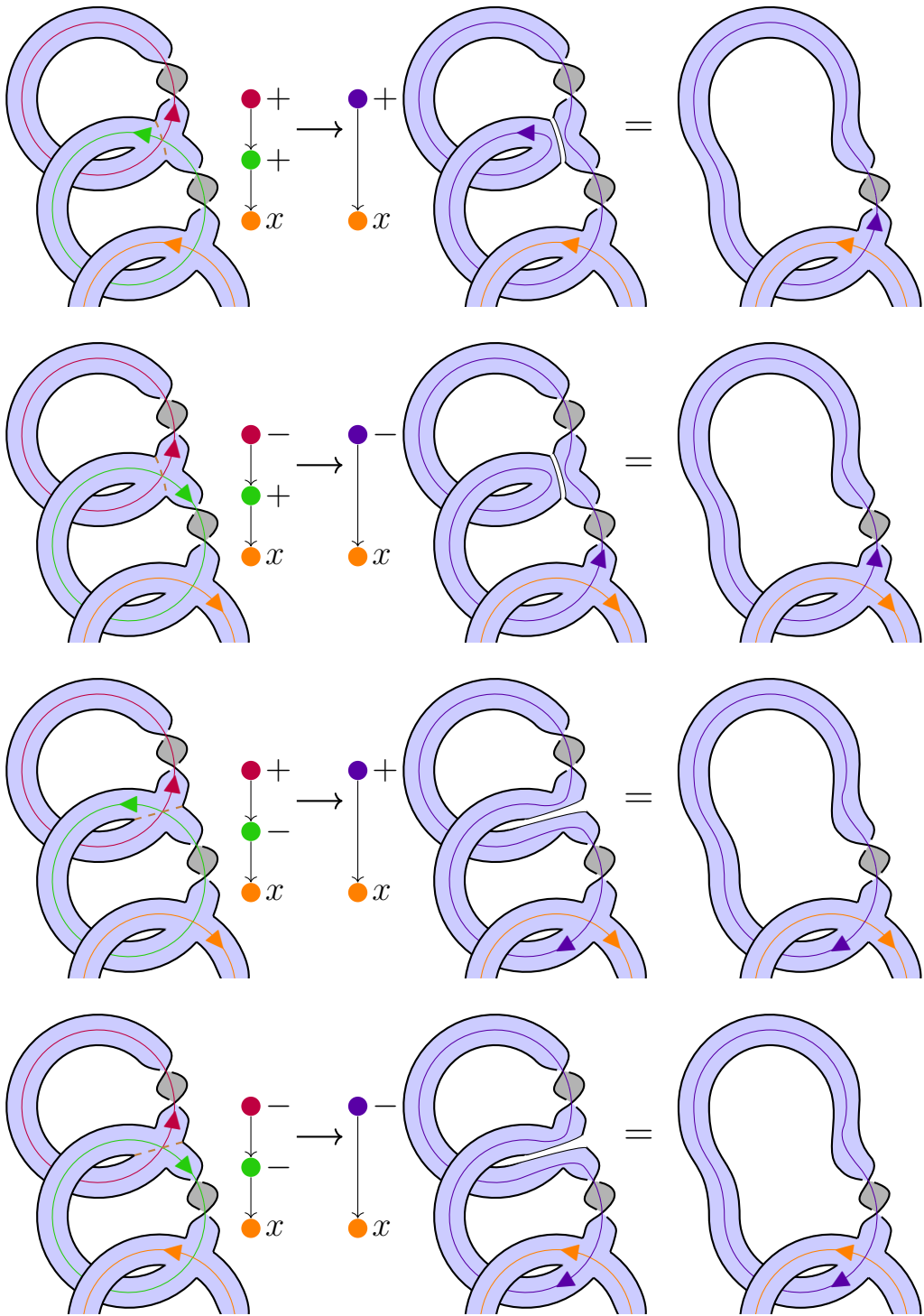

\begin{center}

\caption{All cases of contraction of a $3$-path to an edge preserving the labels of the endpoints.}
\label{pic_contract_path}
\end{center}
\end{figure}

As a corollary, contracting any path of $T_1$ into an edge whose labels match with the labels of the extremities of the path produces a tree $T_2$ such that $\Sigma (T_1) \preccurlyeq \Sigma (T_2)$.

\begin{proof}[Proof of Proposition~\ref{P:mainminor}]
By definition, if $T_1$ admits a homeomorphic embedding into $T_2$, it can be obtained iteratively from $T_2$ by (i) removing a child leaf, (ii) removing a parent leaf, (iii) reducing a label or (iv) contracting a path while preserving the labels of the endpoints. Since no two elements on the alphabet $\{+,-\}$ are comparable, case (iii) cannot happen. Then the cases (i), (ii) and (iv) are handled respectively by Lemma~\ref{lem:minor1}, Lemma~\ref{lem:minor2} and Lemma~\ref{lem:minor3}.\end{proof}

On the other hand, the Kruskal Tree Theorem directly yields the following proposition.

\begin{proposition}\label{prop_link_wqo}
Hopf arborescent links are well-quasi-ordered under the link-minor relation.
\end{proposition}

\begin{proof}
Take an infinite sequence $(L_n)_{n \in \mathbb{N}}$ of Hopf arborescent links, let $(T_n)_{n \in \mathbb{N}}$ a sequence of plane trees such that for all $n \in \N,$ $\Sigma (T_n)$ is a Seifert surface of $L_n$. Then, by Theorem~\ref{T:kruskal}, there exists $i<j$ such that $T_i$ admits a homeomorphic embedding into $T_2$. Hence $L_i$ is a link-minor of $L_j$.
\end{proof}

We can deduce Theorem~\ref{T:wqo} as a direct corollary of Proposition~\ref{P:mainminor} and Proposition~\ref{prop_link_wqo}.

\begin{proof}[Proof of Theorem~\ref{T:wqo}]
  Take an infinite sequence $(\Sigma_n)_{n \in \mathbb{N}}$ of canonical Seifert surfaces of Hopf arborescent links. Then by Proposition~\ref{prop_link_wqo}, $\partial \Sigma_i \hookrightarrow \partial \Sigma_j$ for some $i<j$. By Proposition~\ref{P:mainminor} we have $\Sigma_i \preccurlyeq \Sigma_j$, i.e. the surface-minor order is a well-quasi-order on Hopf arborescent surfaces.
 \end{proof}

\begin{remark}\label{R:minors}
The proof of Proposition~\ref{P:mainminor} highlights that the minor relation on the set of Hopf arborescent surfaces is more subtle and fragile than one might expect. 
Indeed, the cuts involved when taking an incompressible subsurface in the proof of Lemma~\ref{lem:minor3} inevitably merge Hopf bands and thus one needs to be careful in order to control the number of resulting twists. 
In particular, the proof does not seem to generalise to the more general classes of surfaces obtained by plumbing bands with a bounded number of twists (even though everything works well at the level of trees). 
\end{remark}

\section{Decidability of the genus defect for Hopf arborescent links}\label{sec_decidability}

\subsection{Monotonicity of the genus defect}

Now that we proved that link-minor is a well-quasi-order on the set of Hopf arborescent links, we want to highlight a property that is stable for this minor relation. 
Recall that the genus defect~$\Delta_g(L)$ of an oriented link~$L$ is the difference $g(L)-g_4(L)$ between its classical genus and its 4-dimensional genus. 
The latter can be either in the topological locally flat or in the smooth category. 
All statements in this section (and in particular Theorem~\ref{T:main}) hold in both categories. 
We reprove Lemma~6 of \cite{baader2018topological} in the form of Proposition~\ref{prop_stable_defect} using the fact that link-minor implies that the associated Seifert are surface-minors.

\begin{proposition}\label{prop_stable_defect}
The genus defect $\Delta_g$ is monotone on the family of Hopf arborescent links with respect to the link-minor relation, i.e., if $L_1$ is a link-minor of $L_2$, then $\Delta g(L_1) \leq \Delta g (L_2)$.
\end{proposition}

We rely on the following lemma that highlights how the $4$-genus behaves with respect to surface-minors. It uses a cut-and-paste construction and an Euler characteristic argument.

\begin{lemma}\label{lem_defect}
Let $\Sigma$ be an oriented surface of $\Sp^3$ and $\Sigma'$ be a surface-minor of $\Sigma$. If we write  $L = \partial \Sigma$ and $L' = \partial \Sigma'$, then we have $g(\Sigma) - g_4(L) \geq g( \Sigma' ) - g_4(L')$.
\end{lemma}

\begin{proof}
Seeing $\Sp^3$ as the boundary of the $4$-ball $\B^4$, consider a surface $S'$ in $\B^4$ such that $\partial S' = L'$ and $S' \cap \Sp^3 = L'$. 
Gluing the remaining pieces of $\Sigma \smallsetminus \Sigma'$ to $S'$ along $L'$ yields a surface $S$ in $B^4$ such that $\partial S = L$, see Figure~\ref{pic_proof_defect_monotonous}. 

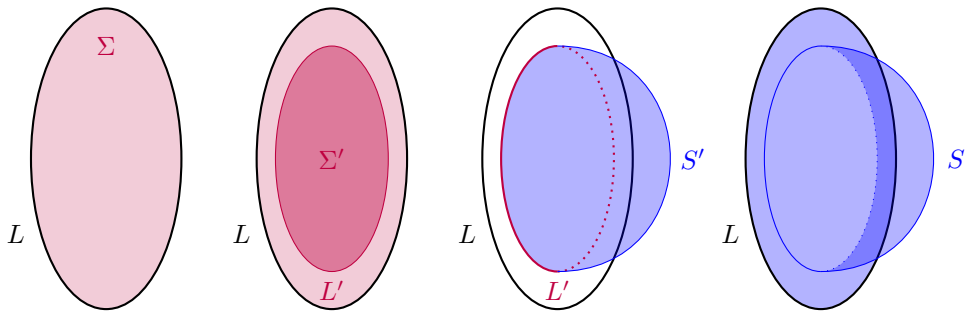
\begin{figure}[H]
\begin{center}
\begin{tikzpicture}
\fill [opacity = 0.2, purple] (0,0) circle (1 and 2);
\draw [black, thick] (0,0) circle (1 and 2);
\node at (-1.2,-1) {$L$};
\node [purple] at (0,1.5) {$\Sigma$};

\begin{scope}[xshift = 3cm]
\fill [opacity = 0.2, purple] (0,0) circle (1 and 2);
\draw [black, thick] (0,0) circle (1 and 2);
\node at (-1.2,-1) {$L$};
\filldraw [purple, fill opacity = 0.4] (0,0) circle (0.75 and 1.5);
\node [purple] at (0,0) {$\Sigma'$};
\node [purple] at (0,-1.75) {$L'$};
\end{scope}

\begin{scope}[xshift = 6cm]
%\fill [opacity = 0.2, purple] (0,0) circle (1 and 2);
\draw [black, thick] (0,0) circle (1 and 2);
\node at (-1.2,-1) {$L$};
\node [purple] at (0,-1.75) {$L'$};
\fill[opacity = 0.3, blue] (0,1.5) arc (90:270:0.75 and 1.5) .. controls +(2,0) and \control{(0,1.5)}{(2,0)};
\draw [purple, thick] (0,1.5) arc (90:270:0.75 and 1.5);
\draw [purple, thick, dotted] (0,-1.5) arc (-90:90:0.75 and 1.5);
\draw [blue] (0,-1.5) .. controls +(2,0) and \control{(0,1.5)}{(2,0)};
\node at (1.8,0) {${\color{blue} S'}$};
\end{scope}

\begin{scope}[xshift = 9.5cm]
\fill [opacity = 0.3, blue] (0,0) circle (1 and 2);
\fill [white] (0,0) circle (0.75 and 1.5);
\draw [black, thick] (0,0) circle (1 and 2);
\node at (-1.2,-1) {$L$};
\fill[opacity = 0.3, blue] (0,1.5) arc (90:270:0.75 and 1.5) .. controls +(2,0) and \control{(0,1.5)}{(2,0)};
\draw [blue] (0,1.5) arc (90:270:0.75 and 1.5);
\draw [blue, dotted] (0,-1.5) arc (-90:90:0.75 and 1.5);
\draw [blue] (0,-1.5) .. controls +(2,0) and \control{(0,1.5)}{(2,0)};
\node [blue] at (1.8,0) {$S$};
\end{scope}
\end{tikzpicture}
\caption{\label{pic_proof_defect_monotonous} Illustration of the construction of $\Sigma$, $\Sigma'$, $L'$, $S,$ and $S'$. Purple surfaces are in $\Sp^3$ while blue one are considered to be in $B^4$.}
\end{center}
\end{figure}

By definition of $g_4(L)$, we have $g_4(L) \leq g(S)$, and thus $g( \Sigma ) - g_4( L ) \geq g( \Sigma ) - g( S )$.
Furthermore, the genus of $S$ is given by $g(S) = g(\Sigma) - g(\Sigma') + g(S')$. Thus one has $g( \Sigma ) - g_4( L ) \geq g( \Sigma' ) - g( S' )$. 
Now if we assume $S'$ to minimise the 4-genus over surfaces bounded by~$L'$, we conclude: $g(\Sigma) - g_4(L) \geq g( \Sigma' ) -g_4(L')$.
\end{proof}

As Hopf arborescent surfaces are of minimal genus for Hopf arborescent links, Lemma~\ref{lem_defect} can be used to prove Proposition~\ref{prop_stable_defect}.

\begin{proof}[Proof of Proposition~\ref{prop_stable_defect}]
Let $L$ and $L'$ be two Hopf arborescent links such that $L'$ is a link-minor of $L$. 
Consider $\Sigma$ and $\Sigma'$ the corresponding canonical Seifert surfaces. 
Since Hopf arborescent links are fibred, $\Sigma$ is a Seifert surface of $L$ of minimal genus (see~Theorem~\ref{th_fibred_surface_minimise}), i.e. $g(L) = g(\Sigma)$ and similarly, $g(\Sigma') = g(L')$. 
By Proposition~\ref{P:mainminor} one has $\Sigma \preccurlyeq \Sigma'$. 
Hence, by Lemma~\ref{lem_defect}, one gets $\Delta_g(L) \geq \Delta_g(L')$.
\end{proof}

\subsection{Proof of Theorem~\ref{T:main}}\label{S:MainProof}

We first show that the link-minor relation can be decided using the decidability of link equivalence. For knots, equivalence can be tested as a combination of an algorithm that allows to decide whether two $3$-manifolds with boundary are homeomorphic~\cite{Kuperberg_elem_rec,book_Matveev} and the Gordon-Luecke Theorem \cite{Gordon_Luecke} that states that two knots are equivalent if their complements, which are $3$-manifolds with boundaries, are equivalent. In the case of links, we additionally need to keep track of a longitude of each torus boundary component in the complement of the link. We refer to the survey of Lackenby \cite[Section~2]{Lackenby_knot_survey} for a summary of the techniques that allow to prove the following theorem:

\begin{theorem}[Link equivalence]\label{th_link_equiv}
Given two links $L_1$ and $L_2$, the problem of testing whether~$L_1$ is ambient isotopic to $L_2$ is decidable.
\end{theorem}

Given a Hopf arborescent link $L$, denote by $\mathcal{T}(L)$ the set of plane trees $T$ whose associated Hopf arborescent surface $\Sigma(T)$ has $L$ as oriented boundary. 
As a corollary, we obtain:

\begin{lemma}\label{lem_find_tree} 
Given a Hopf arborescent link $L$, the set $\mathcal{T}(L)$ is computable.  
\end{lemma}

\begin{proof}
For increasing $k\in \N$, we enumerate all plane trees $T_i$ with $k$ vertices labelled by $\ens{-,+}$ and store the trees $T_i$ such that $\partial \Sigma (T_i)$ is isotopic to $L$, where we test isotopy using Theorem~\ref{th_link_equiv}. 
If we find a $k$ for which such a tree exists, we finish the enumeration for this value and return the stored trees. 
Indeed, plumbing $n$ Hopf bands produces a surface with Betti number $n$. 
By Theorem~\ref{th_fibred_surface_minimise}, all the trees $T_i$ such that $\partial \Sigma (T_i) = L$ produce surfaces $\Sigma(T_i)$ with the same genus, hence have the same number of vertices. 
As the entry is a Hopf arborescent link, there exists a tree $T$ such that $\partial \Sigma (T) = L$. So the algorithm terminates.
\end{proof}

Alternatively, and if one wants some control on the complexity of that algorithm, one can avoid blindly testing for increasing $k$ by first computing the genus of the link~\cite{Hass_trivial_NP,book_Matveev}, or just computing an upper bound to it using, e.g., Seifert's algorithm, and then enumerate only the trees that produce surfaces up to that genus. From Lemma~\ref{lem_find_tree}, we obtain:

\begin{lemma}\label{lem_lminor_test}
Given two Hopf arborescent links $L_1$ and $L_2$, testing if $L_1$ is a link-minor of~$L_2$ is decidable.
\end{lemma}

\begin{proof}
Using Lemma~\ref{lem_find_tree}, we compute $\mathcal{T}(L_1)$ and $\mathcal{T}(L_2)$. 
The trees in~$\mathcal{T}(L_1)$ (resp. $\mathcal{T}(L_2)$) all have the same number $k_1$ (resp. $k_2$) of vertices. 
Then we brute force every possible path contraction to an edge and iterated leaf deletion on trees of $\mathcal{T}(L_2)$ such that the result is a tree with $k_1$ vertices, and test whether it is equal to a tree of $\mathcal{T}(L_1)$. If such a test succeeds, we output yes, otherwise we return no. 
There is a finite number of trees in both $\mathcal{T}(L_1)$ and~$\mathcal{T}(L_2)$ and a finite number of a trees with $k_1$ vertices that homeomorphically embed into a tree of $\mathcal{T}(L_2)$. 
Hence that algorithm eventually terminates. Its correctness follows directly from the definition of link-minor.
\end{proof}

Finally we prove Theorem~\ref{T:main} by using the stability of the genus-defect by link-minor, the previous algorithms, and the well-quasi-order properties.

\begin{proof}[Proof of Theorem~\ref{T:main}]

  By Proposition~\ref{prop_link_wqo}, the order defined by link-minors is a well-quasi-order on the set of Hopf arborescent links. Hence, the set of Hopf arborescent links $\mathcal{H}_{k}$ that have defect at most $k$ is characterized by a finite family $\mathcal{F}_k$ of forbidden minors. It follows, by Proposition~\ref{prop_stable_defect} that $\Delta_g (L) \leq k$ if and only if for all $f$ in $\mathcal{F}_k$,  $f$ is not a link-minor of $L$. Using Lemma~\ref{lem_lminor_test} we test for each $f \in \mathcal{F}_{k}$ if $f$ is a link-minor of $L$. If such a test succeeds, output no, otherwise the input link has $\Delta_g (L) \leq k$.
\end{proof}

As said in the introduction, our proof in not constructive as it relies at its core on the existence of a set of forbidden minors for having defect at most $k$. This set of forbidden minors is not explicit and hard-coded in the algorithm. Furthermore, the sets of excluded minors will be different for the two different notions of defect (smooth and locally flat). It is likely that computing them is a topological challenge requiring arguments of different nature. 

Theorem~\ref{T:wqo} provides the existence of a set of forbidden minors for having defect at most $k$ but for a different and stronger definition of minors on links that relies only on the surface-minor relation on the Seifert surface and not the trees. However deciding this relation, even by a brute force argument, seems challenging: in addition to the fact that no algorithm seems to be known for testing isotopy of surfaces, one would also need to control the complexity of the cutting arcs. Even with positive Hopf arborescent links only, this seems delicate~\cite{Misev_cutting}.

\section{Examples: Hopf arborescent links with non-trivial defect}\label{sec_defect}

It is not clear a priori that the topological and/or smooth defects of  Hopf arborescent links are nonzero. For instance, our building block, the Hopf band, has both defects equal to $0$ since it bounds an annulus in $\Sp^3$; which has genus $0$. Furthermore, as mentioned in the introduction, when Hopf arborescent links are only made with positive Hopf bands, they belong to a class of links called \emph{positive} links, which implies that they are \emph{strongly quasi-positive}~\cite{Rudolph_positive}, and this implies in turn that their smooth 4-genus equal their 3-genus \cite{rudolph1993}. So for this class of links, the smooth defect is always zero. In contrast, in this section, we provide an example of a Hopf arborescent knot for which both the topological and the smooth defect are nonzero, and an argument explaining how to use this knot to provide examples with arbitrarily large defects.

\subparagraph*{Example 1.} A Hopf arborescent link with non trivial topological defect is given by Baader, Feller, Lewark and Liechti~\cite[Example~4]{baader2018topological}. It consists in plumbing $6$ positive Hopf bands in a latin cross-like pattern. Since it is positive, its smooth defect is zero. However, one can find a torus sub-surface of~$S_X$ such that the restricted Alexander polynomial vanishes. By Freedman's Disc Theorem~\cite{freedman1982topology}, this implies that one can remove this subsurface and glue along its boundary a locally flat disc in the 4-ball. This means that the topological defect is at least 1. We refer to the article~\cite{baader2018topological} for more details.

\subparagraph*{Example 2.} By the aforementioned result of Rudolph, in order to obtain a Hopf arborescent link with nonzero smooth defect, we need to use both positive and negative Hopf bands. Here is such an example: the Hopf arborescent link obtained from the tree pictured in Figure~\ref{pic_knot_8_10} is in fact a knot, named $8_{10}$ in the Rolfsen table~\cite{Rolfsen_Knots}. One can readily check from knot censuses (for example from the Knot Atlas~\cite{knotatlas}) the smooth and topological genera of this knot, and we have $\Delta_g^{smooth} (8_{10}) = \Delta_g^{top} (8_{10}) = 2$. For completeness, we provide here another argument to explain why the defects are nonzero.

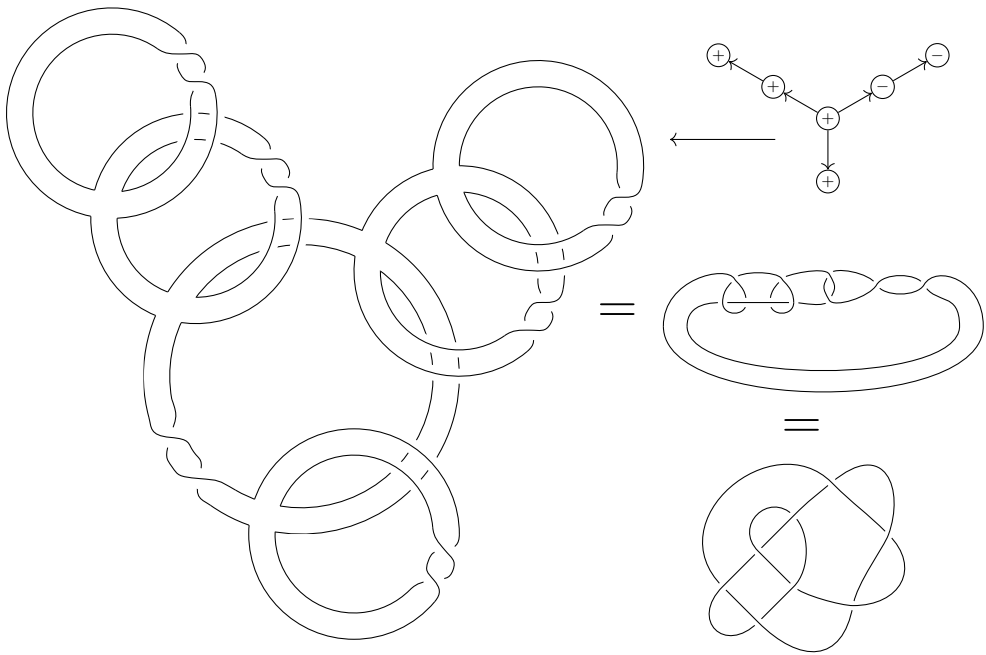
\begin{figure}[h]
\begin{center}
\begin{tikzpicture}[scale = 1.4]
\clip (-3,-2.75) rectangle (7,3.75);

%Large circle 1
\begin{scope}[even odd rule]
\clip [xshift =0.5cm, yshift = -1.5cm] (-3,-3) rectangle (3,3) (-135:0.75) arc (-135:-220:0.75) -- (-220:1) arc (-220:-135:1) --cycle;
\clip [xshift =-1cm, yshift = 1.5cm] (-3,-3) rectangle (3,3) (-80:0.75) arc (-80:-160:0.75) -- (-160:1) arc (-160:-80:1) --cycle;
\clip [xshift =1.5cm, yshift = 1cm] (-3,-3) rectangle (3,3) (80:0.75) arc (80:190:0.75) -- (190:1) arc (190:80:1) --cycle;
\draw [rotate = 125] (115:1.25) arc (115:425:1.25);
\draw [rotate = 125] (110:1.5) arc (110:430:1.5);
\end{scope}
\begin{scope}[rotate = 125]
\coordinate (l) at (80:1.375); 
\coordinate (h) at (100:1.375); 
\coordinate (i) at (90:1.25); 
\coordinate (e) at (90:1.5);
\begin{scope}[even odd rule]
\clip (-1,0.5) rectangle (1,2) (l) circle (0.1cm);
\draw (115:1.25) .. controls +(25:0.2) and \control{(h)}{(-130:0.2)} .. controls +(50:0.2) and \control{(e)}{(180:0.1)} .. controls +(0:0.1) and \control{(l)}{(130:0.2)} .. controls +(-50:0.2) and \control{(65:1.25)}{(165:0.2)}; 
\end{scope}
\begin{scope}[even odd rule]
\clip (-1,0.5) rectangle (1,2) (h) circle (0.1cm);
\draw (110:1.5) .. controls +(20:0.1) and \control{(h)}{(130:0.2)} .. controls +(-50:0.2) and \control{(i)}{(180:0.1)} .. controls +(0:0.1) and \control{(l)}{(-130:0.2)} .. controls +(50:0.2) and \control{(70:1.5)}{(160:0.1)}; 
\end{scope}
\end{scope}

%Circle 2
\def\r{30} \def\xs{-1} \def\ys{1.5}
\begin{scope}[even odd rule]
\clip (-3,-3) rectangle (3,3) (110:1.25) arc (110:180:1.25) -- (180:1.5) arc (180:110:1.5) --cycle;
\clip (-3,-3) rectangle (3,6) [xshift = -1.8cm, yshift = 2.5cm] (-60:0.75) arc (-60:-120:0.75) -- (-120:1) arc (-120:-60:1) --cycle;
\draw [xshift = \xs cm, yshift = \ys cm, line width=0.2cm, white] (10:0.75) arc (10:-70:0.75);
\draw [xshift = \xs cm, yshift = \ys cm, rotate = \r] (25:0.75) arc (25:335:0.75);
\draw [xshift = \xs cm, yshift = \ys cm, line width=0.2cm, white] (10:1) arc (10:-70:1);
\draw [xshift = \xs cm, yshift = \ys cm, rotate = \r] (20:1) arc (20:340:1);
\end{scope}
\begin{scope}[xshift = \xs cm, yshift = \ys cm, rotate = \r]
\coordinate (l) at (-10:0.875); 
\coordinate (h) at (10:0.875); 
\coordinate (i) at (0:0.75); 
\coordinate (e) at (0:1);
\begin{scope}[even odd rule]
\clip (-1.5,-1) rectangle (1.5,1) (l) circle (0.1cm);
\draw (25:0.75) .. controls +(-65:0.1) and \control{(h)}{(150:0.1)} .. controls +(-30:0.1) and \control{(e)}{(90:0.1)} .. controls +(-90:0.1) and \control{(l)}{(30:0.1)} .. controls +(-150:0.1) and \control{(-25:0.75)}{(65:0.1)}; 
\end{scope}
\begin{scope}[even odd rule]
\clip (-1.5,-1) rectangle (1.5,1)  (h) circle (0.1cm);
\draw (20:1) .. controls +(-70:0.1) and \control{(h)}{(30:0.1)} .. controls +(-150:0.1) and \control{(i)}{(90:0.1)} .. controls +(-90:0.1) and \control{(l)}{(150:0.1)} .. controls +(-30:0.1) and \control{(-20:1)}{(70:0.1)}; 
\end{scope}
\end{scope}

%Circle 3
\def\r{30} \def\xs{-1.8} \def\ys{2.5}
\begin{scope}[even odd rule]
\clip (-3,-3) rectangle (3,6) [xshift = -1cm, yshift = 1.5cm] (110:0.75) arc (110:190:0.75) -- (190:1) arc (190:110:1) --cycle;
\draw [xshift = \xs cm, yshift = \ys cm, line width=0.2cm, white] (10:0.75) arc (10:-70:0.75);
\draw [xshift = \xs cm, yshift = \ys cm, rotate = \r] (25:0.75) arc (25:335:0.75);
\draw [xshift = \xs cm, yshift = \ys cm, line width=0.2cm, white] (10:1) arc (10:-70:1);
\draw [xshift = \xs cm, yshift = \ys cm, rotate = \r] (20:1) arc (20:340:1);
\end{scope}
\begin{scope}[xshift = \xs cm, yshift = \ys cm, rotate = \r]
\coordinate (l) at (-10:0.875); 
\coordinate (h) at (10:0.875); 
\coordinate (i) at (0:0.75); 
\coordinate (e) at (0:1);
\begin{scope}[even odd rule]
\clip (-1.5,-1) rectangle (1.5,1) (l) circle (0.1cm);
\draw (25:0.75) .. controls +(-65:0.1) and \control{(h)}{(150:0.1)} .. controls +(-30:0.1) and \control{(e)}{(90:0.1)} .. controls +(-90:0.1) and \control{(l)}{(30:0.1)} .. controls +(-150:0.1) and \control{(-25:0.75)}{(65:0.1)}; 
\end{scope}
\begin{scope}[even odd rule]
\clip (-1.5,-1) rectangle (1.5,1)  (h) circle (0.1cm);
\draw (20:1) .. controls +(-70:0.1) and \control{(h)}{(30:0.1)} .. controls +(-150:0.1) and \control{(i)}{(90:0.1)} .. controls +(-90:0.1) and \control{(l)}{(150:0.1)} .. controls +(-30:0.1) and \control{(-20:1)}{(70:0.1)}; 
\end{scope}
\end{scope}

%Circle 4
\def\r{-30} \def\xs{1.5} \def\ys{1}
\begin{scope}[even odd rule]
\clip (-3,-3) rectangle (3,3) (90:1.25) arc (90:30:1.25) -- (30:1.5) arc (30:90:1.5) --cycle;
\clip (-3,-3) rectangle (6,6) [xshift = 2.25cm, yshift = 2cm] (-120:0.75) arc (-120:-240:0.75) -- (-240:1) arc (-240:-120:1) --cycle;
\draw [xshift = \xs cm, yshift = \ys cm, line width=0.2cm, white] (10:0.75) arc (10:-150:0.75);
\draw [xshift = \xs cm, yshift = \ys cm, rotate = \r] (25:0.75) arc (25:335:0.75);
\draw [xshift = \xs cm, yshift = \ys cm, line width=0.2cm, white] (10:1) arc (10:-150:1);
\draw [xshift = \xs cm, yshift = \ys cm, rotate = \r] (20:1) arc (20:340:1);
\end{scope}
\begin{scope}[xshift = \xs cm, yshift = \ys cm, rotate = \r]
\coordinate (l) at (-10:0.875); 
\coordinate (h) at (10:0.875); 
\coordinate (i) at (0:0.75); 
\coordinate (e) at (0:1);
\begin{scope}[even odd rule]
\clip (-1.5,-1) rectangle (1.5,1) (h) circle (0.1cm);
\draw (25:0.75) .. controls +(-65:0.1) and \control{(h)}{(150:0.1)} .. controls +(-30:0.1) and \control{(e)}{(90:0.1)} .. controls +(-90:0.1) and \control{(l)}{(30:0.1)} .. controls +(-150:0.1) and \control{(-25:0.75)}{(65:0.1)}; 
\end{scope}
\begin{scope}[even odd rule]
\clip (-1.5,-1) rectangle (1.5,1)  (l) circle (0.1cm);
\draw (20:1) .. controls +(-70:0.1) and \control{(h)}{(30:0.1)} .. controls +(-150:0.1) and \control{(i)}{(90:0.1)} .. controls +(-90:0.1) and \control{(l)}{(150:0.1)} .. controls +(-30:0.1) and \control{(-20:1)}{(70:0.1)}; 
\end{scope}
\end{scope}

%Circle 5
\def\r{-30} \def\xs{2.25} \def\ys{2}
\begin{scope}[even odd rule]
\clip (-3,-3) rectangle (6,6) [xshift = 1.5cm, yshift = 1cm] (60:0.75) arc (60:120:0.75) -- (120:1) arc (120:60:1) --cycle;
\draw [xshift = \xs cm, yshift = \ys cm, line width=0.2cm, white] (10:0.75) arc (10:-130:0.75);
\draw [xshift = \xs cm, yshift = \ys cm, rotate = \r] (25:0.75) arc (25:335:0.75);
\draw [xshift = \xs cm, yshift = \ys cm, line width=0.2cm, white] (10:1) arc (10:-130:1);
\draw [xshift = \xs cm, yshift = \ys cm, rotate = \r] (20:1) arc (20:340:1);
\end{scope}
\begin{scope}[xshift = \xs cm, yshift = \ys cm, rotate = \r]
\coordinate (l) at (-10:0.875); 
\coordinate (h) at (10:0.875); 
\coordinate (i) at (0:0.75); 
\coordinate (e) at (0:1);
\begin{scope}[even odd rule]
\clip (-1.5,-1) rectangle (1.5,1) (h) circle (0.1cm);
\draw (25:0.75) .. controls +(-65:0.1) and \control{(h)}{(150:0.1)} .. controls +(-30:0.1) and \control{(e)}{(90:0.1)} .. controls +(-90:0.1) and \control{(l)}{(30:0.1)} .. controls +(-150:0.1) and \control{(-25:0.75)}{(65:0.1)}; 
\end{scope}
\begin{scope}[even odd rule]
\clip (-1.5,-1) rectangle (1.5,1)  (l) circle (0.1cm);
\draw (20:1) .. controls +(-70:0.1) and \control{(h)}{(30:0.1)} .. controls +(-150:0.1) and \control{(i)}{(90:0.1)} .. controls +(-90:0.1) and \control{(l)}{(150:0.1)} .. controls +(-30:0.1) and \control{(-20:1)}{(70:0.1)}; 
\end{scope}
\end{scope}

%Circle 6
\def\r{-20} \def\xs{0.5} \def\ys{-1.5}
\begin{scope}[even odd rule]
\clip (-3,-3) rectangle (3,3) (-90:1.25) arc (-90:-125:1.25) -- (-125:1.5) arc (-125:-90:1.5) --cycle;
\draw [xshift = \xs cm, yshift = \ys cm, rotate = \r, line width=0.2cm, white] (40:0.75) arc (40:120:0.75);
\draw [xshift = \xs cm, yshift = \ys cm, rotate = \r] (25:0.75) arc (25:335:0.75);
\draw [xshift = \xs cm, yshift = \ys cm, rotate = \r, line width=0.2cm, white] (60:1) arc (60:120:1);
\draw [xshift = \xs cm, yshift = \ys cm, rotate = \r] (20:1) arc (20:340:1);
\end{scope}
\begin{scope}[xshift = \xs cm, yshift = \ys cm, rotate = \r]
\coordinate (l) at (-10:0.875); 
\coordinate (h) at (10:0.875); 
\coordinate (i) at (0:0.75); 
\coordinate (e) at (0:1);
\begin{scope}[even odd rule]
\clip (-1.5,-1) rectangle (1.5,1) (l) circle (0.1cm);
\draw (25:0.75) .. controls +(-65:0.1) and \control{(h)}{(150:0.1)} .. controls +(-30:0.1) and \control{(e)}{(90:0.1)} .. controls +(-90:0.1) and \control{(l)}{(30:0.1)} .. controls +(-150:0.1) and \control{(-25:0.75)}{(65:0.1)}; 
\end{scope}
\begin{scope}[even odd rule]
\clip (-1.5,-1) rectangle (1.5,1)  (h) circle (0.1cm);
\draw (20:1) .. controls +(-70:0.1) and \control{(h)}{(30:0.1)} .. controls +(-150:0.1) and \control{(i)}{(90:0.1)} .. controls +(-90:0.1) and \control{(l)}{(150:0.1)} .. controls +(-30:0.1) and \control{(-20:1)}{(70:0.1)}; 
\end{scope}
\end{scope}

%Tree
\begin{scope}[xshift = 5cm, yshift = 2.45cm]
\def\d{0.025}
\def\a{0.6}
\node (n1) [draw, inner sep = \d cm, circle] at (0,0) {{\tiny $+$}};
\node (n2) [draw, inner sep = \d cm, circle] at ($(n1)+(150:\a)$) {{\tiny $+$}};
\node (n3) [draw, inner sep = \d cm, circle] at ($(n2)+(150:\a)$) {{\tiny $+$}}; 
\node (n4) [draw, inner sep = \d cm, circle] at ($(n1)+(30:\a)$) {{\tiny $-$}}; 
\node (n5) [draw, inner sep = \d cm, circle] at ($(n4)+(30:\a)$) {{\tiny $-$}};%-
\node (n6) [draw, inner sep = \d cm, circle] at ($(n1)+(-90:\a)$) {{\tiny $+$}};
\draw [->] (n1) -- (n2);
\draw [->] (n1) -- (n4);
\draw [->] (n1) -- (n6);
\draw [->] (n2) -- (n3);
\draw [->] (n4) -- (n5);
\end{scope}

\draw [<-] (3.5,2.25) -- +(1,0);

\node at (3,0.6) {\huge $=$};

\begin{scope}[xshift = 4cm, yshift = 0.7cm, scale = 0.45]
\def\c{0.1}
\coordinate (n1) at (0,0);
\coordinate (n2) at (0.5,0);
\coordinate (n3) at (0.25,0.5);
\coordinate (n4) at (1,0);
\coordinate (n5) at (1.5,0);
\coordinate (n6) at (1.25,0.5);
\coordinate (n7) at (2.25,0.05);
\coordinate (n8) at (2.25,0.6);
\coordinate (n9) at (3.25,0.375);
\coordinate (n10) at (4.25,0.375);
\coordinate (n11) at (-1.25,-0.5);
\coordinate (n12) at (-0.75,-0.5);
\coordinate (n13) at (5,-0.5);
\coordinate (n14) at (5.5,-0.5);

\begin{scope}[even odd rule]
\clip (-1.5, 1.5) rectangle (6,-2) (n1) circle (\c);
\clip (-1.5, 1.5) rectangle (6,-2) (n5) circle (\c);
\clip (-1.5, 1.5) rectangle (6,-2) (n7) circle (\c);
\draw (n12) .. controls +(90:0.5) and \control{(n1)}{(180:0.2)} -- (n5) .. controls +(0:0.2) and \control{(n7)}{(-120:0.2)};
\end{scope}
\begin{scope}[even odd rule]
\clip (-1.5, 1.5) rectangle (6,-2) (n4) circle (\c);
\clip (-1.5, 1.5) rectangle (6,-2) (n6) circle (\c);
\clip (-1.5, 1.5) rectangle (6,-2) (n7) circle (\c);
\draw (n7) .. controls +(60:0.3) and \control{(n8)}{(-60:0.3)} .. controls +(120:0.2) and \control{(n6)}{(60:0.1)} .. controls +(-120:0.1) and \control{(n4)}{(90:0.3)} .. controls + (-90:0.3) and \control{(n5)}{(-90:0.3)};
\end{scope}
\begin{scope}[even odd rule]
\clip (-1.5, 1.5) rectangle (6,-2) (n3) circle (\c);
\clip (-1.5, 1.5) rectangle (6,-2) (n2) circle (\c);
\draw (n5) .. controls +(90:0.3) and \control{(n6)}{(-60:0.1)} .. controls +(120:0.2) and \control{(n3)}{(60:0.2)} .. controls +(-120:0.1) and \control{(n1)}{(90:0.3)} .. controls +(-90:0.3) and \control{(n2)}{(-90:0.3)};
\end{scope}
\begin{scope}[even odd rule]
\clip (-1.5, 1.5) rectangle (6,-2) (n2) circle (\c);
\draw (n2) .. controls +(90:0.3) and \control{(n3)}{(-60:0.1)} .. controls +(120:0.35) and \control{(n11)}{(90:1)};
\end{scope}
\begin{scope}[even odd rule]
\clip (-1.5, 1.5) rectangle (6,-2) (n9) circle (\c);
\clip (-1.5, 1.5) rectangle (6,-2) (n8) circle (\c);
\draw (n11) .. controls +(-90:1.85) and \control{(n14)}{(-90:1.85)} .. controls +(90:1) and \control{(n10)}{(60:0.5)} .. controls +(-120:0.3) and \control{(n9)}{(-60:0.3)} .. controls +(120:0.3) and \control{(n8)}{(60:0.2)};
\end{scope}
\begin{scope}[even odd rule]
\clip (-1.5, 1.5) rectangle (6,-2) (n8) circle (\c);
\clip (-1.5, 1.5) rectangle (6,-2) (n10) circle (\c);
\draw (n8) .. controls +(-120:0.3) and \control{(n7)}{(120:0.3)} .. controls +(-60:0.2) and \control{(n9)}{(-120:0.3)} .. controls +(60:0.3) and \control{(n10)}{(120:0.3)} .. controls +(-60:0.5) and \control{(n13)}{(90:0.75)} .. controls +(-90:1.25) and \control{(n12)}{(-90:1.25)};
\end{scope}
\end{scope}

\node at (4.75,-0.5) {\huge $=$};

\begin{scope}[xshift = 4cm, yshift = -2cm, scale = 0.45]
\def\c{0.1}
\coordinate (n1) at (0,0);
\coordinate (n2) at (0.75,-0.75);
\coordinate (n3) at (1.5,0);
\coordinate (n4) at (0.75,0.75);
\coordinate (n5) at (1.5,1.5);
\coordinate (n6) at (2.75,-0.4);
\coordinate (n7) at (3.5,1.1);
\coordinate (n8) at (2.3,2.2);
\begin{scope}[even odd rule]
\clip (-1.5, 3) rectangle (4.5,-3) (n1) circle (\c);
\clip (-1.5, 3) rectangle (4.5,-3) (n6) circle (\c);
\clip (-1.5, 3) rectangle (4.5,-3) (n8) circle (\c);
\draw (n1) -- (n2) .. controls +(-45:1.25) and \control{(n6)}{(-100:1.25)} .. controls +(80:0.5) and \control{(n7)}{(-110:0.3)} .. controls +(70:1) and \control{(n8)}{(40:1.5)};
\end{scope}
\begin{scope}[even odd rule]
\clip (-1.5, 3) rectangle (4.5,-3) (n2) circle (\c);
\clip (-1.5, 3) rectangle (4.5,-3) (n4) circle (\c);
\clip (-1.5, 3) rectangle (4.5,-3) (n8) circle (\c);
\draw (n8) .. controls +(-140:0.5) and \control{(n5)}{(45:0.5)} -- (n1) .. controls +(-135:1) and \control{(n2)}{(-135:1)};
\end{scope}
\begin{scope}[even odd rule]
\clip (-1.5, 3) rectangle (4.5,-3) (n2) circle (\c);
\clip (-1.5, 3) rectangle (4.5,-3) (n5) circle (\c);
%\clip (-1.5, 3) rectangle (4.5,-3) (n3) circle (\c);
\draw (n2) -- (n3) .. controls +(45:0.5) and \control{(n5)}{(-45:0.5)};
\end{scope}
\begin{scope}[even odd rule]
\clip (-1.5, 3) rectangle (4.5,-3) (n2) circle (\c);
\clip (-1.5, 3) rectangle (4.5,-3) (n5) circle (\c);
\clip (-1.5, 3) rectangle (4.5,-3) (n3) circle (\c);
\draw (n5) .. controls +(135:0.75) and \control{(n4)}{(135:0.75)} -- (n3);
\end{scope}
\begin{scope}[even odd rule]
\clip (-1.5, 3) rectangle (4.5,-3) (n1) circle (\c);
\clip (-1.5, 3) rectangle (4.5,-3) (n7) circle (\c);
\clip (-1.5, 3) rectangle (4.5,-3) (n3) circle (\c);
\draw (n3) .. controls +(-45:0.25) and \control{(n6)}{(180:0.25)} .. controls +(0:1) and \control{(n7)}{(-45:1)} .. controls + (135:0.75) and \control{(n8)}{(-45:0.75)} .. controls +(135:1.75) and \control{(n1)}{(135:2)};
\end{scope}
\end{scope}
\end{tikzpicture}
\end{center}
\caption{The Hopf arborescent knot $8_{10}$ (from the Rolfsen table).}
\label{pic_knot_8_10}
\end{figure}

The genus-minimizing surface is the Hopf arborescent surface pictured in Figure~\ref{pic_frame_0_curve}, whose genus is $3$, and with $1$ boundary component. The red curve on the left of Figure~\ref{pic_frame_0_curve} is the core of an embedded annulus and its two boundaries are unlinked. Replacing this annulus by two discs glued on its two boundary components yields an immersed surface with genus $1$ less, still bounded by~$K_{8_{10}}$. The singularities of this immersion are of a particular type called ribbon singularities~\cite[p.225]{Rolfsen_Knots} (right side of Figure~\ref{pic_frame_0_curve}). Such singularities can be resolved in the $4$-ball, and thus this implies that the Hopf arborescent surface can be pushed into a smoothly embedded surface in the $4$-ball so that its smooth defect is at least $1$. Furthermore, for any link $L$, we have that $g_4^{smooth} (L) \geq g_4^{top} (L)$ so that its topological defect is also at least $1$.

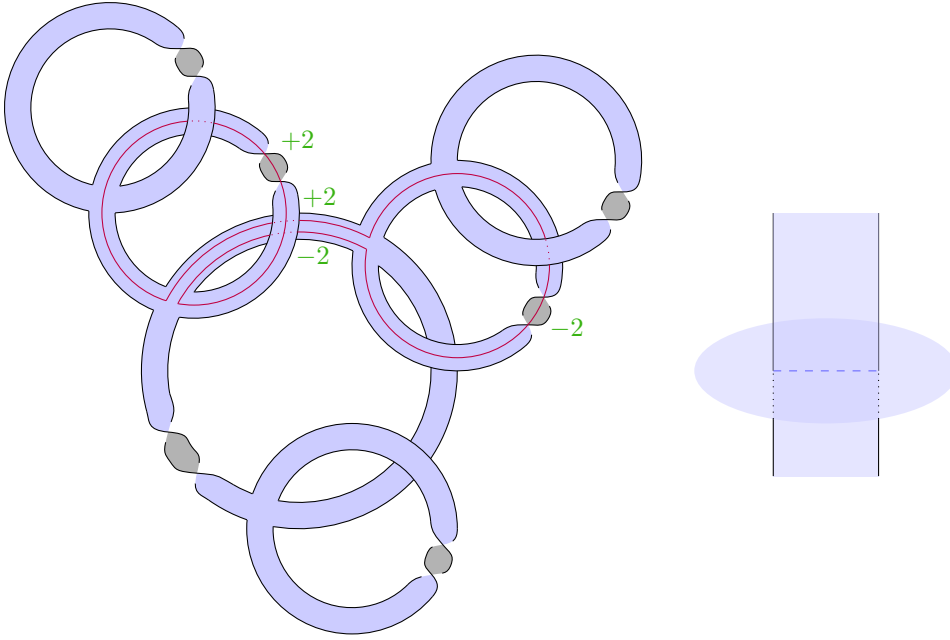
\begin{figure}[h]
\begin{center}
\begin{tikzpicture}[scale = 1.4]
\clip (-3,-2.75) rectangle (7,3.75);

%Large circle 1
\def\r{125} \def\xs{0} \def\ys{0}
\begin{scope}[rotate = \r]
\coordinate (l) at (80:1.375); 
\coordinate (h) at (100:1.375); 
\coordinate (i) at (90:1.25); 
\coordinate (e) at (90:1.5);
\fill [blue!20!white, xshift = \xs cm, yshift = \ys cm] (h) .. controls +(130:0.2) and \control{(110:1.5)}{(20:0.1)} arc (110:430:1.5) .. controls +(160:0.1) and \control{(l)}{(50:0.2)} .. controls +(-50:0.2) and \control{(65:1.25)}{(165:0.2)} arc (425:115:1.25) .. controls +(25:0.2) and \control{(h)}{(-130:0.2)};
\fill [white!40!gray] (h) .. controls +(-50:0.2) and \control{(i)}{(180:0.1)} .. controls +(0:0.1) and \control{(l)}{(-130:0.2)} .. controls +(130:0.2) and \control{(e)}{(0:0.1)} .. controls + (180:0.1) and \control{(h)}{(50:0.2)};
\end{scope}
\begin{scope}[even odd rule]
\clip [xshift =0.5cm, yshift = -1.5cm] (-4,-4) rectangle (4,4) (-135:0.75) arc (-135:-220:0.75) -- (-220:1) arc (-220:-135:1) --cycle;
\clip [xshift =-1cm, yshift = 1.5cm] (-4,-4) rectangle (4,4) (-80:0.75) arc (-80:-160:0.75) -- (-160:1) arc (-160:-80:1) --cycle;
\clip [xshift =1.5cm, yshift = 1cm] (-4,-4) rectangle (4,4) (80:0.75) arc (80:190:0.75) -- (190:1) arc (190:80:1) --cycle;
\draw [rotate = 125] (115:1.25) arc (115:425:1.25);
\draw [rotate = 125] (110:1.5) arc (110:430:1.5);
\end{scope}
\begin{scope}[rotate = \r]
\begin{scope}[even odd rule]
\clip (-2,0.5) rectangle (1,2) (l) circle (0.1cm);
\draw (115:1.25) .. controls +(25:0.2) and \control{(h)}{(-130:0.2)} .. controls +(50:0.2) and \control{(e)}{(180:0.1)} .. controls +(0:0.1) and \control{(l)}{(130:0.2)} .. controls +(-50:0.2) and \control{(65:1.25)}{(165:0.2)}; 
\end{scope}
\begin{scope}[even odd rule]
\clip (-1,0.5) rectangle (1,2) (h) circle (0.1cm);
\draw (110:1.5) .. controls +(20:0.1) and \control{(h)}{(130:0.2)} .. controls +(-50:0.2) and \control{(i)}{(180:0.1)} .. controls +(0:0.1) and \control{(l)}{(-130:0.2)} .. controls +(50:0.2) and \control{(70:1.5)}{(160:0.1)}; 
\end{scope}
\end{scope}

%Circle 2
\def\r{30} \def\xs{-1} \def\ys{1.5}
\begin{scope}[xshift = \xs cm, yshift = \ys cm, rotate = \r]
\coordinate (l) at (-10:0.875); 
\coordinate (h) at (10:0.875); 
\coordinate (i) at (0:0.75); 
\coordinate (e) at (0:1);
\fill [blue!20!white] (h) .. controls +(30:0.1) and \control{(20:1)}{(-70:0.1)} arc (20:340:1) .. controls +(70:0.1) and \control{(l)}{(-30:0.1)}  .. controls +(-150:0.1) and \control{(-25:0.75)}{(65:0.1)} arc (335:25:0.75) .. controls +(-65:0.1) and \control{(h)}{(150:0.1)};
\fill [white!40!gray] (h).. controls +(-30:0.1) and \control{(e)}{(90:0.1)} .. controls +(-90:0.1) and \control{(l)}{(30:0.1)} .. controls + (150:0.1) and \control{(i)}{(-90:0.1)} .. controls +(90:0.1) and \control{(h)}{(-150:0.1)};
\end{scope}
\begin{scope}[even odd rule]
\clip (-3,-3) rectangle (3,3) (110:1.25) arc (110:180:1.25) -- (180:1.5) arc (180:110:1.5) --cycle;
\clip (-3,-3) rectangle (3,6) [xshift = -1.8cm, yshift = 2.5cm] (-60:0.75) arc (-60:-120:0.75) -- (-120:1) arc (-120:-60:1) --cycle;
\draw [xshift = \xs cm, yshift = \ys cm, rotate = \r] (25:0.75) arc (25:335:0.75);
\draw [xshift = \xs cm, yshift = \ys cm, rotate = \r] (20:1) arc (20:340:1);
\end{scope}
\begin{scope}[xshift = \xs cm, yshift = \ys cm, rotate = \r]
\begin{scope}[even odd rule]
\clip (-1.5,-1) rectangle (1.5,1) (l) circle (0.1cm);
\draw (25:0.75) .. controls +(-65:0.1) and \control{(h)}{(150:0.1)} .. controls +(-30:0.1) and \control{(e)}{(90:0.1)} .. controls +(-90:0.1) and \control{(l)}{(30:0.1)} .. controls +(-150:0.1) and \control{(-25:0.75)}{(65:0.1)}; 
\end{scope}
\begin{scope}[even odd rule]
\clip (-1.5,-1) rectangle (1.5,1)  (h) circle (0.1cm);
\draw (20:1) .. controls +(-70:0.1) and \control{(h)}{(30:0.1)} .. controls +(-150:0.1) and \control{(i)}{(90:0.1)} .. controls +(-90:0.1) and \control{(l)}{(150:0.1)} .. controls +(-30:0.1) and \control{(-20:1)}{(70:0.1)}; 
\end{scope}
\end{scope}

%Circle 3
\def\r{30} \def\xs{-1.8} \def\ys{2.5}
\begin{scope}[xshift = \xs cm, yshift = \ys cm, rotate = \r]
\coordinate (l) at (-10:0.875); 
\coordinate (h) at (10:0.875); 
\coordinate (i) at (0:0.75); 
\coordinate (e) at (0:1);
\fill [blue!20!white] (h) .. controls +(30:0.1) and \control{(20:1)}{(-70:0.1)} arc (20:340:1) .. controls +(70:0.1) and \control{(l)}{(-30:0.1)}  .. controls +(-150:0.1) and \control{(-25:0.75)}{(65:0.1)} arc (335:25:0.75) .. controls +(-65:0.1) and \control{(h)}{(150:0.1)};
\fill [white!40!gray] (h).. controls +(-30:0.1) and \control{(e)}{(90:0.1)} .. controls +(-90:0.1) and \control{(l)}{(30:0.1)} .. controls + (150:0.1) and \control{(i)}{(-90:0.1)} .. controls +(90:0.1) and \control{(h)}{(-150:0.1)};
\end{scope}
\begin{scope}[even odd rule]
\clip (-3,-3) rectangle (3,6) [xshift = -1cm, yshift = 1.5cm] (110:0.75) arc (110:190:0.75) -- (190:1) arc (190:110:1) --cycle;
\draw [xshift = \xs cm, yshift = \ys cm, rotate = \r] (25:0.75) arc (25:335:0.75);
\draw [xshift = \xs cm, yshift = \ys cm, rotate = \r] (20:1) arc (20:340:1);
\end{scope}
\begin{scope}[xshift = \xs cm, yshift = \ys cm, rotate = \r]
\begin{scope}[even odd rule]
\clip (-1.5,-1) rectangle (1.5,1) (l) circle (0.1cm);
\draw (25:0.75) .. controls +(-65:0.1) and \control{(h)}{(150:0.1)} .. controls +(-30:0.1) and \control{(e)}{(90:0.1)} .. controls +(-90:0.1) and \control{(l)}{(30:0.1)} .. controls +(-150:0.1) and \control{(-25:0.75)}{(65:0.1)}; 
\end{scope}
\begin{scope}[even odd rule]
\clip (-1.5,-1) rectangle (1.5,1)  (h) circle (0.1cm);
\draw (20:1) .. controls +(-70:0.1) and \control{(h)}{(30:0.1)} .. controls +(-150:0.1) and \control{(i)}{(90:0.1)} .. controls +(-90:0.1) and \control{(l)}{(150:0.1)} .. controls +(-30:0.1) and \control{(-20:1)}{(70:0.1)}; 
\end{scope}
\end{scope}
%Curve
\begin{scope}[even odd rule]
\clip (-3,-3) rectangle (3,6) [xshift = \xs cm, yshift = \ys cm] (10:0.75) arc (10:-60:0.75) -- (-60:1) arc (-60:10:1) --cycle;
\clip (-3,-3) rectangle (3,6)  (180:1.32) arc (180:120:1.32) -- (120:1.43) arc (120:180:1.43) --cycle;
\draw [purple, xshift = -1 cm, yshift = 1.5 cm] (0,0) circle (0.875);
\end{scope}
\begin{scope}
\clip [xshift = \xs cm, yshift = \ys cm] (10:0.75) arc (10:-60:0.75) -- (-60:1) arc (-60:10:1) --cycle;
\draw [purple, dotted, xshift = -1 cm, yshift = 1.5 cm] (0,0) circle (0.875);
\end{scope}

%Circle 4
\def\r{-30} \def\xs{1.5} \def\ys{1}
\begin{scope}[xshift = \xs cm, yshift = \ys cm, rotate = \r]
\coordinate (l) at (-10:0.875); 
\coordinate (h) at (10:0.875); 
\coordinate (i) at (0:0.75); 
\coordinate (e) at (0:1);
\fill [blue!20!white] (h) .. controls +(30:0.1) and \control{(20:1)}{(-70:0.1)} arc (20:340:1) .. controls +(70:0.1) and \control{(l)}{(-30:0.1)}  .. controls +(-150:0.1) and \control{(-25:0.75)}{(65:0.1)} arc (335:25:0.75) .. controls +(-65:0.1) and \control{(h)}{(150:0.1)};
\fill [white!40!gray] (h).. controls +(-30:0.1) and \control{(e)}{(90:0.1)} .. controls +(-90:0.1) and \control{(l)}{(30:0.1)} .. controls + (150:0.1) and \control{(i)}{(-90:0.1)} .. controls +(90:0.1) and \control{(h)}{(-150:0.1)};
\end{scope}
\begin{scope}[even odd rule]
\clip (-3,-3) rectangle (3,3) (90:1.25) arc (90:30:1.25) -- (30:1.5) arc (30:90:1.5) --cycle;
\clip (-3,-3) rectangle (6,6) [xshift = 2.25cm, yshift = 2cm] (-120:0.75) arc (-120:-240:0.75) -- (-240:1) arc (-240:-120:1) --cycle;
\draw [xshift = \xs cm, yshift = \ys cm, rotate = \r] (25:0.75) arc (25:335:0.75);
\draw [xshift = \xs cm, yshift = \ys cm, rotate = \r] (20:1) arc (20:340:1);
\end{scope}
\begin{scope}[xshift = \xs cm, yshift = \ys cm, rotate = \r]
\begin{scope}[even odd rule]
\clip (-1.5,-1) rectangle (1.5,1) (h) circle (0.1cm);
\draw (25:0.75) .. controls +(-65:0.1) and \control{(h)}{(150:0.1)} .. controls +(-30:0.1) and \control{(e)}{(90:0.1)} .. controls +(-90:0.1) and \control{(l)}{(30:0.1)} .. controls +(-150:0.1) and \control{(-25:0.75)}{(65:0.1)}; 
\end{scope}
\begin{scope}[even odd rule]
\clip (-1.5,-1) rectangle (1.5,1)  (l) circle (0.1cm);
\draw (20:1) .. controls +(-70:0.1) and \control{(h)}{(30:0.1)} .. controls +(-150:0.1) and \control{(i)}{(90:0.1)} .. controls +(-90:0.1) and \control{(l)}{(150:0.1)} .. controls +(-30:0.1) and \control{(-20:1)}{(70:0.1)}; 
\end{scope}
\end{scope}

%Circle 5
\def\r{-30} \def\xs{2.25} \def\ys{2}
\begin{scope}[xshift = \xs cm, yshift = \ys cm, rotate = \r]
\coordinate (l) at (-10:0.875); 
\coordinate (h) at (10:0.875); 
\coordinate (i) at (0:0.75); 
\coordinate (e) at (0:1);
\fill [blue!20!white] (h) .. controls +(30:0.1) and \control{(20:1)}{(-70:0.1)} arc (20:340:1) .. controls +(70:0.1) and \control{(l)}{(-30:0.1)}  .. controls +(-150:0.1) and \control{(-25:0.75)}{(65:0.1)} arc (335:25:0.75) .. controls +(-65:0.1) and \control{(h)}{(150:0.1)};
\fill [white!40!gray] (h).. controls +(-30:0.1) and \control{(e)}{(90:0.1)} .. controls +(-90:0.1) and \control{(l)}{(30:0.1)} .. controls + (150:0.1) and \control{(i)}{(-90:0.1)} .. controls +(90:0.1) and \control{(h)}{(-150:0.1)};
\end{scope}
\begin{scope}[even odd rule]
\clip (-3,-3) rectangle (6,6) [xshift = 1.5cm, yshift = 1cm] (60:0.75) arc (60:120:0.75) -- (120:1) arc (120:60:1) --cycle;
\draw [xshift = \xs cm, yshift = \ys cm, rotate = \r] (25:0.75) arc (25:335:0.75);
\draw [xshift = \xs cm, yshift = \ys cm, rotate = \r] (20:1) arc (20:340:1);
\end{scope}
\begin{scope}[xshift = \xs cm, yshift = \ys cm, rotate = \r]
\begin{scope}[even odd rule]
\clip (-1.5,-1) rectangle (1.5,1) (h) circle (0.1cm);
\draw (25:0.75) .. controls +(-65:0.1) and \control{(h)}{(150:0.1)} .. controls +(-30:0.1) and \control{(e)}{(90:0.1)} .. controls +(-90:0.1) and \control{(l)}{(30:0.1)} .. controls +(-150:0.1) and \control{(-25:0.75)}{(65:0.1)}; 
\end{scope}
\begin{scope}[even odd rule]
\clip (-1.5,-1) rectangle (1.5,1)  (l) circle (0.1cm);
\draw (20:1) .. controls +(-70:0.1) and \control{(h)}{(30:0.1)} .. controls +(-150:0.1) and \control{(i)}{(90:0.1)} .. controls +(-90:0.1) and \control{(l)}{(150:0.1)} .. controls +(-30:0.1) and \control{(-20:1)}{(70:0.1)}; 
\end{scope}
\end{scope}
%Curve
\begin{scope}[even odd rule]
\clip (-3,-3) rectangle (3,6) [xshift = \xs cm, yshift = \ys cm] (-60:0.75) arc (-60:-120:0.75) -- (-120:1) arc (-120:-60:1) --cycle;
\clip (-3,-3) rectangle (3,6)  (80:1.32) arc (80:50:1.32) -- (50:1.43) arc (50:80:1.43) --cycle;
\draw [purple, xshift = 1.5 cm, yshift = 1 cm] (0,0) circle (0.875);
\end{scope}
\begin{scope}
\clip [xshift = \xs cm, yshift = \ys cm] (-60:0.75) arc (-60:-120:0.75) -- (-120:1) arc (-120:-60:1) --cycle;
\draw [purple, dotted, xshift = 1.5 cm, yshift = 1 cm] (0,0) circle (0.875);
\end{scope}
\begin{scope}[even odd rule]
\clip (-3,-3) rectangle (3,6) [xshift = -1cm, yshift=1.5cm] (-70:0.75) arc (-70:10:0.75) -- (10:1) arc (10:-70:1);
\clip ($(-1,1.5)+(180:0.875)$) arc (180:270:0.875) -- ($(1.5,1)+(-150:0.875)$) arc (-150:-230:0.875) --cycle;
\draw [purple] (180:1.32) arc (180:60:1.32) -- (60:1.43) arc (60:180:1.43) --cycle;
\end{scope}
\begin{scope}
\clip [xshift = -1cm, yshift=1.5cm] (-70:0.75) arc (-70:10:0.75) -- (10:1) arc (10:-70:1);
\draw [purple, dotted] (180:1.32) arc (180:60:1.32) -- (60:1.43) arc (60:180:1.43) --cycle;
\end{scope}

%Circle 6
\def\r{-20} \def\xs{0.5} \def\ys{-1.5}
\begin{scope}[xshift = \xs cm, yshift = \ys cm, rotate = \r]
\coordinate (l) at (-10:0.875); 
\coordinate (h) at (10:0.875); 
\coordinate (i) at (0:0.75); 
\coordinate (e) at (0:1);
\fill [blue!20!white] (h) .. controls +(30:0.1) and \control{(20:1)}{(-70:0.1)} arc (20:340:1) .. controls +(70:0.1) and \control{(l)}{(-30:0.1)}  .. controls +(-150:0.1) and \control{(-25:0.75)}{(65:0.1)} arc (335:25:0.75) .. controls +(-65:0.1) and \control{(h)}{(150:0.1)};
\fill [white!40!gray] (h).. controls +(-30:0.1) and \control{(e)}{(90:0.1)} .. controls +(-90:0.1) and \control{(l)}{(30:0.1)} .. controls + (150:0.1) and \control{(i)}{(-90:0.1)} .. controls +(90:0.1) and \control{(h)}{(-150:0.1)};
\end{scope}
\begin{scope}[even odd rule]
\clip (-3,-3) rectangle (3,3) (-90:1.25) arc (-90:-125:1.25) -- (-125:1.5) arc (-125:-90:1.5) --cycle;
\draw [xshift = \xs cm, yshift = \ys cm, rotate = \r] (25:0.75) arc (25:335:0.75);
\draw [xshift = \xs cm, yshift = \ys cm, rotate = \r] (20:1) arc (20:340:1);
\end{scope}
\begin{scope}[xshift = \xs cm, yshift = \ys cm, rotate = \r]
\begin{scope}[even odd rule]
\clip (-1.5,-1) rectangle (1.5,1) (l) circle (0.1cm);
\draw (25:0.75) .. controls +(-65:0.1) and \control{(h)}{(150:0.1)} .. controls +(-30:0.1) and \control{(e)}{(90:0.1)} .. controls +(-90:0.1) and \control{(l)}{(30:0.1)} .. controls +(-150:0.1) and \control{(-25:0.75)}{(65:0.1)}; 
\end{scope}
\begin{scope}[even odd rule]
\clip (-1.5,-1) rectangle (1.5,1)  (h) circle (0.1cm);
\draw (20:1) .. controls +(-70:0.1) and \control{(h)}{(30:0.1)} .. controls +(-150:0.1) and \control{(i)}{(90:0.1)} .. controls +(-90:0.1) and \control{(l)}{(150:0.1)} .. controls +(-30:0.1) and \control{(-20:1)}{(70:0.1)}; 
\end{scope}
\end{scope}

\node [green!70!purple] at ($(1.5,1)+(-30:1.2)$) {$-2$};
\node [green!70!purple] at ($(-1,1.5)+(35:1.2)$) {$+2$};
\node [green!70!purple] at ($(-1,1.5)+(7:1.2)$) {$+2$};
\node [green!70!purple] at ($(-1,1.5)+(-20:1.2)$) {$-2$};

\begin{scope}[xshift = 5cm]
\fill [blue!20!white, opacity = 0.5] (0,0) circle (1.25cm and 0.5cm);
\draw [dashed, blue] (-0.5, 0) -- (0.5,0);
\draw (-0.5,0) -- (-0.5,1.5);
\draw (0.5,0) -- (0.5,1.5);
\fill [blue!20!white, opacity = 0.5] (-0.5,-1) rectangle (0.5,1.5);
\begin{scope}[even odd rule]
\clip (-2,-1.5) rectangle (2,1.5) (0,0) circle (1.25cm and 0.5cm);
\draw (-0.5,-1) -- (-0.5,0);
\draw (0.5,-1) -- (0.5,0);
\end{scope}
\begin{scope}
\clip (0,0) circle (1.25cm and 0.5cm);
\draw [dotted] (-0.5,-1) -- (-0.5,0);
\draw [dotted] (0.5,-1) -- (0.5,0);
\end{scope}
\end{scope}
\end{tikzpicture}
\caption{Left: A simple red curve on the surface associated to $\Sigma (8_{10})$. The intersection number in green show that the boundaries of the corresponding annulus are unlinked. Right : A ribbon intersection.}
\label{pic_frame_0_curve}
\end{center}
\end{figure}

One can also check in the censuses that these two examples are hyperbolic.

\paragraph*{Arbitrarily large defect.}

We can use the previous examples to create Hopf arborescent links with arbitrarily large (smooth or topological) defect. 
The following proposition proves that the defect coming from tree patterns add up if they are isolated enough.

\begin{proposition}\label{prop_isolated_pattern}
Let $L$ be an arborescent link associated to a plane tree $T$. Let $T_1, \ldots, T_n$ be disjoint subtrees of $T$ that are pairwise at distance at least two within $T$. Denote $L_i =\partial (\Sigma (T_i))$, then $\oper{\sum}{1 \leq i \leq n}{}{\Delta_g (L_i)}  \leq \Delta_g (L)$.
\end{proposition}

\begin{proof}
For each $i$, identifying the Hopf bands of the vertices of $T_i$ within $\Sigma (T)$ provides a family of embeddings of the $\Sigma (T_i)$. These embeddings are disjoint by the distance assumption. 

We can now repeat the proof of Lemma~\ref{lem_defect} and Proposition~\ref{prop_stable_defect} with a family of links to obtain the desired inequality. Since the $\Sigma(T_i)$ are disjoint subsurfaces of $\Sigma(L)$,  $\bigcup \Sigma (T_i)$ is a surface-minor of $\Sigma (L)$. As in the proof of Lemma~\ref{lem_defect}, we define for each $L_i$, a surface $S'_i$ in $\B^4$ that has $L_i$ as its boundary. Gluing each $S'_i$ to the pieces of $\Sigma (T) \smallsetminus \bigcup \Sigma (T_i)$ along its associated $L_i$ yields a surface $S$ in $\B^4$ such that $\partial S = L$. 

Now we have $g(S) = g(\Sigma (T)) - \sum_i g(\Sigma(T_i)) + \sum_i g(S'_i)$. It follows, by definition of $g_4$ that $g(\Sigma (T)) - g_4 (L) \geq g(\Sigma (T)) - g(S) = \sum_i g(\Sigma (T_i)) - g(S'_i)$. By Theorem~\ref{th_fibred_surface_minimise}, and taking each $S'_i$ to be of minimal genus, $\Delta_g (L) = g(L) - g_4(L) \geq \sum_i g(L_i) - g_4(L_i) = \sum_i \Delta_g (L_i)$.
\end{proof}

Consider a Hopf arborescent link $L$ associated to some tree $T$ such that $L$ has non zero defect. By identifying the roots of $n$ copies of $T$ with the leaves of a $n$-star graph (one vertex of degree $n$, $n$ vertices of degree $1$), we create a tree $T'$ such that the link $L' = \partial \Sigma (T')$ has defect at least $n\Delta_g (L)$ by Proposition~\ref{prop_isolated_pattern}, therefore obtaining links with arbitrarily large topological and smooth defect.

\subparagraph*{Acknowledgements.} We would like to thank Sebastian Baader for helpful discussions, and the anonymous reviewers for their questions and suggestions which allowed us to significantly improve the paper.

\bibliographystyle{plainurl}
\bibliography{biblio}

\end{document}